\documentclass{theoretics}
\usepackage{color}
\usepackage[english]{babel}
\usepackage{amsfonts}
\usepackage{hyphenat}
\usepackage[capitalise,noabbrev,nameinlink]{cleveref}
\graphicspath{{./figures/},{./figures/}}
\usepackage{tikz-cd}
\usepackage{stmaryrd}
\usepackage{subcaption}
\usepackage{xspace}
\usepackage{thm-restate}
\usepackage{url}
\usepackage{soul}

\usepackage{ogonek}

\newcommand{\Z}{\ensuremath{\mathbb{Z}}\xspace}
\newcommand{\R}{\ensuremath{\mathbb{R}}\xspace}
\newcommand{\Q}{\ensuremath{\mathbb{Q}}\xspace}
\newcommand{\N}{\mathbb{N}}

\newcommand{\realRAM}{\textnormal{real RAM}\xspace}
\newcommand{\wordRAM}{\textnormal{word RAM}\xspace}

\renewcommand{\phi}{\varphi}
\newcommand{\eps}{\varepsilon}

\newcommand{\ER}{\ensuremath{\exists \mathbb{R}}\xspace}
\newcommand{\NP}{\ensuremath{\textrm{NP}}\xspace}
\newcommand{\PSPACE}{\ensuremath{\textrm{PSPACE}}\xspace}

\newcommand{\set}[2]{\ensuremath{\left\{ #1 \, \middle| \, #2 \right\}}}
\renewcommand{\mod}[1]{\left(\mathrm{mod} \; #1\right)}

\newcommand{\expl}{\mathrm{expl}\xspace}
\newcommand{\sign}{\textrm{sign}\xspace}
\newcommand{\wellbehaved}{\mbox{well-behaved}\xspace}
\newcommand{\wellbehavedness}{\mbox{well-behavedness}\xspace}
\newcommand{\Wellbehaved}{\textrm{Well-behaved}\xspace}
\newcommand{\equalzero}[1]{\ensuremath{\mathtt{EqualZero}(#1)}\xspace}
\newcommand{\largerzero}[1]{\ensuremath{\mathtt{LargerZero}(#1)}\xspace}
\newcommand{\largerzeronoarg}{\ensuremath{\mathtt{LargerZero}}\xspace}
\newcommand{\convex}{\ensuremath{\textrm{convexly curved}}\xspace}

\newcommand{\concave}{\ensuremath{\textrm{concavely curved}}\xspace}

\newcommand{\nonlinear}{\mbox{curved}\xspace}
\newcommand{\Nonlinear}{\ensuremath{\textrm{Curved}}\xspace}
\newcommand{\triplealgebraic}{\ensuremath{\textrm{triple algebraic}}\xspace}
\newcommand{\nicelyComputable}{\mbox{nicely computable}\xspace}
\newcommand{\equisatisfiable}{\mbox{equisatisfiable}\xspace}

\newcommand{\var}[1]{\left\llbracket#1\right\rrbracket}

\newcommand{\domain}{\ensuremath{U}\xspace}
\newcommand{\domainadherent}{domain adherent\xspace}

\newcommand{\constraintFormula}{constraint formula\xspace}

\newcommand{\maxIndPlanar}{MISPLANAR\xspace}

\newcommand{\continuousConstraintSatisfactionProblem}{continuous constraint satisfaction problem\xspace}

\newcommand{\ConstraintSatisfactionProblem}{Constraint satisfaction problem\xspace}
\newcommand{\constraintSatisfactionProblem}{constraint satisfaction problem\xspace}
\newcommand{\CCSP}{CCSP\xspace}

\newcommand{\CE}{CE\xspace}
\newcommand{\CEEXPL}{CE-EXPL\xspace}
\newcommand{\etr}{ETR\xspace}
\newcommand{\etrinv}{ETR-INV\xspace}
\newcommand{\CCIEXPL}{CCI-EXPL\xspace}
\newcommand{\CCIEXPLstar}{\ensuremath{\textrm{\CCIEXPL}^*}\xspace}
\newcommand{\CCI}{CCI\xspace}
\newcommand{\CCIstar}{\ensuremath{\textrm{\CCI}^*}\xspace}

\newcommand{\etrconjunction}{\ensuremath{\textrm{ETR-CONJ}}\xspace}
\newcommand{\etrsmall}{\ensuremath{\textrm{ETR-SMALL}}\xspace}

\newcommand{\etrsquare}{\ensuremath{\textrm{ETR-SQUARE}}\xspace}

\newcommand{\etrami}{\ensuremath{\textrm{ETR-AMI}}\xspace}
\newcommand{\etrcompact}{\ensuremath{\textrm{ETR-COMPACT}}\xspace}

\ThCSnewtheoita{hypothesis}

\crefname{ineq}{Inequality}{Inequalities}
\crefname{cond}{Condition}{Conditions}
\creflabelformat{cond}{(#2#1#3)}
\crefname{sidefigure}{Figure}{Figures}

\newcommand{\pack}[3]{}

\title{On Classifying Continuous Constraint Satisfaction Problems}

\ThCSauthor{Tillmann Miltzow}{t.miltzow@uu.nl}[https://orcid.org/0000-0003-4563-2864]
\ThCSaffil{Computer Science Department, Utrecht University}

\ThCSauthor{Reinier F. Schmiermann}{r.f.schmiermann@uu.nl}[https://orcid.org/0009-0001-3810-3676]
\ThCSaffil{Mathematics Department, Utrecht University}

\ThCSshortnames{T. Miltzow and R.F. Schmiermann}
\ThCSshorttitle{On Classifying Continuous Constraint Satisfaction Problems}

\ThCSkeywords{Constraint Satisfaction Problems, Existential Theory of
  the Reals, Computational Complexity, Classification}

\ThCSthanks{A preliminary version of this article appeared at FOCS 2021. We thank Eunjung Kim and Saoirs\'{e} Freeney for their helpful comments.
  We thank Mikkel Abrahamsen for letting us 
  use parts of~\cite{abrahamsen2019dynamic}.
  The first author is generously supported by the NWO Veni grant 016.Veni.192.250.
  The second author is generously supported by the NWO Vidi grant  VI.Vidi.192.012.
  We would like to thank anonymous reviewers 
  for their incredibly valuable feedback.}

\ThCSyear{2024}
\ThCSarticlenum{10}
\ThCSdoicreatedtrue
\ThCSreceived{Mar 7, 2022}
\ThCSrevised{Nov 3, 2023}
\ThCSaccepted{Feb 2, 2024}
\ThCSpublished{Apr 16, 2024} % if submitted after Friday 20:00 CEST

\begin{document}

\maketitle

\begin{abstract}
    A continuous constraint satisfaction problem (CCSP) is a 
    constraint satisfaction problem (CSP) 
    with the real numbers as domain.
    We engage in a systematic study to classify CCSPs that are complete for  the
    Existential Theory of the Reals, i.e., \ER-complete. 
    To define this class, we first consider the problem ETR, which also 
    stands for Existential Theory of the Reals. 
    In an instance of this problem we are given a sentence of the form
    $    \exists x_1, \ldots, x_n \in \R : \Phi(x_1, \ldots, x_n)$,
    where~$\Phi$ is a well-formed quantifier-free formula consisting of the symbols $\{0, 1, x_1, \ldots, x_n, +, \cdot, \geq, >, \wedge, \vee, \neg\}$.
    The goal is to check whether this sentence is true.
    Now the class \ER is the family of all problems that admit a polynomial-time many-one reduction to ETR.
    It is known that $\NP \subseteq \ER \subseteq \PSPACE$.

    We restrict our attention on CCSPs with addition constraints ($x+y=z$)
    and which satisfy another mild technical condition.
    Previously, it was shown that multiplication constraints ($x\cdot y = z$), squaring constraints ($x^2 = y$),
    or inversion constraints ($x\cdot y = 1$) are sufficient to establish \ER-completeness.
    We extend this in the strongest possible sense for equality constraints as follows.
    We show that {CCSPs} (with addition constraints and that satisfy another mild technical condition) that have \textit{any} one \wellbehaved \nonlinear
    equality constraint ($f(x,y) = 0$) are \ER-complete.
    We further extend our results to inequality constraints.
    We show that together \textit{any} \wellbehaved \convex and \textit{any} \wellbehaved \concave inequality constraint
    ($f(x,y) \geq 0$ and $g(x,y) \geq 0$)
    imply \ER-completeness on the class of such CCSPs.
    
    Here, we call a function $f \colon \domain \to \R$ \wellbehaved if it is a $C^3$-function, $f(0, 0) = 0$, all its first and second partial derivatives $f_x$, $f_y$, $f_{xx}$, $f_{xy}$, $f_{yy}$ are rational in $(0, 0)$, $f_x(0, 0) \neq 0$ or $f_y(0, 0) \neq 0$, and another mild technical constraint. 
    Furthermore, we call $f$ \nonlinear if the curvature of the curve given by $f(x, y) = 0$ is nonzero at the origin. 
    In this case we call $f$ either \convex if the curvature is negative, or \concave if it is positive.
\end{abstract}

\begin{sidefigure}[htbp]
    \centering
    \includegraphics{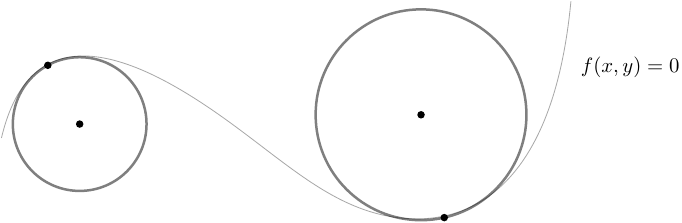}
    \caption{The set defined by $f(x,y) = 0$ is \wellbehaved and at some positions \convex and at others \concave, indicated by the two circles.}
    \label{fig:kappa}
\end{sidefigure}

%\newpage TS

%%%%%%%%%%%%%%%%%%%%%%%%%%%%%%%%%%
\section{Introduction}
\label{sec:intro}
%%%%%%%%%%%%%%%%%%%%%%%%%%%%%%%%%%
In geometric packing, we are given a set of two-dimensional pieces, a container and a set of motions.
The aim is to move the pieces into the container without overlap, and while respecting the given motions.
Recently, Abrahamsen, Miltzow and Seiferth showed that many geometric packing variants are \ER-complete (FOCS 2020)~\cite{etrPacking}.
Despite the fact that the first arXiv version is roughly 100 pages long, the high-level
approach follows the same principle as many other hardness reductions.
First, they showed that a technical intermediate problem is hard and then they reduced from this technical problem.
In their work, \etrinv, a specific continuous constraint satisfaction problem, serves as this intermediate \ER-complete problem. A complete definition of continuous constraint satisfaction problems is provided in the subsequent section.
Specifically, \etrinv contains essentially only addition constraints ($x+y = z$) and inversion constraints ($x\cdot y = 1$).
In the second step, they showed how to encode addition and inversion using geometric objects.
This enabled them to show in a unified framework that various geometric packing
problems are \ER-complete.

The inversion constraint is particularly handy as it was shown in various other 
works that it is particularly easy to encode geometrically~\cite{AnnaPreparation, AreasKleist, NestedPolytopesER, ARTETR}.
Curiously, Abrahamsen, Miltzow and Seiferth left arguably the most interesting case
of packing convex polygonal objects into a square container open.
The missing puzzle piece seemed to be a gadget to encode the inversion constraint for this case.

We take an alternative approach and engage in a systematic study of continuous constraint
satisfaction problems in their own respect.
The aim is to fully classify all continuous constraint satisfaction problems by their
computational complexity.
Polynomial time, \NP-complete, and \ER-complete are some apparent complexities, but as we will
see, they may not be the only ones that are relevant, see \Cref{sub:discussion}.
Our first application shows that packing convex polygons into a square under rigid motions is \ER-complete.
It arises as a combination of a small adaption of the framework by~\cite{etrPacking} and our structural results.
As a result the paper by Abrahamsen, Miltzow and Seiferth considerably shortens to about 70 pages.

\begin{remark}
    Although the \ER-completeness of packing convex polygons into a square container under rigid motions was first pointed out in the conference version of this paper, the proof is more readable in the context of the paper by Abrahamsen, Miltzow and Seiferth~\cite{etrPacking}.
    The arxiv version of their paper incorporated the results presented in the conference version of this paper.
    In turn, some of the technical results of their paper are incorporated in this paper as \Cref{sec:etr-square} for the sake of readability and completeness.
\end{remark}

We give a short introduction 
to constraint satisfaction problems and the complexity class~\ER{}.

\subsection{Constraint Satisfaction Problems}
Constraint satisfaction problems (CSPs) are a wide class of computational decision problems.
In order to give a formal definition, we first introduce several other terms.
\begin{definition}[Signature]
    A \textit{signature} is a set of symbols
    together with \textit{arities} $\ell \in \N$.
     Each symbol has exactly one arity attached to it. 
\end{definition}
Often the signature distinguishes between function symbols and 
    relation symbols. 
    We will only use relation symbols.
 We will only use signatures of finite size, to avoid dealing with issues of description complexity. 
 For finite signatures, we can simply assume that each symbol has constant description complexity.

\begin{definition}[Structure]
    A structure consists of a set~\domain, called the \textit{domain}, 
    a signature~$\tau$
    and an \textit{interpretation} of each symbol.
    If $\alpha \in \tau$ is a symbol of arity~$\ell$, then 
    the interpretation is a set $\alpha \subseteq \domain^{\ell}$.
\end{definition}
In the literature, the term \emph{template} is also used as a synonym for structure.
To make this more tangible, consider the following example.
We define the domain $\domain = \{0,1\}$, the symbol~$+_2$ of arity~$3$
and the symbol~$\mathbf{1}$ of arity~$1$. 
We interpret $+_2$ as $\set{(x,y,z)\in U^3}{x+y\equiv z \mod 2}$,
and $\mathbf{1}$ as $\set{x\in \domain}{x = 1}$.
This defines a structure $S_1 = \langle \domain  , +_2, \mathbf{1} \rangle $.
Note that it is common to use a symbol and its interpretation interchangeably.
Specifically, many symbols are used in the literature with their common interpretation,
e.g., $\leq$ is interpreted as $\set{(x,y) \in \domain}{x \leq y}$ and 
$+$ is interpreted as $\set{(x,y,z) \in \domain^3}{x + y = z}$. 
We refer to the symbols and interpretations of a structure merely as \textit{constraints}. 
We will usually denote these constraints by the equation that they enforce.
For example, we write $x^2 = y$ for the constraint $c = \set{ (x,y)\in \domain^2}{x^2 = y }$.
\begin{definition}[\ConstraintSatisfactionProblem]
    Given a structure $S= (\domain,\tau)$ we define a \textit{\constraintFormula}
    $\Phi := \Phi(x_1, \ldots, x_n)$ to be a conjunction 
    $c_1 \wedge \ldots \wedge c_m$
    for $m \geq 0$, where each~$c_i$ is of the form $c(y_1, \ldots, y_\ell)$ for some 
    $c \in \tau$ and variables $y_1, \ldots, y_\ell \in \{x_1, \ldots, x_n\}$. 
    We also define $V(\Phi) \subseteq \domain^n$ as
    $V(\Phi) := \set{\mathbf{x} \in \domain^n}{\Phi(\mathbf{x})}$.
    In the \textit{\constraintSatisfactionProblem} (CSP) with structure~$S$, we are given a \constraintFormula{} $\Phi$, and are asked whether $V(\Phi) \neq \emptyset$.
\end{definition}
Consider the \constraintFormula  $\Phi = (x_1 + x_2 \equiv x_4 \mod 2) \land (x_2 + x_3 \equiv x_4 \mod 2) \land (x_2 = 1)$. This
gives an instance of a CSP with structure~$S_1$ as above.
Note that 
% $(x_1,x_2,x_3,x_4) = 
$(0,1,0,1) \in V(\Phi)$.
It can be interesting whether the CSP with structure $S_1$ is polynomial time solvable.

In this paper, we restrict ourselves to the reals as domain, i.e., $\domain = \R$
and denote them as \textit{\continuousConstraintSatisfactionProblem{}s} (CCSPs).

We are mainly interested in CCSPs where the constraints are semi-algebraic over the integers (see \Cref{subsec:ER} for a formal definition).
There are some constraints that are not semi-algebraic, in other words, not computable on the real RAM~\cite{SmoothingGap}.
For example, constraints involving  $\sin, \cos, \exp, \log$ or testing if a number is an integer.
We do not want to forbid those types of constraints in the general definition, as it might be interesting to study some of them. 
Some of our hardness results actually apply to non-computable functions.
There are constraints that limit us to finite domains. 
For example, $x(x-1) = 0$. 
Although they are not truly continuous, they are indeed semi-algebraic and thus we have to deal with them as well
in the general definition.
For our results, we use both discrete constraints, e.g.~$x=1$, and truly continuous constraints, e.g.~$f(x) = y$  with $f$ three times differentiable.

Constraint satisfaction problems  have a long history
in algorithmic studies~\cite{schaefer1978complexity, bulatov2006dichotomy, bulatov2017dichotomy, zhuk2020proof, marx2005parameterized,dyer2010approximation}.
There are two application-driven motivations to study them.
On the one hand, it is possible to easily encode many fundamental
algorithmic problems directly as a CSP.
Then, given an efficient algorithm for those types of CSPs, 
we have immediately also solved those other algorithmic problems.
On the other hand, if we can encode CSPs into algorithmic problems,
then  any hardness result for the CSP immediately carries over to the algorithmic problem. 
Next to an application-driven motivation, it is fair to
say that they deserve a study in their own right
as fundamental mathematical objects.
CSPs  form a very versatile language 
and often allow for a complete classification by their
computational complexity. 
Specifically, the \textit{dichotomy conjecture} states that every class of CSP with a finite domain
is either \NP-complete or polynomial-time solvable. 
Schaefer showed the conjecture for domains of size two~\cite{schaefer1978complexity}.
Recently, Bulatov and Zhuk could confirm the conjecture independently~\cite{bulatov2017dichotomy, zhuk2020proof}
for any finite domain.
Note that one can also try to find a classification from the parameterized complexity
perspective~\cite{marx2005parameterized} 
or the approximative counting perspective~\cite{dyer2010approximation}.

In this paper, we focus on CSPs with \R as domain
and we are interested in the class of CSPs that are \ER-complete. 
We want to point out that there is also a large body of research
that deals with infinite domains~\cite{zhuk2021complete, viola2020combined, bodirsky2008non, bodirsky2010complexity, jonsson2016constraint}.
Most relevant for us is the work by Bodirsky, Jonsson and von Oertzen~\cite{bodirsky2012essential}, who also studied CSPs over the reals and showed that a host of them are \NP-hard to decide.
Specifically, they defined a subset $S$ of $\R^n$ as \textit{essentially convex} if for all $a, b \in S$, the straight
line segment intersects the complement $\overline{S}$ of $S$ in finitely many points.
They show that CSPs that contain $x=1$, $x \leq y$, $x+y = z$, and at least one constraint that is \textit{not} essentially convex
are \NP-hard.
However, their techniques do not imply \ER-hardness.
See also~\cite{bodirsky2017constraint} for an overview of results for the real domain.

\subsection{Existential Theory of the Reals}\label{subsec:ER}
The class of the existential theory of the reals \ER (pronounced as `ER') is a complexity class that has gained a lot of interest, especially within the computational geometry community. 
To define this class, we first consider the problem ETR, which also stands for Existential Theory of the Reals. 
In an instance of this problem, we are given  a  sentence of the form
\[
\exists x_1, \ldots, x_n \in \R : \Phi(x_1, \ldots, x_n),
\]
where~$\Phi$ is a well-formed quantifier-free formula consisting of the symbols $\{0, 1, x_1, \ldots, x_n, +, \cdot, \geq, >, \wedge, \vee, \neg\}$, the goal is to check whether this sentence is true. 
We will refer to the formula $\Phi$ which might appear in an ETR-instance as an ETR-formula.
As an example of an ETR-instance, we could take $\Phi  =  (x \cdot y^2 + x \geq 0) \wedge \neg(y < 2x)$.
The goal of this instance would be to determine whether there exist real numbers~$x$ and~$y$ satisfying this formula.
Now the class \ER is the family of all problems that admit a polynomial-time many-one reduction to ETR.
With the notation above, it can be shown that the CSP of the structure $\mathbf{R} = \langle \R, \cdot, +, \mathbf{1} \rangle $ is \ER-complete~\cite{matousek2014intersection}, see also \cref{lem:Reduction-AMI-INV}.

It is known that
\[
\NP \subseteq \ER \subseteq \PSPACE.
\]
The first inclusion follows from the definition of \ER as follows. 
Given any Boolean satisfiability formula, we can replace each positive occurrence of a variable $x$ by $x = 1 $. For example $(x\lor \lnot y) \land (\lnot x \lor z)$ becomes
$( x=1 \lor  \lnot (y=1)) \land ( \lnot (x=1) \lor z=1)$.

Showing the second inclusion was first done by Canny in his seminal paper~\cite{canny1988some}. 
The reason that \ER is an important complexity class is that a number of common problems in computational geometry, game theory, machine learning, and other areas have been shown to be complete for this class. 

We 
use $|\Phi|$ to denote the length of $\Phi$, that is, the number of bits necessary to write down
$\Phi$. 

We want to point out that there are some subtleties in the definition of the formula length. 
Naively, to encode a natural number $n$ requires $\Theta(n)$ bits, i.e., $n = 1 + 1 + \ldots + 1$.
However, it is possible to encode it in $O(\log n)$ bits, using Horner's rule applied to the binary expansion of $n$. 
For example,  $27 =  1 + 2(1 + 2(0 + 2(1 + 2))) = 1 + (1+1)(1 + (1+1)(0 + (1+1)(1 + (1+1)))) $.
Furthermore, we want to emphasize that the reductions used for defining \ER are performed in the \wordRAM model (or equivalently on a Turing machine), and not on a real Random Access Machine (\realRAM) or in the Blum-Shub-Smale model.

The definition of a formula naturally leads to the definition of semi-algebraic sets.
We say a set $S\subseteq \R^n$ is \textit{semi-algebraic}, if there exists a formula $\varphi$ such that 
$S = \set{x\in \R^n}{\varphi(x)}$.
Consequently, the (bit)-complexity of a semi-algebraic set is the shortest length of any formula
defining the set.
Note that our definition of a semi-algebraic set is more common in a computer science context~\cite{matousek2014intersection}.
In the context of algebraic geometry, semi-algebraic sets would usually allow polynomials with real coefficients. 
For example, consider the set $S = \set{x\in \R}{x-e = 0}$, containing Euler number $e$.
Note that $S$ is typically semi-algebraic for an algebraic geometer~\cite{basu2006algorithms}, but typically not for a computer scientist~\cite{matousek2014intersection}.
Given a point $x\in \R^n$ and a semi-algebraic set~$S\subset \R^n$, we can decide on the \realRAM if $x\in S$.
(We refer the reader to the work by Erickson, Hoog, and Miltzow for a detailed definition of the \realRAM and decidability~\cite{SmoothingGap}.)
This is easy to see as we only need to evaluate the defining formula of $S$.
Interestingly the reverse direction also holds. 
If we can decide $x\in S$ for any $x\in \R^n$ then $S$ needs to be semi-algebraic~\cite{SmoothingGap}.

\paragraph{Scope.} The main reason that the complexity class \ER gained traction in recent years 
is the increasing number of important algorithmic problems that are \ER-complete.
Marcus Schaefer established the current name and pointed out first that
several known NP-hardness reductions actually  imply \ER-completeness~\cite{Schaefer2010}.
Note that some important reductions that establish \ER-completeness were done
before the class was named.

Problems that have a continuous solution space and non-linear relation between partial solutions
are natural candidates to be \ER-complete.
Early examples are related to the recognition of geometric structures:
points in the plane~\cite{mnev1988universality,shor1991stretchability},
geometric linkages~\cite{schaefer2013realizability},
segment graphs~\cite{kratochvil1994intersection, matousek2014intersection},
unit disk graphs~\cite{mcdiarmid2013integer, kang2011sphere},
ray intersection graphs~\cite{cardinal2017intersection}, and
point visibility graphs~\cite{cardinal2017recognition}.
In general, the complexity class is more established in the graph drawing community~\cite{AnnaPreparation, AreasKleist, schaefer2021complexity, erickson2019optimal}.
Yet, it is also relevant for studying polytopes~\cite{richter1995realization, NestedPolytopesER}.
There is a series of papers related to Nash-Equilibria~\cite{berthelsen2019computational, Schaefer-ETR, garg2015etr, bilo2016catalog, bilo2017existential}.
Another line of research studies matrix factorization problems~\cite{chistikov_et_al:LIPIcs:2016:6238, shitov2016universality, Shitov16a, TensorRank}.
Other \ER-complete problems are the Art Gallery Problem~\cite{ARTETR,ArtJack}, Covering polygons with convex polygons~\cite{abrahamsen2021covering}, and 
training neural networks~\cite{abrahamsen2021training, TrainFullNeuralNetworks, Z92}.

\paragraph{Practical Implications.}
At first glance, the significance of \ER-completeness might not be immediately apparent, particularly given that most of these problems are already known to be \NP-hard.
The significance has different aspects.
One reason is that we are intrinsically interested in establishing
the true complexity of important algorithmic problems.
Furthermore, \ER-completeness helps us to understand better the 
difficulties encountered when designing algorithms for those types of problems.
While we have a myriad of techniques for NP-complete problems, most
of these techniques are of limited use when we consider \ER-complete problems.
The reason is that \ER-complete problems have an infinite set of possible
solutions that are intertwined in a sophisticated way.
Many researchers have hoped
to discretize
the solution space, but success was limited~\cite{Simon, matousek2014intersection}.
The complexity class \ER connects all of those different problems and tells 
us that we can either discretize all of them or none of them.
To illustrate our lack of sufficient worst-case methods, note that we do 
not know the smallest square container to pack eleven unit squares, see \Cref{fig:eleven}.

\begin{sidefigure}[tbhp]
    \centering
    \includegraphics[page = 3]{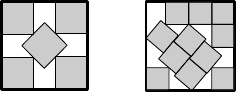}
    \caption{Left: Five unit squares into a minimum square container. Right: 
    This is the best known packing of eleven unit squares into a square container~\cite{gensane2005improved}.}
    \label{fig:eleven}
\end{sidefigure}

\paragraph{Technique.}
In order to show \ER-hardness, usually two steps are involved.
The first step is a reduction to a technical variant of ETR.
The second step is a reduction from that variant to the problem at hand.
Those ETR variants are typically CCSPs with only very limited types of constraints.
It is common to have an addition constraint ($x+y=z$), 
and a non-linear constraint, like one of the following:
$$
    z = x\cdot y, \quad 
    z = x^2,    \quad 
    1 = x\cdot y.
$$
To find the right non-linear constraint is crucial for the second step, as it is often very difficult to encode non-linear constraints in geometric problems.
Previous proof techniques relied on expressing multiplication
indirectly using other operations.
To be precise, we say that a constraint~$c$ of~arity $\ell$ has a \textit{primitive positive definition}
in structure~$S$, if there is a \constraintFormula $\Phi$ in~$S$ such that $c(y_1,\ldots,y_\ell)$ 
if and only if $\exists x_1,\ldots,x_k :\Phi(y_1,\ldots,y_\ell,x_1,\ldots,x_k)$.
In that case, $\Phi$ is called a \textit{primitive positive formula}, or just \textit{pp-formula}.
For instance, we can express multiplication using squaring and addition as follows:
\[ x\cdot y = \left(\frac{x+y}{2}\right)^2 - \left(\frac{x-y}{2}\right)^2.\]
This translates into a pp-formula 
as follows. $\exists A_0,A_1,A_2,B_0,B_1,B_2 : $

\vspace{0.3cm}
$
\begin{array}{rcl}
   A_0 & = & x+y,  \\
   A_0 & = & A_1 + A_1, \\
   A_2 & = & A_1^2, \\
\end{array} $
\hspace{1cm}
$
\begin{array}{rcl}   
   x & = & y+B_0, \\
   B_0 & = & B_1 + B_1, \\
   B_2 & = & B_1^2,
\end{array}   
$
\hspace{1cm}
$A_2 =  B_2 + z.$
\vspace{0.3cm}

\noindent Given a pp-formula, we can reduce a CSP with constraint~$c$ 
to a CSP with a different signature.
Here, we replaced the ternary constraint $x\cdot y = z$ by the binary constraint $x^2 = y$.

Furthermore, there are often some range constraints of the form 
$x>0$, $x\in [1/2,2]$ or even $x \in [-\delta,\delta]$, for some
$\delta = O(n^{-c})$, where $n$ is the number of variables. These constraints can be imposed on either all, or a subset of the variables.
This makes the above reduction more involved, as we need to pay attention to the ranges in every step.
Range constraints are important as we may
only be able to encode variables in a certain limited range.
Finally, it may be useful to know some structural properties
of the variable constraint graph, like planarity~\cite{AnnaPreparation}.

Overall, those techniques have their limitations.
As the reductions rely on an explicit way to express
one non-linear constraint by another non-linear constraint and addition,
we have to find those identities.
To illustrate this, we encourage the reader to find a way to express multiplication (in some range)
using $x^2 + y^2 = 1$ and linear constraints.
(We consider the constraint $x= 1$ to be linear. See \Cref{app:Circle} for the solution.)
This gets more tricky when dealing with inequality constraints.
For instance, it is not clear how to express multiplication with
$x\cdot y \geq 1$ and $x^2 + y^2 \geq 1$.
We offer $10$~euro to the first person, who is able 
to find a pp-formula to do so.
Note that our theorems imply that those two inequalities together 
with linear constraints are enough to establish \ER-completeness,
but we do not describe a pp-formula.
At last, translating a pp-formula into a reduction that
respects the range constraints for every variable becomes very tedious
and lengthy.
Furthermore, it only establishes \ER-completeness for those specific constraints.
See Abrahamsen and Miltzow~\cite{abrahamsen2019dynamic} for some of those reductions.

To overcome this limitation, we develop a new technique that 
establishes \ER-completeness for virtually any one non-linear equality constraint.
We extend our results and show that any one 
convex and any one concave inequality constraint are also sufficient to establish \ER-completeness.
See \Cref{sub:results} for a formal description of our results and \Cref{sub:overview} 
for an overview of our techniques. 

\subsection{Results}
\label{sub:results}
We focus on the special case with essentially
only one addition constraint and any one non-linear constraint, see \Cref{def:f-ETR}.
While this may seem like a strong limitation,
note that addition constraints are commonly easy to encode.
In most applications, the non-linear constraint is the crucial~one.
Before we introduce the main definition, we first specify more precisely how we define the non-linear constraints.

\begin{definition}[Function constraints]
Let $\domain \subseteq \R^2$ and let $f \colon \domain \to \R$ be any function. 
Now we define two constraints corresponding to $f$ as
\[
\equalzero{f}
= \set{(x, y) \in \domain}{f(x,y) = 0} \cup (\R^2 \setminus \domain),
\]
and 
\[
\largerzero{f}
= \set{(x, y) \in \domain}{f(x,y) \geq 0} \cup (\R^2 \setminus \domain).
\]
\end{definition}

For convenience, we often use the shorthand notation $f(x, y) = 0$ and $f(x, y) \geq 0$.
Note that this definition means that the constraints \equalzero{f} and \largerzero{f} are satisfied whenever $(x, y)$ is outside of the domain $\domain$ of the function $f$.
We defined the constraints in this way as it turns out that it makes the soundness part of the proof for future reductions considerably easier.
One simply does not need to worry about solutions leaving the domain.
We will show that our difficult instances are actually \domainadherent, as we will define below.
\begin{definition}[\domainadherent]
Let $\Phi$ be a \CCSP formula
that contains some function constraints, i.e.~$\largerzero{f}$ or $\equalzero{f}$.
Here $f$ is a function on the domain $\domain\subset \R^2$.
We say a solution $x = (x_1,\ldots,x_n)\in \R^n$ is \emph{\domainadherent} 
if for every function constraint on variables $x_i,x_j$,
we have that $(x_i,x_j) \in \domain$.
We say $\Phi$ is \domainadherent if this is true for all solutions.
\end{definition}

\begin{definition}[Curved equality problem (\CE)] \label{def:f-ETR} Let $\domain\subseteq \R^2$ and let $f$ be a function $f \colon \domain \rightarrow \R $.
Then we define the signature $C(f, \delta)$ as
\[C(f, \delta) = \{x+y = z,\;  \equalzero{f},\; x \geq 0,\; x = \delta\}.\]
In the \textit{\CE{} problem}, the input consists of a $\delta \in \R$ and a constraint formula $\Phi$ on $n$ variables.
The formula $\Phi$ corresponds to the structure $\langle \R, C(f, \delta) \rangle$, where we are promised that $V(\Phi) \subseteq [-\delta, \delta]^n$ and that $V(\Phi)$ is \domainadherent. 
We are asked whether $V(\Phi) \neq \emptyset$.
\end{definition}

Note that the two promises $V(\Phi) \subseteq [-\delta, \delta]^n$ and \domainadherent, while formally independent, are proven essentially in the same way, by scaling variables sufficiently close to the origin.

We would like to emphasize that we are not having a constraint of the form $x\in [-\delta,\delta]$ in the signature. 
The property $V(\Phi) \subseteq [-\delta, \delta]^n$ will be imposed only by using constraints of the structure $\langle \R , C(f,\delta) \rangle$.
Thus, we are dealing with so-called promise problems from computational complexity.
Note that this promise is difficult to check.
We will discuss promise problems again in \Cref{sub:discussion}.

This promise is also the reason why we defined the constraint $\equalzero{f}$ to be true everywhere outside of the domain of $f$.
We anticipate that most applications of our result will likely focus on the case where $[-\delta, \delta]^2$ is a subset of \domain.
We will discuss the \equalzero{f} constraint in more detail in \Cref{sub:discussion}.

Note that although the problem is called curved equality problem, we make 
no assumptions on~$f$ as part of the definition. 
We do this explicitly, as there are various technical ways to formulate those assumptions.
Abrahamsen, Adamaszek, and Miltzow~\cite{ARTETR, abrahamsen2019dynamic} essentially showed 
that \CE{} is \ER-complete for $f = (x-1)(y-1) - 1$.
Here, we generalize this to a wider set of functions~$f$ defined 
below.
Recall that a set $T$ is a \textit{neighborhood} of a point $p$ if there is an open set $S$ with $p \in S \subseteq T$.

\begin{definition}[\Wellbehaved, \triplealgebraic]
    A function $f\colon \domain \rightarrow \R$ is \textit{\wellbehaved} 
    if the following conditions are~met. 
    \begin{itemize}
        \item $f$ is a $C^3$-function, with $\domain \subseteq \R^2$ being a neighborhood of $(0, 0)$,
        \item 
        $f(0, 0) = 0$, and all partial derivatives $f_x$, $f_y$, $f_{xx}$, $f_{xy}$ and $f_{yy}$ are rational
        in $(0,0)$.
        \item $f_x(0, 0) \neq 0$ or $f_y(0, 0) \neq 0$.
    \end{itemize}

    A function $f\colon \domain \rightarrow \R$ is \textit{\triplealgebraic} if each of the three sets \domain, 
        \set{(x,y)\in \domain}{f(x,y) = 0} and \set{x,y\in \domain}{f(x,y) \geq 0}
        is semi-algebraic.
\end{definition}

Note that if $p(x,y)$ is a polynomial of the form $\sum_{i,j} a_{i,j} x^i y^j$, then 
$p$ is \wellbehaved if and only if $a_{0,0} = 0$, $a_{1,0},a_{0,1},a_{2,0},a_{1,1},a_{0,2}$ are rational,
and 
at least one of $a_{1,0}$ and $a_{0,1}$ is nonzero.
We want to point out that some readers might find it easier to think of \domain 
as a disk with a small radius. 
To see that this is equally strong note that a disk around the origin is also a neighborhood of the origin
and also a semi-algebraic set. (We need to ask for the radius to be an algebraic number.)
Thus the \ER-completeness also works for the case that \domain is such a disk.
But also the~\ER-completeness for disks implies the \ER-completeness for neighborhoods.
Although a formal proof is a bit tedious the intuition is that we can restrict the
range of the variables to lie within the disk given by the neighborhood condition.
We decided to use the language of neighborhoods, instead of disks, as we find it more
convenient to work with neighborhoods, at the cost of being a bit more abstract than
absolutely necessary.

\begin{definition}[\Nonlinear]
    Let  $f\colon \domain \rightarrow \R$ be a function that is \wellbehaved. 
    We write the curvature of $f$ at zero by
    \[\kappa = \kappa(f)  = \left(\frac{f_y^2f_{xx} - 2f_xf_yf_{xy} + f_x^2f_{yy}}{(f_x^2 + f_y^2)^{\frac{3}{2}}}\right)(0,0),\]
    see \Cref{fig:kappa} for an illustration.
    We say $f$ is 
    \begin{itemize}
        \item \textbf{\nonlinear} if $\kappa(f) \neq 0$,
        \item \textbf{\convex} if $\kappa(f) < 0$, and
        \item \textbf{\concave} if $\kappa(f) >0$.
    \end{itemize}
\end{definition}

Note that the magnitude of $\kappa$ equals the inverse of the radius of the  osculating circle of the curve $\{(x,y)\in \R^2 : f(x,y) = 0\}$ at the origin, see \Cref{fig:kappa2}. 
This is the circle which approximates the curve as close as possible. 
The sign of $\kappa$ indicates on which side the osculating circle touches the curve. It is positive if this is on the side where $f$ is negative, negative if the circle touches on the side where $f$ is positive, and zero if the origin is an inflection point and the osculating circle is a line.

\begin{sidefigure}[tbph]
    \centering
    \includegraphics{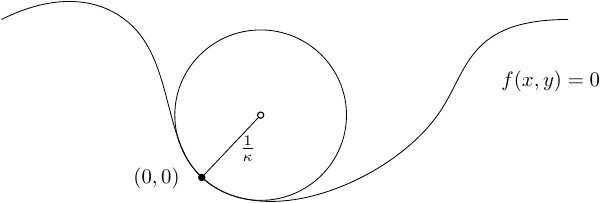}
    \caption{
    The formula for $\kappa$ describes the inverse of the radius of the osculating circle touching $(0,0)$ on the curve defined by $f$.}
    \label{fig:kappa2}
\end{sidefigure}

Note that we can define the simpler expression $\kappa' = \kappa'(f)$
 \[\kappa'(f)  = \left(f_y^2f_{xx} - 2f_xf_yf_{xy} + f_x^2f_{yy}\right)(0,0),\]
 and it holds that $\sign(\kappa) = \sign (\kappa')$.
 For this reason, we will work with $\kappa'$ instead of $\kappa$.

Consider a polynomial~$p$ of the form
$p(x, y) = \sum_{i, j} a_{ij} x^iy^j$.
Then $\kappa'(p)$ equals
\[ \kappa'(p) = a_{01}^2 2a_{20} - 2a_{10}a_{01}a_{11} + a_{10}^2 2a_{02}.\]

In order to define the possible domain of $\delta$, we still need one more definition.

\begin{definition}
We say a function $T \colon \N \to \Q_{>0}$ is \textit{bounded} if there is a constant $C$ such that~$T(n)\leq C$, for all $n$.
 The function $T$ is referred to as \textit{\nicelyComputable} if $T(n)$ can be expressed as a fraction of integers represented in binary, and this representation is computable in time polynomial to the size of the integer representation of $n$.
\end{definition}
We will create instances with $\delta = T(n)$ as input.
Some functions that satisfy these conditions are $T(n) = 1$,  $T(n) = n^{-c}$, for some fixed constant $c$, or
the function $T(n) = 2^{-n}$.

Now, we are ready to state our main theorem for equality constraints.
     \newcommand{\theoremastring}{}
\begin{restatable}{theorem}{fETRTHM}
\label{thm:Equality}{\normalfont\bfseries\theoremastring}
Let $f \colon \domain \rightarrow \R$ be a function that is  \wellbehaved, \nonlinear, and \triplealgebraic. 
Let $T$ be a function that is both bounded and \nicelyComputable. In this setting, \CE{} is \ER-complete, even when considering only instances where $\delta = T(n)$, with $n$ being the number of variables.
\end{restatable}

Note that \ER-membership follows from the fact that \domain and \set{x,y\in\R^2}{f(x,y) = 0} 
        are semi-algebraic.
        Therefore, there must be \etr formulas $\varphi_U$ and $\varphi_f$ such that the following two statements hold.
        \begin{itemize}
            \item  $\varphi_U(x,y)$ is true if and only if $(x,y)\in U$.
        \item $\varphi_f(x,y)$ is true if and only if $(x,y)\in \set{x,y\in \domain}{f(x,y) = 0}$.
        \end{itemize}
        Let $\Phi$
        be a \CE-formula.
        We replace each occurrence of $\equalzero{f}(x,y)$ in $\Phi$
        by $\varphi_f(x,y) \lor \lnot \varphi_U(x,y)$.
        This gives us a new equivalent \etr formula $\Phi'$.
        And thus \CE is in \ER.
        Note that the fact that $f$ is \triplealgebraic is not needed for the \ER-hardness part of \Cref{thm:Equality}.

\bigskip

The motivation for this article was to give a convenient tool to show \ER-hardness
of geometric packing. 
Unfortunately, we are only capable of encoding inequality constraints 
in geometric packing. 
Thus, in order to apply our techniques to geometric packing, 
we  adapt \Cref{thm:Equality} to inequality constraints.
In the following we define the \textit{convex concave inequality problem} (\CCI),
which is completely analogous to \CE{} with one subtle 
difference.
The constraint $f(x,y) = 0$ is replaced by 
$f(x,y)\geq 0$ and $g(x,y)\geq 0$.
The \nonlinear constraint $f(x,y) = 0$ is replaced by \convex and \concave conditions $f(x,y) \geq 0$ and $g(x, y) \geq 0$.

\begin{definition}[Convex concave inequality problem (\CCI)] \label{def:CCI} 

Let $\domain\subseteq \R^2$ and let $f,g$ be functions such that $f,g \colon \domain \rightarrow \R$.
Then we define the signature $C(f,g, \delta)$ as
\[C(f,g,\delta) = \{x+y = z,\;  \largerzero f,\;  \largerzero{g},\; x \geq 0,\; x = \delta\}.\]

In the \textit{\CCI{} problem}, the input consists of a $\delta \in \R$ and a constraint formula $\Phi$ on $n$ variables.
The formula $\Phi$ corresponds to the structure $\langle \R, C(f,g, \delta) \rangle$, where we are promised that $V(\Phi) \subseteq [-\delta, \delta]^n$ and \domainadherent. 
We are asked whether $V(\Phi) \neq \emptyset$.
\end{definition}

    \newcommand{\theorembstring}{}
\begin{restatable}{theorem}{CCITHM}
 \label{thm:Inequality}{\normalfont\bfseries\theorembstring}
Let  $f,g \colon \domain \rightarrow \R$ be \wellbehaved and \triplealgebraic.
Furthermore, let $f,g$ be respectively \convex and \concave.
Let $T$ be bounded and \nicelyComputable.
 In this setting, \CCI{} is \ER-complete, even when considering only instances where $\delta = T(n)$, with $n$ being the number of variables.
\end{restatable}

To show that \CCI is in \ER 
goes along the same lines 
as the proof that \CE is in \ER.
Again, the fact that $f$ is \triplealgebraic is not needed for the \ER-hardness part of \Cref{thm:Inequality}.

%%%%%%%%%%%%%%%%%%%%%%%%%%%%%%%%%%%%%%%%%%%%%%%%%%%%%%%%%%
\subsection{Discussion}
\label{sub:discussion}
%%%%%%%%%%%%%%%%%%%%%%%%%%%%%%%%%%%%%%%%%%%%%%%%%%%%%%%%%%
% 
\Cref{thm:Equality} and \Cref{thm:Inequality} are strong generalizations
 of the \ER-completeness of \etrinv. 
The problem \etrinv was instrumental in establishing \ER-completeness for both the Art Gallery problem~\cite{ARTETR} and the conference version of the proof for geometric packing~\cite{etrPacking}.
 
 One of the major obstacles of the \ER-completeness proofs of the Art Gallery problem
 was to find a way to encode inversion. 
 If the authors had known \Cref{thm:Equality} back then, it would have been sufficient to encode essentially any \wellbehaved and \nonlinear constraint on two variables, which is much easier.
 In this section, we discuss strengths, limitations and different perspectives with respect to our main results.

\paragraph{Comparison of Main Theorems.}
In order to compare \Cref{thm:Equality} and \Cref{thm:Inequality}, consider the following two
signatures and their interpretation for some given \wellbehaved and \nonlinear~$f$:
\[C_1 = \{ x+y = z,\ x\geq 0,\ x = \delta,\ \equalzero{f} \}, \]
and 
\[C_2 = \{ x+y = z,\ x\geq 0,\ x = \delta,\ \largerzero{f},\
\largerzero{-f}  \}.\]
Clearly,~$C_2$ is more expressive than~$C_1$.
Therefore, \ER-hardness of \CE implies \ER-hardness of \CCI in the special case $g = -f$.
For unrelated~$f$ and $g$ there is no further relation between the two theorems.
However, in the special case of $f = y- \bar{f}(x)$
and $g = \bar{f}(x) -y$, 
\ER-hardness of \CCI implies \ER-hardness of \CE as follows:
we can encode each constraint of the form $f(x,y) = y-\bar{f}(x) \geq 0$ using the new constraints
$f(x,z_1) = z_1-\bar{f}(x) = 0$, $z_2 = y - z_1$, and $z_2 \geq 0$. 
Similarly, constraints of the form $g(x,y) = \bar{f}(x)-y \geq 0$ can be encoded in $C_1$.

\paragraph{Promise Problems.}
A \textit{promise problem} is defined as an algorithmic problem where the instances are restricted to those which satisfy a certain condition.
In other words, we \textit{guarantee} that the condition holds.

In \Cref{thm:Equality}, we gave a promise on the problem instance. 
Namely, we guarantee that the solution set will be contained
in a box of a certain size. 
It is not very common that promises are formulated in this way. 
However, promise problems are very common and in particular, they can also be treated as decision problems. 
A prime example is the independent set problem on planar graphs (\maxIndPlanar ).
\maxIndPlanar is known to be \NP-complete and it is also a promise problem.
The main difference is that we can check whether a graph is planar in polynomial time. 
However, even if planarity was undecidable
the \NP-completeness of \maxIndPlanar would still be valid.

In our scenario, generally, it is not straightforward to verify the promise that the value of each variable lies within the range $[-\delta,\delta]$.
This makes our result a bit unusual.
However, it is relatively straightforward to enforce that all solutions are in the desired range.
This follows in two steps. 
In the first step,
we employ a known lemma from the real algebraic geometry literature, which ensures that some solutions must be inside a large ball.
Thereafter, we replace each variable by a scaled copy of itself.
This will require some small adaption of the constraints.
We can then enforce $-\delta\leq x \leq \delta$, by the constraints ($s = x + \delta $ and $s\geq 0$) as well as 
($t  + x= \delta $ and $t\geq 0$) without changing the truth value of the instance.

 \paragraph{First-order Theory of the Reals.} 
    With the full first-order theory of the reals it is easier than with CCSPs to define all semi-algebraic sets. 
    Specifically, we can define non-convex sets using only convex constraints, as follows. 
    If we allow a single convex constraint $D = \set{(x,y) \in \R^2}{x^2 + y^2 \leq 1}$, then the following formula $\varphi$ describes the upper half of the boundary of the disk, 
    given by $\set{(x, y) \in \R^2}{x^2 + y^2 = 1 \land y \geq 0}$:
    \[
    \phi(x, y) = D(x,y) \land \forall_{z \in \R} (D(x,z) \Rightarrow z \leq y).
    \]
    Another way to construct a non-convex constraint using convex constraints is as follows.
    \[\phi(x,y) = D(x,y) \land \lnot D(x-1,y)\]
    Note that the second example only uses negations and no quantifiers.
        
    Using just the language of CCSPs, it is however impossible to encode such a set using only linear constraints and the constraint $D(x,y)$, as any CCSP instance of this form describes a convex set. Note in particular that we may not apply quantifier elimination to the given formula~$\phi$, since this is impossible without introducing non-linear constraints different from~$D(x,y)$. 
    For a more extensive analysis of the semi-algebraic sets which can be described using the first order theory of the reals when the set of atomic formulas is restricted, we refer to \cite{marker1992additive, peterzil1992structure, peterzil1993reducts}.

    \paragraph{Convex Constraints.} We want to point out that addition and \convex constraints alone seem not
to be sufficient to establish \ER-completeness, as convex programs 
have efficient approximation algorithms~\cite{boyd2004convex}.
As convex programming is so efficiently fast solvable in practice
it would be a big surprise if it would be \ER-complete.
However, there are reasons to believe that
convex programming is not polynomial-time solvable, 
see the discussion by O'Donnell~\cite{o2017sos}.
See~\cite{pataki2021exponential, allender2009complexity, tarasov2008semidefinite} for an in-depth discussion why convex programming is potentially not polynomial-time solvable.

\paragraph{Concave Constraints.} When we remove the convex constraint but keep the concave constraint 
in \Cref{thm:Inequality} 
then we do not know if the problem is \ER-complete.
It is easy though to establish \NP-hardness in this case~\cite{bodirsky2012essential}. 
We consider the option that there is another complexity class \texttt{Concave} that characterizes
such CCSPs.
As with geometric packing with convex pieces, polygonal containers and translations grant the possibility to encode only linear and concave constraints. This problem is a natural 
candidate to be \texttt{Concave}-complete.
We are curious if this intuition could be supported in some mathematically rigorous way.
    
\paragraph{Unary Constraints.} 
    Note that the constraint $x = \delta$ is necessary to ensure that the origin is not always a valid solution.
    Although $x\geq 0$ may not be 
    necessary to imply \ER-completeness,
    our proof heavily relies on it.
    As an example where this constraint is not needed, consider the case where we have the constraint $y = x^2$. In this case we could replace any constraint of the form $x \geq 0$ by $x = z^2$, for some new variable $z$.
    In applications, it is usually very easy to encode unary constraints.

    \paragraph{Binary Constraints.} If we remove the addition constraint, we are left only with constraints
    in at most two variables. This seems too weak to establish \ER-completeness,
    as setting~$x$ determines~$y$, up to finitely many options once we impose the constraint 
    $f(x,y) = 0$.
    On the other hand, very large and irrational solutions can be enforced,
    which makes it unlikely for those CCSPs to be contained in general in \NP.
    We wonder about the algorithmic complexity of CCSPs with only binary constraints.
    
    \paragraph{Ternary Constraints.}
    Given the discussion above, it seems plausible that
    at least one ternary constraint is required to establish \ER-completeness.
    Therefore, we find it interesting to focus on ternary constraints.
    Let's consider first 
    the natural ternary multiplication constraint $x\cdot y = z$.
    First, we notice that setting all variables to zero satisfies this constraint. If this is the only constraint then we have a polynomial time algorithm.
    Second, the ternary multiplication constraint $x\cdot y = z$ can be transformed to
    the linear constraint $\log x + \log y = \log z$~\cite{bodirsky2017constraint}, in case all variables are positive. 
    This trick can help in case the all-zero solution is not allowed or other unary constraints are introduced.
    Therefore, the multiplication constraint does not lead to \ER-hardness by itself.
    Furthermore, due to the logarithm trick, multiplication seems somewhat easier than other ternary constraints.
    
    It is plausible that this trick or similar tricks can only be applied
    to exceptional ternary constraints. 
    We leave it as an exciting open problem to determine which ternary constraints lead to \ER-complete \CCSP{}s.
    
    \paragraph{Arbitrary Constraints.} 
    We want to point out that our results only concern constraints coming from \wellbehaved functions, instead of allowing arbitrary constraints. 
    Such a restriction is necessary, since otherwise we could, for example, consider CSPs with a constraint that forces a variable to be an integer. 
    This would allow us to encode arbitrary Diophantine equations, making the problem undecidable.
    Even more Bodirsky and Grohe~\cite{bodirsky2008non} showed that \textit{any} algorithmic decision problem has an equally difficult CSP problem.
    As a consequence, any type of classification of continuous constraints must  limit the set of allowed constraints in some~way.

    \paragraph{Variable-Constraint Graph.} We have completely neglected the variable-constraint incidence graph in this paper.
    Previous work showed that this graph can be restricted, by self-reduction
    and a clever application of the addition constraint~\cite{AreasKleist,AnnaPreparation}.
    We are curious if it is possible to classify 
    hereditary graph classes for which \CCI is \ER-complete.

    \paragraph{Universality Results.} Previous reductions of \ER-completeness usually also imply 
    so-called universality results. 
    Giving a proper introduction to universality results is outside the scope of this paper.
    Universality results translate topological and algebraic phenomena from one type of CSP to another type.
    See the lecture notes by Matou{\v{s}}ek for some introduction to universality theorems in this context~\cite{matousek2014intersection}.
    Our methods do not seem to imply these types of universality results.
    Specifically, if $f$ is a complicated function that is not even a polynomial, 
    it seems implausible that $f$ can be used to construct, say, $\sqrt{2}$.
    
    \paragraph{Algebraic Derivatives.} 
    Given the applications that we are aware of, the most complicated part was to check that $f,g$ and their derivatives are rational at the origin. 
    
    We wonder whether it might be sufficient if the values of $f,g$ and its derivatives are algebraic at the origin. 
    This weaker condition might follow from some general argument that avoids computing $f,g$ and its derivatives.

    \paragraph{Constraints true outside of Domain.}
    We have defined the constraints \equalzero{f} and \largerzero{f} to be true outside of the domain of $f$.
    We use this formulation to make our results slightly easier to apply.
    The use case is as follows.
    Assume that we have some \CE instance $\Phi$ and we are building 
    an instance $I$ of some other type of, say geometric, problem from it.
    Assume that we are able to construct a gadget 
    representing a suitable \nonlinear and \wellbehaved function~$f$.
    We have to show that $\Phi$ is a yes instance if and only if $I$ is a yes instance.
    One direction is commonly easy.
    In case that $\Phi$ is  a yes instance then there must be some $x$ that satisfies 
    all the constraints. 
    Typically, the construction of $I$ together with $x$, directly shows that $I$ is also a yes instance.
    Sometimes, the reverse direction is more difficult.
    We have to show that if~$I$ is a yes-instance then $\Phi$ is one as well.
    The tricky part is that it is sometimes conceivable that $I$ has a solution,
    but that this solution could potentially leave the intended range.
    We still have to show that all constraints of $\Phi$ are satisfied.
    This is now trivial for the constraints \equalzero{f} and \largerzero{f},
    as they are defined to be always true outside of the domain of $f$.
    In other words, the way we defined the constraints derived from $f$ ensures
    that we do not need to worry about variables leaving their range
    when applying our results.

%%%%%%%%%%%%%%%%%%%%%%%%%%%%%%%%%%%%%%%%%%%%%%%%%%%%%%%%%%
\subsection{Alternative Descriptions}
\label{sub:Alternative}
%%%%%%%%%%%%%%%%%%%%%%%%%%%%%%%%%%%%%%%%%%%%%%%%%%%%%%%%%%
In this subsection, we want to make some comments that might make it easier to apply our results to CCSPs where the constraints are given in explicit form or by a parametrization.
Before we delve into technical details consider the following example.

\begin{sidefigure}[thbp]
    \centering
    \includegraphics{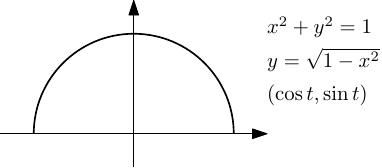}
    \caption{Three descriptions of the points on the semi-circle.}
    \label{fig:semi-circle}
\end{sidefigure}

\begin{example}
    Let $S = \set{(x,y)\in \R^2}{x^2+y^2 = 1}$ be the circle.
    We can define the function $f(x,y) = x^2+y^2 -1$,
    which describes the set $S$ by \set{(x,y)\in \R^2}{f(x,y)=0}.
    We can easily check that $f$ is 
    \wellbehaved 
    and \nonlinear and apply 
    Theorem~\ref{thm:Equality}.
    (Note that $f(0,0) \neq 0$, but this can be easily fixed by shifting $f$. We do not do this here to keep the notation simple.)
    
    In a different setting, we may know that $S$ is the graph of a function $h\colon [-1,1] \to \R$ given by $h(x) = \sqrt{1-x^2}$. 
    This would mean
    \[
    S = \set{(x,h(x)) \in \R \times \R}{x \in [-1,1]}\cup \set{(x,-h(x)) \in \R \times \R}{x \in [-1,1]}. 
    \]
    If we are given such a description, it is possible to 
    % go
    rewrite the condition $y = h(x)$ as $f(x,y) = 0$ for $f(x,y) = y-h(x)$, and we can check whether $f$ satisfies the necessary conditions. 
    Instead, it turns out we can more easily check the relevant conditions directly on $h$.

    Another description of the set $S$ could be by a parametrization $\gamma(t) = (\cos t, \sin t)$.
    With this it holds that $S=\set{\gamma(t)\in \R^2}{t \in [-\pi,\pi]}$.
    While we know $f$ in this specific case, in general, it is not so easy anymore to give an explicit description of $f$.
    We will give some conditions on $\gamma$ which can be checked to ensure that our theorems can be applied when the constraint has such a parametrized form.
\end{example}
In this section, we derive sufficient conditions to check that $f$ exists and is \wellbehaved and \nonlinear, even when we do not know how to describe $f$ explicitly.
These conditions are used in at least two applications~~\cite{etrPacking,FabianMapLabel}.
In this section, we ignore issues about \ER-membership as they are not so easy to handle and we believe that most readers care about the \ER-hardness part anyways.

\paragraph{Explicit Description.} 
We consider the case that the constraint $f(x,y) = 0$ is described by $y = h(x)$. 
\begin{definition}[\wellbehaved, \nonlinear]
     We say $h \colon I \to \R$ is \textit{\wellbehaved} if it satisfies the following conditions.
    \begin{itemize}
        \item $h$ is a $C^3$ function, with  $I\subset \R$ being an interval with $0$ in its interior.
        \item $h(0) = 0$, and $h'(0)$ and $h''(0)$ are rational.
    \end{itemize}

    \noindent We say $h$ is 
    \begin{itemize}
        \item \textbf{\nonlinear} if $h''(0) \neq 0$,
        \item \textbf{\convex} if $h''(0) > 0$, and
        \item \textbf{\concave} if $h''(0) < 0$.
    \end{itemize}
\end{definition}
Note that our definition of $h$ being \convex corresponds to $h$ being a convex function.
Note that we have now defined the terms \wellbehaved, \nonlinear etc both for $f$ as well as for $h$. 
The next lemma justifies this overload, as it shows that they also exactly correspond to one another.
   \begin{lemma}
    \label{lem:ExplixitRational}
    Let $h \colon I \to \R$ be a \wellbehaved function and
    $f \colon I\times \R \to \R$ be defined as $f(x, y) = y - h(x)$. 
    Then $f$ is \wellbehaved as well.
    Furthermore, any property (\nonlinear, \convex, \concave) that is satisfied by $h$ is also satisfied by $f$.
    \end{lemma}
\begin{proof}
    We check all conditions of \wellbehavedness one-by-one.
    We note that $\domain = I\times \R$ is indeed a neighborhood of the origin.
    As $h$ is $C^3$, so is $f$.
    It holds that $f(0,0) = 0 - h(0) = 0$.
    The derivatives have the following form:
    \[f_x = -h',\quad f_y = 1,\quad f_{xx} = -h'',\quad f_{xy} = 0,\quad f_{yy} = 0.\]
    Recall that by Young's theorem $f_{xy} = f_{yx}$.
    As $h(0)$, $h'(0)$, and $h''(0)$ are rational, so are the values of $f_x,f_y,f_{xx},f_{xy},f_{yy}$ when evaluated in the origin.
    We now compute 
    \[\kappa'(f) = 
    \left(f_y^2f_{xx} - 2f_xf_yf_{xy} + f_x^2f_{yy}\right)(0,0)
    = - h''(0) + 2h'(0) \cdot 0 + (-h'(0))^2 \cdot 0 = -h''(0).\]    
    Thus the properties are also in correspondence to one another.
\end{proof}

We are now ready to go to the parametrized description.

\paragraph{Parameterized Description.} 
We next consider the case that the set $S$ that describes the constraint is given by a parametrization $\gamma = (a, b) \colon I \to \R^2$.

\begin{definition}
We say a parametrization $\gamma = (a, b) \colon I \to \R^2$ is \textit{\wellbehaved} if it satisfies the following conditions.
\begin{itemize}
    \item  $\gamma$ is a $C^3$ parametrization, with  $I\subset \R$ being an interval with $0$ in its interior.
    \item $\gamma(t) = (0, 0)$ if and only if $t = 0$.
    \item The functions $a', b', a'', b''$ are all rational in $0$
        and $a'(t) > 0$, $\forall t \in I$.
\end{itemize}

\noindent
We define $\kappa'(\gamma) = a'' \cdot b' - b''\cdot a'(0)$.
We say $\gamma$ is 
    \begin{itemize}
        \item \textbf{\nonlinear} if $\kappa'(\gamma) \neq 0$,
        \item \textbf{\convex} if $\kappa'(\gamma) < 0$, and
        \item \textbf{\concave} if $\kappa'(\gamma) >0$.
    \end{itemize}

\end{definition}

Here, we made the assumptions that $a'(t)> 0, \forall t \in I$, so that we can apply the inverse function theorem globally.
This makes the notation easier and removes some pathological cases when $\gamma$ approaches the origin for larger $t$.
Intuitively, the condition states that $\gamma$ goes from left to right.
Note that if $a'(t)< 0, \forall t \in I$,
we can just replace the parametrization $\gamma(t)$ by $\gamma(-t)$.
It is also not such a strong assumption as 
$a'(0)\neq 0$ also implies 
$a'(t)\neq 0$  $\forall t\in J$ for some sufficiently small open interval $J$ containing zero.

As we will show later, if $\gamma$ is \wellbehaved then $\gamma$ describes the graph of a function $h(x) = y$.
See \Cref{fig:bad-curve} for a geometric illustration for the different cases that can occur if one of the conditions is dropped.
Specifically, we can describe $h(x)$ by $b(a^{-1}(x))$.

\begin{figure}[t]
    \centering
    \includegraphics{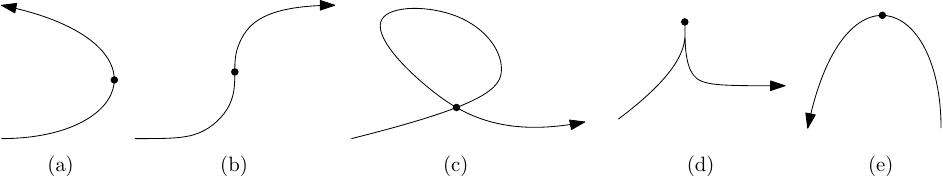}
    \caption{The curves displayed in (a) to (d) do not satisfy our conditions. In (a) and (b) the situation could be fixed by reversing the role of $x$ and $y$. In (c) the situation could be fixed by restricting the range of $t$. Although $\gamma$ is smooth in (d), the corresponding function $h = b \circ a^{-1}$ is not. In (e) it is the case that $a'(t)<0$. We can define a new parametrization $\gamma(-t)$, which satisfies now all conditions.}
    \label{fig:bad-curve}
\end{figure}

\begin{lemma}
    \label{lem:ParaRational}
    Let the parametrization $\gamma$ be \wellbehaved. Then there is an interval around the origin $J \subseteq \R$ and some \wellbehaved $f$ such that
    \[
    \set{\gamma(t) \in \R^2}{t \in I} = \set{(x, y) \in J\times \R}{f(x,y) = 0},
    \]
    and such that $\set{(x, y) \in J\times \R}{f(x,y) \geq 0}$ is exactly the set of points above the curve given by $\gamma$.
    
    Furthermore, it is possible to choose $f$ such that any property (\nonlinear, \convex, \concave) that is satisfied by $\gamma$ is also satisfied by $f$.
\end{lemma}

The proof will make use of a specific version of 
the inverse function theorem that 
we state here for the benefit of the reader.

\begin{theorem}[Inverse Function Theorem]
Let $I \subseteq \R$ be some interval containing $0$ in its interior, and let $a \colon I \to \R$ be a $C^3$-function, with $a(0) = 0$ and $a'(t) \neq 0$ for all $t \in I$.
In this situation, $a \colon I \to a(I)$ is invertible, and its inverse $a^{-1} \colon a(I) \to I$ is a $C^3$-function.
\end{theorem}
We want to point out that the function $a^{-1}$ 
often cannot be expressed in closed form. 
For instance, if $a(t) = t^5 - t - 1$ for $t \in [1,2]$, then $a'(t)$ is positive for all $t$, but $a^{-1}(0)$ does not admit a closed form expression~\cite{Wikipedia2023Galois}.
And thus it is also not so difficult to find examples of parametrizations for which we cannot find a closed form expression by some function $f$.
Therefore, it is really useful to have conditions on $\gamma$ that we can check instead of having to find~$f$.

\begin{proof}[Proof of Lemma~\ref{lem:ParaRational}]
    First we argue that $\gamma$ describes the graph of a function $h$.
    Recall that $\gamma$ consists of the two components $a$ and $b$, i.e., $\gamma(t) = (a(t),b(t))$.
    We note that all conditions of the inverse function theorem as stated above are satisfied for $a$, thus if we let $J = a(I)$, then there is inverse function $a^{-1} \colon J \rightarrow I$ that is a $C^3$ function.
    We now define 
    \[ h(x) = b(a^{-1}(x)),\]
    for all $x\in J$.
    Using the fact that $a^{-1}$ is an inverse of $a$, it follows that
    \[
    \set{\gamma(t) \in \R^2}{t \in I} = \set{(x, h(x)) \in \R^2}{x \in J} = \set{(x,y) \in J\times \R}{y = h(x)}.
    \]
    If we define $f(x,y) = y - h(x)$ for $(x,y) \in J \times \R$, it follows that
    \[
    \set{\gamma(t) \in \R^2}{t \in I} = \set{(x, y) \in J \times \R}{f(x,y) = 0}.
    \]
    Note that a point $(x,y)$ in $J \times \R$ lies above the curve given by $\gamma$ if and only if $y \geq h(x)$, which is equivalent to $f(x,y) \geq 0$. This proves the first part of the lemma.
    
    For proving the second part, by \Cref{lem:ExplixitRational}, it is sufficient to evaluate $h(x)=b(a^{-1}(x))$ and its derivatives at $x=0$.
    We start with
    \[
    h'(0) = \frac{b'(a^{-1}(0))}{a'(a^{-1}(0))} = \frac{b'(0)}{a'(0)}.
    \]
    We continue with
    \[
    h''(0) = \frac{b''(a^{-1}(0))a'(a^{-1}(0)) - a''(a^{-1}(0))b'(a^{-1}(0))}{[a'(a^{-1}(0))]^3}
    =\frac{b''(0)a'(0) - a''(0)b'(0)}{[a'(0)]^3}
    = -\frac{\kappa'(\gamma)}{[a'(0)]^3}.
    \]
    Note that, since $a'(0) > 0$, this implies that $h''(0)$ and $\kappa'(\gamma)$ have opposite signs.
    This finishes the proof.
\end{proof}

%%%%%%%%%%%%%%%%%%%%%%%%%%%%%%%%%%%%%%%%%%%%%%%%%%%%%%%%%%
\subsection{Proof Overview for \CE{} and \CCI}
\label{sub:overview}
%%%%%%%%%%%%%%%%%%%%%%%%%%%%%%%%%%%%%%%%%%%%%%%%%%%%%%%%%%
The proofs of \cref{thm:Equality} and \cref{thm:Inequality} follow several steps which we explain in this section. 
We start by explaining how \ER-hardness of \CE (\cref{thm:Equality}) can be proven, and then we say how this can be modified to prove the hardness of \CCI (\cref{thm:Inequality}). The structure of the proof is visualized in \cref{fig:overview}.

\paragraph{Notation.} 
In our proofs, newly introduced variables will often be denoted by using double square brackets, like this: $\var{f(x)}$, $\var{x+y}$, $\var{x^2}$, etc. 
In this notation, formally the whole expression including the brackets and the symbols within it should be understood as the name of the variable, without any special meaning. 
The symbols within the brackets will usually denote the value which is intuitively represented by the variable.

\begin{figure}[tbph]
    \centering
    \includegraphics[page =2]{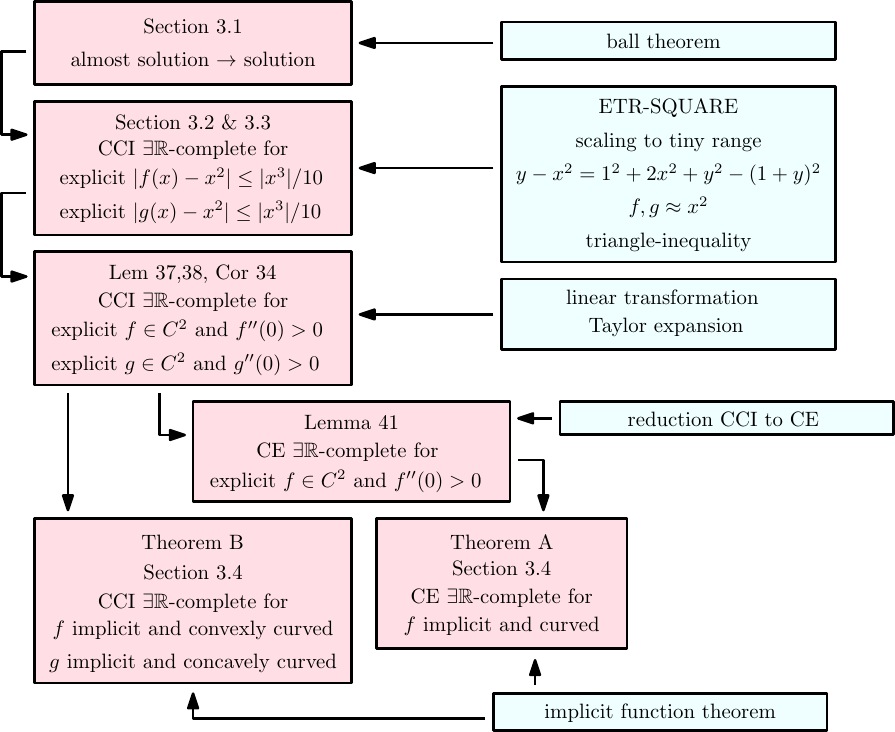}
    \caption{A formal overview of the different steps of the proof to \Cref{thm:Equality} and \Cref{thm:Inequality}.
    }
    \label{fig:overview}
\end{figure}

\paragraph{Ball Theorem.}
 One of the most important tools that we employ is a lemma
 from real algebraic geometry~\cite{basu2010bounding}.
 It states that for every ETR-formula $\Phi$ there is a ball
 $B$ whose radius only depends on the length~$L$ of $\Phi$, such that the following property is satisfied:
 if $\Phi$ has at least one solution~$x$ then 
 there must be also a solution $y$ inside the  ball~$B$.
 This theorem tells us that solutions cannot get too large.
 To get an intuition, consider the system of equations
 $x_0 = 2, x_{i+1} = x_i^2$, for $i = 0,1,\ldots,n-1$. 
 Clearly, $x_n = 2^{2^n}$, which is 
 doubly exponentially large.
 The ball theorem essentially states that we cannot get much larger numbers.

\paragraph{Range.}
    To introduce range constraints is common practice and 
    we inherit them from previous work~\cite{abrahamsen2019dynamic, etrPacking}.
    We repeat here the argument, for the benefit of the reader.
    In order to restrict the range of every variable, 
    we first note that the ball theorem already tells us that
    the range of each variable may be limited by some number~$r$.
    We construct $\eps = \delta / r$ and replace
    every variable $x$ by $\var{\eps x} = \eps \cdot x$ and
    consequently we need to adapt all constraints.
    For instance $x \cdot y = z$ becomes
    $ \var{\eps x} \cdot \var{\eps y}  = \var{\eps z} \eps $.
    In this way, we can easily ensure that if there is a solution at all
    then there is at least one solution with all variables in the 
    range $[-\delta, \delta]$. 
    
    We will make use of this re-scaling trick to 
    place all variables in an even smaller range close to zero, as the behavior of~$f$ and~$g$
    is better understood close to the origin.
    Specifically, the error $|f(x) - x^2| \leq \eps^3$
    is small enough to pretend that $f$ behaves like a squaring function.

\paragraph{Approximate Solution.}
    Using the ball theorem, we will establish that
    equality constraints of the form
    $p(x) = 0$ can be slightly weakened to $|p(x)| \leq \eps$
    for some sufficiently small $\eps$.
    To get an intuition consider the following highly simplified cases.
    
    Assume we are given a polynomial equation
    $p(x) = 0$, with $p\in \Z[X_1,\ldots,X_n]$ 
    and we are looking for a solution $x\in \Z^{n}$.
    Then in particular, we know that for all 
    $x\in \Z^{n}$ that $p(x) \in \Z$.
    This readily implies that we can equivalently ask for some $x\in \Z^n$
    that satisfies $| p(x)| \leq \tfrac{1}{2}$.
    Now, this is trivial for integers as integers have distance at least one to each other.
    But we can generalize the same principle also 
    to rational and algebraic numbers.

    Let $ S = \{ \frac{a_1}{b_1}, \ldots, \frac{a_1}{b_1}\} $ be $n$ rational numbers
    with $|a_i|,|b_i| \leq L$.
    Thus it is easy to see that $q,r \in S$ have minimum distance $\frac{1}{L^2}$.
    This implies that if $ |q - r| \leq  \frac{1}{L^2} $, for some $q,r \in S$,
    we can infer that $q=r$.
    Again, this may seem almost trivial, but relies on the simple fact
    that rational numbers with bounded numerator and denominator have a minimum distance to one another.

    \Cref{lemma:approx} generalizes the idea to algebraic numbers.
    Using the ball theorem, we will establish that algebraic numbers also have some minimum distance to one another, 
    if we restrict their bit-complexity.

\paragraph{\etrsquare .}
    We use a theorem by Abrahamsen and Miltzow that shows that \etrsquare
    is \ER-complete~\cite{abrahamsen2019dynamic}. 
    In this variant, we essentially have only addition ($x+y = z$) and squaring constraints ($x^2 = y$).
    Furthermore, the range of each variable is restricted to a small range around zero.
    For the sake of completeness and readability, we present a self-contained proof in \Cref{sec:etr-square}.
    
\paragraph{Explicit.}
    Given those tools, we can show that 
    we can replace a squaring constraint with explicit constraints ($f(x) = y$).
    We start by only considering $f$ which satisfy 
    \begin{equation}
      | f(x) - x^2 | \leq \frac{1}{10} x^{3}.  \label[cond]{Cond:AlmostSquare}
    \end{equation}
    The idea of the reduction from \etrsquare is simple but tedious.
    We can rewrite the constraint $x^2 = y$ as a linear combination of squares 
    as follows \[1^2 + 2 x^2 +  y^2 - (1 + y)^2 = 0.\]
    Now, we can replace each square using the function $f$
    to $f(1) + 2 f(x) +  f(y) - f(1 + y) = 0$.
    As~$f$ is approximately squaring, this implies
    that we are approximately enforcing the constraint~$x^2 = y $.
    In other words, we enforce 
    $|x^2 - y| \leq \eps$.
    Note that this is the technically most tedious step to make rigorous as we will later see.
    As we have discussed above it is sufficient to enforce each constraint approximately. 
    The technical difficulty is many-fold.
    We need to work with scaled variables, instead of the original variables.
    Furthermore, we have to take into consideration that 
    when we construct~$\eps$ that this also makes the formula longer.
    In particular, this means that the definition of~ $\eps$ cannot depend 
    on the newly constructed instance, but has to depend on the original instance.
    
    Using linear transformations and Taylor expansion on $f$, we can 
    replace \Cref{Cond:AlmostSquare} relatively easily  by \Cref{cond:diffbar}:
    
    \begin{equation}
        f \textrm{ is three times differentiable and } f''(0) > 0. \label[cond]{cond:diffbar}
    \end{equation}

    \begin{sidefigure}[tbhp]
        \centering
        \includegraphics{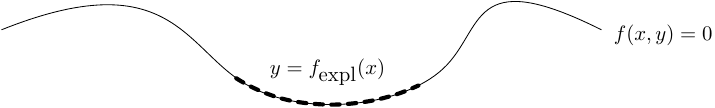}
        \caption{The implicit function theorem tells us that there is
        an function $f_\expl$ such that the curve $y = f_\expl(x)$  is locally identical to the curve $f(x,y) = 0$.}
        \label{fig:ImplicitFunction}
    \end{sidefigure}
    
\paragraph{Implicit.}
    We are now ready to handle the more general case of constraints in implicit
    form ($f(x,y) = 0$).
    The implicit function theorem tells us that there is
    a function $f_\expl$ such that the curve $y = f_\expl(x)$  
    is locally identical to the curve $f(x,y) = 0$, see \Cref{fig:ImplicitFunction}.
    The properties of the partial derivatives of~$f$ translate to properties of the partial derivatives of $f_\expl$.
    In this way, we can infer hardness of the CSP with constraint $f(x, y) = 0$ from the problem with constraint~$y = f_\expl(x)$.

\paragraph{Inequalities.}
    Until this point, we discussed the hardness proof of \CE. In \CCI, we instead have inequality constraints.
    The case of inequalities goes analogously to the equality case.
    We need one \convex and one \concave inequality.
    Whenever we want to upper bound an expression, we use 
    one inequality and whenever we need to lower bound something,
    we use the other one.
    While on the surface this is not so difficult, 
    it makes the reduction from \etrsquare to \CCI considerably more tedious.
    Specifically, it makes it harder to have an intuition
    on several technical steps and the meaning of several
    intermediate variables.
  
  %\paragraph{Acknowledgment.}

%%%%%%%%%%%%%%%%%%%%%%%%%%%%%%%%%%%%%%%%
\section{ETR-Square} 
\label{sec:etr-square}
%%%%%%%%%%%%%%%%%%%%%%%%%%%%%%%%%%%%%%%%

This section is dedicated to showing that ETR can be reduced to \etrsquare.
We execute all steps of the reduction in great detail for the sake of completeness.
This section is largely copied from
the paper by Abrahamsen and Miltzow~\cite{abrahamsen2019dynamic}. 
We mainly simplified the proofs, as we only show \ER-completeness and we leave out the parts that were needed to preserve topological or algebraic properties. Note that large parts in this section can be considered folklore.
Similar reductions have been described by Schaefer and \v{S}tefankovi\v{c}~\cite{Schaefer-ETR}.
We say two \etr formulas are \textit{\equisatisfiable} if they have the same truth value.
We say that a formula $\Phi$ is true if $V(\Phi)$ is non-empty.

%%%%%%%%%%%%%%%%%%%%%%%%%%%%%%%%%%%
\subsection{Reduction to Conjunctive Form}
\label{sec:Conjunction}
%%%%%%%%%%%%%%%%%%%%%%%%%%%%%%%%%%%

\begin{definition}
An \etrconjunction formula 
$\Phi=\Phi(x_1,\ldots,x_n)$ 
is a conjunction $C_1\land\ldots\land C_m$,
where $m\geq 0$ and each $C_i$ is of one of the two forms
\begin{align*}
x\geq 0,\quad p(y_1,\ldots,y_l) = 0
\end{align*}
for $x,y_1,\ldots,y_l \in \{x_1, \ldots, x_n\}$ and $p$ a polynomial.

In the problem \etrconjunction, we are given an \etrconjunction formula $\Phi$.
The goal is to decide if~$V(\Phi)$ is non-empty.
\end{definition}

We want to point out that \etrconjunction is a CSP, but it has an infinite set of possible constraints, one for each polynomial. 
Thus, we find it inconvenient to use the language of CSPs here. Furthermore, we allow polynomials represented by any well-defined term built from the symbols $\{0, 1, x_1, \ldots, x_n, +, -, \cdot\}$ and brackets, (e.g.~$p = (x+1)(x-1) + x$), and not just those in standard form (e.g.~$p = x^2+x -1$). 
Note that since there are no strict inequalities in a formula~$\Phi$ in \etrconjunction, the set $V(\Phi)$ is closed.
We show how to reduce a general \etr formula
to an \etrconjunction formula.

\begin{lemma}
\label{lem:Reduction-ETR-Conjunction}
  Given an \etr formula $\Phi$, we can compute  in linear time an \equisatisfiable \etrconjunction formula~$\Psi$.
\end{lemma}
\begin{proof}
 We start with an \etr formula $\Phi$
 and modify it repeatedly to attain an \etrconjunction formula
 $\Psi$. 
 Each modification leads
 to an \equisatisfiable formula.
 Our modifications can be summarized in four steps.
 (1) Delete ``$\lnot$''. 
 (2) Delete ``$>$''.
 %(3) Move ``$\geq$'' to variables only.
 (3) Delete ``$\geq$''.
 (4) Delete ``$\lor$''.
 In the rest of this proof, $p$ and $q$ denote polynomials.

 Step (1): Here, we merely ``pull'' every negation $\lnot$
 in front of every atomic predicate. 
 For instance 
 $\lnot(A\lor B\lor C)$ becomes 
 $(\lnot A\land \lnot  B\land \lnot C)$.
 To see that this can be done in linear time,
 note that the length of $\Phi$ is at least
 the number of atomic predicates.
 At the end of this process, every atomic predicate
 is preceded by either a negation or not. 
 It may be that $\land$ and $\lor$ symbols are swapped,
 but each is counted as one symbol. 
 
 Thereafter each 
 atomic predicate preceded by $\lnot$ is replaced as follows:
 \begin{align*}
  \lnot (q > 0) \quad &\mapsto \quad -q \geq 0\\
  \lnot (q = 0) \quad &\mapsto \quad (q > 0) \,  \lor \,  (-q > 0) \\
  \lnot (q \geq 0)\quad  &\mapsto \quad -q > 0
 \end{align*}
 Those replacements are done repeatedly until there
 are no occurrences of ``$\lnot$'' left in the formula.
 
 Step (2): We replace each 
 {strict} inequality as follows:
 \[q > 0 \quad \mapsto  \quad (q\cdot y\cdot y-1 = 0),\]
 where $y$ is a new variable.
 Those replacements are done repeatedly till there
 are no occurrences of ``$>$'' left in the formula.
 
 Step (3): We replace all atomic predicates
 of the form 
 $q \geq 0$
 by the predicate
 $q - z^2 = 0$, where $z$ denotes a new variable.

 Step~(4): We delete disjunctions as follows.
 It will also be necessary to replace some conjunctions.
 Let $\Phi$ be the formula
 after Step~(1)--(3).
 Let $\Psi$ be an, initially empty, \etrconjunction formula. In this step, we will describe an algorithm to repeatedly modify $\Phi$ and $\Psi$ in such a way that $\Phi \land \Psi$ stays \equisatisfiable to the initial value of $\Phi$. We will continue these modifications until $\Phi$ consists of just a single equation of the form $p = 0$.

 While $\Phi$ is not of this form, it either contains a disjunction of the form $p = 0 \lor q = 0$, or a conjunction of the form $p = 0 \land q = 0$. The disjunction $p = 0 \lor q = 0$ we can replace by a single equation $p\cdot q = 0$. For a conjunction $p = 0 \land q = 0$, we add new variables $x, y$ and replace it in $\Phi$ by $x \cdot x + y \cdot y = 0$, while we also replace $\Psi$ by $\Psi \land (p-x = 0) \land (q-y = 0)$. Note that in each step, the number of atomic formulas in $\Phi$ is reduced by 1, so we know that the reduction terminates in a linear number of steps. 

 When $\Phi$ consists of just a single polynomial equation, we can replace $\Psi$ by $\Psi \land \Phi$, and we conclude that $\Psi$ is an \etrconjunction formula which is \equisatisfiable to the original $\Phi$.

 At first, it might seem easier %in Case~(ii) 
 to replace
 $p =0 \land q = 0$ by $p\cdot p+q\cdot q = 0$.
 However, we want our reduction to be linear and
 the simplified step could, if done repeatedly,
 lead to very long formulas.
 With the replacement rules we have suggested, the length of the formula increases by at most a constant factor.
 This reduction takes linear time and the
 final formula $\Psi$ is an \etrconjunction formula.
\end{proof}

%%%%%%%%%%%%%%%%%%%%%%%%%%%%%%%%%%%
\subsection{Reduction to Compact Semi-Algebraic Sets}
\label{sec:CompactSets}
%%%%%%%%%%%%%%%%%%%%%%%%%%%%%%%%%%%

In this section, we show the hardness of \etrcompact. 
In that variant, we are promised that the solution space is compact.
To do so, we employ a theorem that states that every solution space is either empty or intersects a large ball.

\begin{definition}
In the problem \etrcompact, we are given an \etrconjunction formula $\Phi$ with the promise that $V(\Phi)$ is compact.
The goal is to decide if $V(\Phi)$ is non-empty.
\end{definition}

We need a tool from 
real algebraic geometry.
The following corollary has been pointed out 
by Schaefer and \v{S}tefankovi\v{c}~\cite{Schaefer-ETR} in a simplified form.
We always use logarithms with base two.

\begin{corollary}
[Basu, Roy~\cite{basu2010bounding} Theorem~2]
\label{cor:BallIntersect}
Let $\Phi$ be an instance of \etr of length $L\geq 4$ such that $V(\Phi)$ is a non-empty subset of $\R^n$.
Let $B$ be the set of points in $\R^n$ at distance at most~$2^{L^{8n}} = 2^{2^{8n\log L}}$ from the origin.
Then $B\cap V(\Phi)\neq\emptyset$.
\end{corollary}

\begin{lemma}
\label{lem:Reduction-Conjunction-Compact}
There is a reduction from \etrconjunction to \etrcompact in $O(L\log L)$ time, where $L$ is the length of the formula.
\end{lemma}
\begin{proof}
 Let an instance $\Phi$ of \etrconjunction be given and define $k = \lceil 8n\log L\rceil$.
 To make an \equisatisfiable formula $\Psi$ such that $V(\Psi)$ is compact, we start by including all the variables and constraints of $\Phi$ in $\Psi$.
We then construct the variables $\var{2^{2^0}},\ldots,\var{2^{2^k}}$, which will always take the values $2^{2^{0}},\ldots, 2^{2^k}$ respectively. We use repeated squaring as follows.
 \begin{align*}
  \var{2^{2^0}} -1-1 &= 0 \\
  \var{2^{2^1}} - \var{2^{2^0}}\cdot \var{2^{2^0}} &= 0 \\
		\vdots & \\
	  \var{2^{2^k}} - \var{2^{2^{k-1}}}\cdot \var{2^{2^{k-1}}} &= 0
 \end{align*}
	For each variable $x$ of $\Phi$, we now introduce the variables $\var{x+2^{2^k}}$ and  $\var{2^{2^k}-x}$
	and the constraints
	\begin{align*}
		\var{x+2^{2^k}}-x-\var{2^{2^k}} & =0 \\
		\var{x+2^{2^k}} & \geq 0 \\
		\var{2^{2^k}-x}-\var{2^{2^k}}+x & =0 \\
		\var{2^{2^k}-x} & \geq 0.
	\end{align*}
Note that this corresponds to introducing the constraint $-2^{2^k}\leq x\leq 2^{2^k}$ in $\Psi$.

Observe that the ball $B$ centered around the origin with radius $2^{2^k}$ is contained in the cube~$\left[-2^{2^k}, 2^{2^k}\right]^n$, so $B \cap V(\Phi) \subseteq V(\Psi)$. It now follows by Corollary~\ref{cor:BallIntersect} that
 $$V(\Phi)\neq\emptyset\Leftrightarrow V(\Psi)\neq\emptyset.$$
 Note that $V(\Psi)$ is compact since $\Psi$ contains no strict inequalities and each variable is bounded.
 This finishes the proof.
\end{proof}

%%%%%%%%%%%%%%%%%%%%%%%%%%%%%%%%%%%
\subsection{Reduction to \etrami}
\label{sec:ami}
%%%%%%%%%%%%%%%%%%%%%%%%%%%%%%%%%%%
In this section we will show \ER-hardness of the problem \etrami.
The term \etrami is an abbreviation for
\textit{{\bf E}xistential {\bf T}heory of the {\bf R}eals with
{\bf A}ddition, {\bf M}ultiplication, and {\bf I}nequalities}.

\begin{definition}[\etrami]
\label{def:etrami}
{We define the set of constraints $C_{\mathrm{AMI}}$ as
\[
C_{\mathrm{AMI}} = \left\{ x + y = z, \; x \cdot y = z, \; x \geq 0, \; x = 1 \right\}
\]
Now we define the \etrami problem as the CCSP given by $C_{\mathrm{AMI}}$.}
\end{definition}

\begin{lemma}[\etrami Reduction]
\label{lem:Reduction-AMI-INV}
Given an instance of \etrcompact defined by a formula $\Phi$, we can in $O(|\Phi|)$ time construct
an  \equisatisfiable \etrami formula $\Psi$ such that $V(\Psi)$ is compact. 
\end{lemma}
\begin{proof}
 Recall that $\Phi$ is a conjunction of atomic formulas of the form $p=0$ for a polynomial~$p$ and $x\geq 0$ for a variable $x$.
 Each polynomial $p$ may contain minuses, zeros, and ones.
 The reduction has four steps.
 In each step, we make changes to $\Phi$.
 In the end, $\Phi$ has become a formula $\Psi$ with the desired properties.
 In step~(1)--(3), we remove unwanted ones, zeros, and minuses by replacing them with constants.
 In step~(4), we eliminate complicated polynomials.

 Step~(1):
 We introduce the constant variable $\var{1}$ and the constraint $\var{1}=1$ to~$\Phi$.
 We then replace all appearances of $1$ with $\var{1}$ in the atomic formulas of the form $p=0$.
 
 Step~(2):
 We introduce the constant variable $\var{0}$ and the constraint $\var{1}+\var{0}=\var{1}$ to~$\Phi$.
 We then replace all appearances of $0$ with $\var{0}$ except in the constraints of the form $x\geq 0$.
 
 Step~(3):
 We introduce the constant variable $\var{-1}$ and the constraint $\var{1}+\var{-1}=\var{0}$ to~$\Phi$.
 We then replace all appearances of minus with a multiplication by $\var{-1}$ in $\Phi$.
 
 Step~(4):
 We replace bottom-up every occurrence 
 of multiplication and addition by a new variable
 and an extra addition or multiplication constraint.
 Here are two examples of such replacements:

 \begin{align*}
  (x_1 + x_2\cdot x_4 + x_5)\cdot x_6 = \var{0} \quad & \mapsto\quad (x_1 + z_1 +x_5)\cdot x_6 = \var{0}\land z_1=x_2\cdot x_4 \\
  (x_1 + z_1 + x_5)\cdot x_6 = \var{0} \quad & \mapsto\quad (z_2 +x_5)\cdot x_6 = \var{0}\land z_2=x_1+ z_1.
 \end{align*}

 In this way, every atomic predicate is eventually transformed to atomic predicates of \etrami or is of the form $x=\var{0}$.
 In the latter case, we replace $x=\var{0}$ by $x+\var{0}=\var{0}$.
 
 To see that the reduction is linear, note that
 every replacement adds a constant to 
 the length of the formula.
 Furthermore, at most linearly many 
 replacements will be done.
 All the above steps preserve the truth value of the formula and the compactness of the solution set.
\end{proof}

%%%%%%%%%%%%%%%%%%%%%%%%%%%%%%%%%%%%%%%
\subsection{Reduction to \etrsmall}
\label{sec:SMALL}
%%%%%%%%%%%%%%%%%%%%%%%%%%%%%%%%%%%%%%%
In this section, we show the hardness of \etrsmall, as defined below. 
The reduction works in two steps.
In the first step, we create a very small number using
repeated squaring and in the second step, 
we scale every variable to be in the correct range.
{\begin{definition}[\etrsmall]
We define the set of constraints $C_{\mathrm{SMALL}}$ as
\[
C_{\mathrm{SMALL}} = \left\{ x + y = z, \; x \cdot y = z, \; x \geq 0, \; x = \tfrac{1}{2} \right\}
\]
Now we define the \etrsmall problem as the CCSP given by $C_{\mathrm{SMALL}}$. Furthermore, {for every instance $\Phi$} we are promised that $V(\Phi) \subseteq \left[-\frac{1}{2}, \frac{1}{2}\right]^n$.
\end{definition}}

We are going to present a reduction from the problem \etrami to \etrsmall.
As a preparation, we present another tool from
real algebraic geometry.
Schaefer~\cite{schaefer2013realizability} made the following simplification of a result from~\cite{basu2010bounding}, which we will use.
More refined statements can be found in~\cite{basu2010bounding}.

\begin{corollary}[\cite{basu2010bounding}]
\label{cor:BallContain}
If a bounded semi-algebraic set in $\R^n$ has bit-complexity at most
$L \geq 5n$, then all its points have distance at most $2^{2^{L+5}}$
from the origin.
\end{corollary}

\begin{lemma}[\etrsmall Reduction]
\label{lem:Reduction-AMI-SMALL}
  Given an \etrami formula $\Phi$ such that $V(\Phi)$ is compact, we can in $O(|\Phi|)$ time construct an  \equisatisfiable instance of \etrsmall.
\end{lemma}

\begin{proof}
Let $\Phi$ be an instance of $\etrami$ with $n$ variables $x_1,\ldots,x_n$.
We construct an instance~$\Psi$ of $\etrsmall$. 

We set $\eps =  2^{-2^{L+6}}$, where $L=|\Phi|$. 
In $\Psi$, we first define a constant variable $\var{\eps}$. 
This is obtained by exponentiation by squaring, using $O(L)$ 
new constant variables and constraints.
We first define $\var{0}$, and $\var{2^{-2^0}}$, i.e.~$1/2$, by the equations
\begin{align*}
\var{2^{-2^0}} & = \tfrac{1}{2} \\
\var{0}+\var{2^{-2^0}} & = \var{2^{-2^0}} \\
\end{align*}
We then use the following equations for all $i\in\{0,\ldots,L+5\}$,
\begin{align*}
\var{2^{-2^{i}}}\cdot\var{2^{-2^{i}}} & =\var{2^{-2^{i+1}}} \\
\end{align*}
Finally, we define $\var{\eps}$ by the constraint $\var{\eps}+\var{0}=\var{2^{-2^{L+6}}}$.

In $\Psi$, we use the variables $\var{\eps x_1},\ldots,\var{\eps x_n}$ instead of
$x_1,\ldots,x_n$.
An equation of $\Phi$ of the form $x=1$ is transformed to the equation
$\var{\eps x} + \var{0}=\var{\eps}$ in $\Psi$.
An equation of $\Phi$ of the form $x+y=z$ is transformed to the equation
$\var{\eps x} + \var{\eps y} = \var{\eps z}$ of $\Psi$.
For an equation of $\Phi$ of the form $x\cdot y=z$, we also introduce a variable $\var{\eps^2 z}$ of $\Psi$ and the equations
\begin{align*}
\var{\eps x} \cdot \var{\eps y} & = \var{\eps^2 z} \\
\var{\eps}\cdot \var{\eps z} & = \var{\eps^2 z}.
\end{align*}
At last, constraints of the form 
$x \geq 0$ become $\var{\eps x} \geq 0$.

We now describe a function $f : V(\Phi) \rightarrow V(\Psi)$ in order to show that $\Phi$ and $\Psi$ are  \equisatisfiable.
Let $\mathbf x= (x_1,\ldots,x_n)\in V(\Phi)$.
In order to define $f$, it suffices to specify the values of the variables of $\Psi$ depending on $\mathbf x$.
 For all the constant variables $\var{2^{-2^0}},\var{2^{-2^1}},\var{2^{-2^2}},\ldots$, we define them in the natural way as $\var{2^{-2^i}} = 2^{-2^i}$ and $\var{\eps} = 2^{-2^{L+6}}$.
For all $i\in\{1,\ldots,n\}$, we now define $\var{\eps x_i} = \eps x_i$ and (when $\var{\eps^2 x_i}$ appears in $\Psi$) $\var{\eps^2 x_i} = \eps^2 x_i$.
Since $\mathbf x$ is a solution to $\Phi$, it follows from the constraints of $\Psi$ that these assignments are a solution to $\Psi$.

We need to verify that $\Psi$ defines an \etrsmall problem, i.e., that $\Psi$ satisfies the promise that $V(\Psi)\subset [-\tfrac{1}{2},\tfrac{1}{2}]^m$, where $m$ is the number of variables of $\Psi$.
To this end, consider an assignment of the variables of $\Psi$ that satisfies all the constraints.
Note first that the constant variables are non-negative and at most $\tfrac{1}{2}$.
For the other variables, we consider the inverse~$f^{-1}$, which is given by the assignment $x_i= \var{\eps x_i}/\eps$ for all $i\in\{1,\ldots,n\}$.
It follows that this yields a solution to~$\Phi$.
Since $V(\Phi)$ is compact, it follows from Corollary~\ref{cor:BallContain} that $|\var{\eps x_i}/\eps|\leq 2^{2^{L+5}}$.
Hence $|\var{\eps x_i}|\leq \eps\cdot 2^{2^{L+5}}= 2^{-2^{L+6}}\cdot 2^{2^{L+5}}\leq \tfrac{1}{2}$.
Similarly, when $\var{\eps^2 x_i}$ is a variable of $\Psi$, we get~$|\var{\eps^2 x_i}|\leq \eps<\tfrac{1}{2}$.

By the existence of $f$ and $f^{-1}$, we have now established that $V(\Phi) \neq \emptyset$ if and only if~$V(\Phi) \neq \emptyset$. In other words, $\Phi$ and $\Psi$ are \equisatisfiable.
The length of $\Psi$ is $O(L)$ longer than the length of $\Phi$, and $\Psi$ can thus be computed in $O(|\Phi|)$ time.
\end{proof}

\subsection{Reduction to \etrsquare}
From here, we can prove the hardness of the following problem:
\begin{definition}[\etrsquare]
We define the set of constraints $C_{\mathrm{SQUARE}}$ as
\[
C_{\mathrm{SQUARE}} = \left\{x+y = z, \; y = x^2, \; x \geq 0, \; x = 1\right\}.
\]
Now we define the \etrsquare problem as the CCSP given by $C_{\mathrm{SQUARE}}$. Furthermore, {for every instance $\Phi$} we are promised that $V(\Phi) \subseteq [-1, 1]^n$.
\end{definition}

\begin{lemma}[\etrsquare Reduction]
\label{lem:Reduction-SMALL-Square}
  Given an instance $\Phi$  of \etrsmall,
we can in $O(|\Phi|)$ time construct an  \equisatisfiable instance of \etrsquare.
\end{lemma}

\begin{proof}
We start with an \etrsmall instance $\Phi$. To this instance we add a variable $\var{1}$ and a constraint $\var{1} = 1$. Next we replace every constraint of the form $x = \frac{1}{2}$ by a constraint $x + x = \var{1}$. Finally, for every constraint of the form $x \cdot y = z$, we introduce the following new variables:
\[
\var{x^2}, \var{y^2}, \var{x+y}, \var{(x+y)^2}, \var{x^2+2xy}, \var{2xy}, \var{xy}
\]
and we add the following constraints:
\begin{align*}
\var{x^2} &= x^2\\
\var{y^2} &= y^2\\
\var{x+y} &= x+y\\
\var{(x+y)^2} &= \var{x+y}^2\\
\var{(x+y)^2} &= \var{x^2+2xy} + \var{y^2}\\
\var{x^2+2xy} &= \var{x^2} + \var{2xy}\\
\var{2xy} &= \var{xy} + \var{xy}\\
\var{xy} &= z.
\end{align*}
Every constraint of the form $x + y = z$ or $x \geq 0$ is not changed. After performing all these changes, which only needs linear time, we have an \etrsquare formula $\Psi$. Furthermore, it can be checked that every solution of this \etrsquare formula corresponds uniquely to a solution of the original \etrsmall formula. Also the fact that $V(\Phi) \subseteq [-1/2, 1/2]^n$ can be seen to imply that $V(\Psi) \subseteq [1, 1]^m$, where $m$ is the number of variables in $\Psi$. 
This proves that $\Psi$ is an \etrsquare instance which is \equisatisfiable to the original \etrsmall instance. Therefore the reduction is valid.
\end{proof}

%%%%%%%%%%%%%%%%%%%%%%%%%%%%%%%%%%%%%%%%
\section{Proof of CCSP-Theorems} 
\label{sec:main}
%%%%%%%%%%%%%%%%%%%%%%%%%%%%%%%%%%%%%%%%
In this section, we will prove \cref{thm:Equality} and  \Cref{thm:Inequality}. 

%%%%%%%%%%%%%%%%%%%%%%%%%%%%%%%%%%%%%%
\subsection{Approximate Solutions}
\label{sub:ApproximateSolutions}
%%%%%%%%%%%%%%%%%%%%%%%%%%%%%%%%%%%%%%

We start by proving a lemma which plays an important role in the proofs of \cref{thm:Equality} and  \Cref{thm:Inequality}.
 The following lemma intuitively states the following: if an \etrsquare formula~$\Phi$ has something which is ``almost a solution'', with an error of at most $2^{-2^{O(|\Phi|)}}$, then $\Phi$ also admits an actual solution. Similar results were established in~\cite{Argyrios2022_AppriximatinETR}.
\begin{definition}
Let $\Phi = \Phi(x_1, \ldots, x_n)$ be an \etrsquare formula such that $V(\Phi) \subseteq [-1, 1]^n$. For $\eps \geq 0$, define $\Phi_\eps$ as the formula where every constraint of the form $y = x^2$ is replaced by a constraint of the form $|y-x^2| \leq \eps$, and where constraints $-1 \leq x \leq 1$ are added for every $x \in \{x_1, \ldots, x_n\}$.
\end{definition}

\begin{lemma} 
\label{lemma:approx}
Let $\Phi = \Phi(x_1, \ldots, x_n)$ be an \etrsquare formula such that $V(\Phi) \subseteq [-1, 1]^n$. 
Now there exists an $M \in \Z$ with $M = O(|\Phi|)$ and $\eps = 2^{-2^M}$, 
such that $\Phi$ and $\Phi_\eps$ are \equisatisfiable.
\end{lemma}

We use the following result from Schaefer and \v{S}tefankovi\v{c} \cite{Schaefer-ETR}, see also~\cite{jeronimo2010minimum}:
\begin{corollary}[Corollary 3.4 from \cite{Schaefer-ETR}] \label{cor:pos_dist}
If two semi-algebraic sets in $\R^n$ each of bit-complexity at most $L \geq 5n$ have positive distance (for example, if they are disjoint and compact), then that distance is at least $2^{-2^{L+5}}$.
\end{corollary}
Here, the distance between two subsets $X, Y \subseteq \R^n$ is defined as $\inf_{x\in X, y \in Y}d(x, y)$. Note that in the case where $X$ and $Y$ are compact, the infimum in this definition may be replaced by a minimum.
\begin{proof}[Proof of Lemma~\ref{lemma:approx}]
First note that for any $\eps \geq 0$, we have $V(\Phi) \subseteq V(\Phi_\eps)$. In particular, if $V(\Phi)$ is nonempty, then also any $V(\Phi_\eps)$ is nonempty. For the rest of the proof, we will assume that $V(\Phi)$ is empty, and construct $M$ and $\eps = 2^{-2^M}$ such that $V(\Phi_\eps)$ is also guaranteed to be empty.

Suppose that $\Phi$ has $n$ variables, and contains $r$ constraints of the form $y = x^2$. For every $\varepsilon \geq 0$, we define $\Phi_\varepsilon'$ as the formula on variables $x_1, \ldots, x_n, \eta_1, \ldots, \eta_r$ obtained from $\Phi$ by replacing every constraint of the form $y = x^2$ by constraints $y = x^2 + \eta_i$ and $-\varepsilon \leq \eta_i \leq \varepsilon$, and where constraints $-1 \leq x \leq 1$ are added for every $x \in \{x_1, \ldots, x_n\}$. Note that there is a natural bijection between $V(\Phi_\varepsilon)$ and $V(\Phi_\varepsilon')$ for every $\varepsilon \geq 0$. Since we assumed $V(\Phi) = \emptyset$, it also follows that $V(\Phi_0') = \emptyset$. We furthermore define $\Phi_\infty'$ in the same way, except that we drop the constraints of the form $- \varepsilon \leq \eta_i \leq \varepsilon$. Observe that $V(\Phi_\infty')$ is bounded: the fact that every variable $x_i$ is bounded by 1 in absolute value, implies that every variable $\eta_i$ is bounded by 2 in absolute value. In particular $V(\Phi_\infty')$ is compact.

Next define $\Psi$ to be the formula on the same variables $x_1, \ldots x_n, \eta_1, \ldots, \eta_r$ which enforces $\eta_i = 0$ for all $1 \leq i \leq r$ and $-1 \leq x_i \leq 1$ for all $1 \leq i \leq n$. Note that $V(\Phi_\infty') \cap V(\Psi) = V(\Phi_0') = \emptyset$, and furthermore that both $V(\Phi_\infty')$ and $V(\Psi)$ are compact and nonempty. This means that we can apply \cref{cor:pos_dist}. Let $L$ be the maximum of $5n$ and the bit complexities of $\Phi_\infty'$ and $\Psi$. Note that $L$ is linear in the length of $\Phi$. 
We conclude that the distance between $V(\Phi_\infty')$ and $V(\Psi)$ is at least $2^{-2^{L+5}}$, by \Cref{cor:pos_dist}.

Recall that $r$ is the number of constraints of the form $y = x^2$ in $\Phi$.
Setting $M = L + 6$, it can be shown that $r \cdot 2^{-2^M} < 2^{-2^{L+5}}$.
Specifically, $M$ is linear in the length of $\Phi$.
We take $\varepsilon = 2^{-2^M}$. 
Suppose, for the purpose of contradiction, that $V(\Phi_\varepsilon) \neq \emptyset$, and therefore also $V(\Phi_\varepsilon') \neq \emptyset$. Let $P \in V(\Phi_\varepsilon')$, and let $P'$ be the point we obtain by setting all the $\eta_i$-coordinates of $P$ to 0. Now $P'$ is contained in $V(\Psi)$. Since every $\eta_i$-coordinate of $P$ was bounded by $\varepsilon$, the distance between~$P$ and $P'$ is at most $r\varepsilon$, therefore $P$ can be seen to have distance at most $r\varepsilon < 2^{-2^{L+5}}$ to $V(\Psi)$. Furthermore, $P$ is also contained in $V(\Phi_\infty')$. This implies that $V(\Phi_\infty')$ has distance smaller than~$2^{-2^{L+5}}$ to $V(\Psi)$. This contradicts the result from applying \Cref{cor:pos_dist}.

We conclude that indeed $V(\Phi) = \emptyset$ implies $V(\Phi_\varepsilon) = \emptyset$. This completes the proof of the lemma.
\end{proof}

%%%%%%%%%%%%%%%%%%%%%%%%%%%%%%%%%%%%%%%%%%
\subsection{Almost Square Explicit Equality Constraints}
\label{sec:SquareExplicit}
%%%%%%%%%%%%%%%%%%%%%%%%%%%%%%%%%%%%%%%%%%
Using \Cref{lemma:approx}, we are able to prove that an explicit version \CE is also \ER-complete, with some additional assumptions.
Note that this subsection is technically not needed 
for the proof of \Cref{thm:Equality} and \Cref{thm:Inequality}.
We will prove a similar lemma also for the inequality case.
And the inequality case can be used to also prove the 
equality case.
Yet, we believe that seeing the proof first for the equality case
makes it much easier to read \Cref{sec:Implicit}.

\begin{definition}[\CEEXPL] \label{def:ceexpl}
    Let $U\subseteq \R$ and
    let $f \colon \domain \to \R$ be a function. 
    We define the \CEEXPL problem to be the \CE problem corresponding to the function $f^* \colon U \times \R  \to \R$ defined by~$f^*(x, y) = y - f(x)$.
\end{definition}

Note that for this definition of $f^*$, we have
\[
\equalzero{f^*} = \set{(x, y) \in \R \times \R}{x \not\in U \lor y = f(x)}.
\]
In particular, this means that if we know that a variable $x$ is forced to lie in $U$, then $(x,y) \in \equalzero{f^*}$ will exactly imply that $y = f(x)$. In what follows, we will ensure we are in the case where all variables are contained in $U$, so this enables us to enforce constraints of the form $y = f(x)$ on these variables, while also implying that the constructed instance is \domainadherent.

The goal of this section is to prove the following result:
\begin{restatable}{lemma}{fETRweak} \label{lemma:delta-f-ETR}
Let $\domain \subseteq \R$ be a neighborhood of $0$, and let $f \colon \domain \to \R$ be a function such that $|f(x) - x^2| \leq \frac{1}{10}|x|^3$ for all $x \in \domain \subseteq \R$. 
Let $T$ be strictly bounded from above by $\frac{1}{4}$ and \nicelyComputable. Furthermore, assume that the interval $[-T(n), T(n)]$ is contained in $U$ for all $n$. In this setting, \CEEXPL{} is \ER-hard, even when only considering instances where $\delta = T(n)$, with $n$ being the number of variables.
\end{restatable}

The reason that we impose these specific constraints on $f$, which enforce $f$ to be very similar to squaring, is that the proof will use \ER-hardness of a problem involving a squaring constraint. Furthermore, this specific case can be generalized to more general functions $f$.

\begin{proof}
Before giving the details of the construction, we will first give an overview of the used approach. The idea of this proof is to start with an instance of \etrsquare, and convert this into a \CEEXPL{} instance by using~$f$ to approximate squaring. In order to ensure that~$f$ approximates squaring close enough, the whole instance is scaled by some small factor~$\eps$, so every variable $x$ is replaced by a variable representing $\eps x$ instead.

The linear constraints and inequalities are easy to rewrite in terms of $\eps x$, for example a constraint of the form $x + y = z$ can be rewritten to $\eps x + \eps y = \eps z$.

Handling a squaring constraint $y = x^2$ is a bit more complicated. The first step is to rewrite this to a constraint involving $\eps x$ and $\eps y$, we get $ \eps \cdot \eps y = (\eps x)^2$. However, in the \CEEXPL{} problem there is no easy way to simulate the multiplication on the left-hand side of this equation. To solve this, we rewrite this equation to only use sums and differences of squares:
\[
\eps^2 + 2(\eps x)^2 + (\eps y)^2 - (\eps + \eps y)^2 = 0.
\]
To simplify notation a bit, we will denote $t_1 = \eps$, $t_2 = \eps x$, $t_3 = \eps y$ and $t_4 = \eps + \eps y$. Using this notation the equation becomes $t_1^2 + 2t_2^2 + t_3^2 - t_4^2 = 0$. This is still not something we can directly enforce in a \CE{} formula. However, by applying the function $f$, squaring can be approximated. Furthermore, \cref{lemma:approx} on a high level implies that such an approximation is enough to guarantee the existence of a solution to the original equations. This is why in the \CE{} formulation we enforce
\begin{equation}
f(t_1) + 2f(t_2) + f(t_3) - f(t_4) = O(\eps^3). \label{eq:f-ETR_squaring}
\end{equation}
Enforcing the $= O(\eps^3)$ presents another problem: we cannot easily compute $\eps^3$. To counter this, we instead bound the left-hand side of the equation in absolute value by $2(f(\eps + f(\eps)) - f(\eps))$, which is approximately equal to $4\eps^3$ (note that in the case that $f(x) = x^2$ for all $x$, this expression would actually be equal to $4\eps^3 + 2\eps^4$).
The details of this reduction and a proof of its correctness will be worked out in the remainder of this proof.

\paragraph{Reduction.}
Let~$\Phi$ be an \etrsquare formula. We will now construct a \CE{} formula~$\Psi$ such that $V(\Phi) \neq \emptyset$ if and only if $V(\Psi) \neq \emptyset$. The value of $\delta$ will be defined at the end of the construction as $T(n)$, with $n$ the final number of variables. Since the construction itself does not depend on the exact value of $\delta$, this does not cause any issues. The only important property is that $\delta < \frac{1}{4}$.

Let~$M$ be the constant obtained by applying \cref{lemma:approx} to~$\Phi$, and let~$L$ be the smallest positive integer such that $2^{-2^L} \leq \frac{1}{100} \cdot 2^{-2^M}$ and $L\geq 3$. 
We start by introducing variables~$\var{\delta_i}$ for $0 \leq i \leq L$. The variable~$\var{\delta_0}$ satisfies the constraint $\var{\delta_0} = \delta$, and for each $1 \leq i \leq L$ we add a constraint
\[
\var{\delta_i} = f(\var{\delta_{i-1}}).
\]
Denote $\var{\eps} = \var{\delta_L}$. The idea behind these definitions is to simulate repeated squaring, as we will see later they force the value of~$\var{\eps}$ to be in the interval $\left(0, 2^{-2^L}\right]$.

Next, we introduce a new variable $\var{\approx 2\eps^3}$ together with a (constant) number of constraints and auxiliary variables that enforce
\[
\var{\approx 2\eps^3} = f(\var{\eps} + f(\var{\eps})) - f(\var{\eps}).
\]
This can be done explicitly by introducing auxiliary variables $\var{f(\eps)}$, $\var{\eps + f(\eps)}$ and $\var{f(\eps + f(\eps))}$ and adding the following constraints:
\begin{align*}
\var{f(\eps)} &= f(\var{\eps})\\
\var{\eps + f(\eps)} &= \var{\eps} + \var{f(\eps)}\\
\var{f(\eps + f(\eps))} &= f(\var{\eps + f(\eps)})\\
\var{f(\eps + f(\eps))} &= \var{\approx 2\eps^3} + \var{f(\eps)}.
\end{align*}
In the rest of this proof, and in future proofs of this paper, we will not give such explicit constraints anymore. The variable $\var{\approx 2\eps^3}$ will be used to bound the error on the constraints replacing squaring constraints, as mentioned in the overview of this proof. Stated differently, it replaces the ``$= O(\eps^3)$'' part of \cref{eq:f-ETR_squaring}.

Now, for each variable~$x$ of~$\Phi$, we add a variable~$\var{\eps x}$ to~$\Psi$, together with some constraints which enforce that $-\var{\eps} \leq \var{\eps x} \leq \var{\eps}$. Furthermore each constraint $x + y = z$ is replaced by $\var{\eps x} + \var{\eps y} = \var{\eps z}$, each constraint $x \geq 0$ is replaced by $\var{\eps x} \geq 0$ and each constraint $x = 1$ is replaced by $\var{\eps x} = \var{\eps}$.

For each constraint $y = x^2$, we build \cref{eq:f-ETR_squaring} as in the overview. To do this, we first introduce variables $\var{t_1}$, $\var{t_2}$, $\var{t_3}$ and $\var{t_4}$ satisfying
\begin{align*}
\var{t_1} &= \var{\eps}\\
\var{t_2} &= \var{\eps x}\\
\var{t_3} &= \var{\eps y}\\
\var{t_4} &= \var{\eps} + \var{\eps y}.
\end{align*}
(Note that, even though $x$ and $y$ are suppressed in the notation here, the variables $\var{t_1}$, $\var{t_2}$, $\var{t_3}$ and $\var{t_4}$ should actually be distinct variables for each constraint $y = x^2$.) Next we introduce a new variable~$\var{\eta_{x, y}}$ representing the left-hand side of \cref{eq:f-ETR_squaring}:
\[
\var{\eta_{x, y}} = 
f(\var{t_1}) + 2f(\var{t_2}) + f(\var{t_3}) - f(\var{t_4}).
\]

The next step is to bound this variable, for this we add constraints which enforce
\begin{align*}
\var{\eta_{x, y}} &\geq -2\var{\approx 2\eps^3}\quad \text{and}\\
\var{\eta_{x, y}} &\leq 2\var{\approx 2\eps^3}.
\end{align*}
This completes the construction of~$\Psi$. Now we can take $\delta$ to equal $T(n)$, where $n$ is the number of variables in $\Psi$. Note that in this construction, every constraint of $\Phi$ is replaced by a constant number of constraints in $\Psi$, and therefore we have $|\Psi| = O(|\Phi|)$. In particular this reduction can be executed in linear time.

\paragraph{Calculations.}
To show the validity of the reduction, we first perform some side calculations. We define $\delta_0 = \delta$, for $1 \leq i \leq L$ we take $\delta_i = f(\delta_{i-1})$, and we take $\eps = \delta_L$. We deduce the following facts:
\begin{align}
|f(x) - x^2| &\leq \frac{1}{10}|x|^3 & \text{for } x \in [-\delta, \delta] \setminus \{0\} \label[ineq]{fact:f_approx_square}\\
0 &< f(x) \leq 2x^2 & \text{for } x \in [-\delta, \delta] \setminus \{0\} \label[ineq]{fact:f_bounds}\\
\eps &\leq \frac{1}{100} \min\left(2^{-2^M}, \delta\right)& \label[ineq]{fact:eps_bound}\\
f(\eps) &< \eps& \label[ineq]{fact:feps_bound}\\
f(\eps + f(\eps)) - f(\eps) &\in [\eps^3, 3\eps^3]&\label[ineq]{fact:eps_cubed_approx}
\end{align}
\Cref{fact:f_approx_square} is one of the assumptions and is repeated here just for clarity. Combining this with~$\delta \leq \frac{1}{4}$, \cref{fact:f_bounds} follows.

Using induction with the fact that $0 < f(x) \leq 2x^2$ for $x \in [-\delta, \delta] \setminus \{0\}$, it follows that $0 < \delta_i \leq 2^{-2^i-1}$ for all~$i$, so $0 < \eps \leq \frac{1}{2}2^{-2^L}$. 

Using the definition of~$L$, we get that $\eps \leq \tfrac{1}{100}2^{-2^M}$.
Using that $L\geq 3$ and $\delta \leq 1/4$, we get that $\eps \leq \delta^{2^3} = \delta^{16} \leq \tfrac{1}{100}\delta$.
Together this implies \cref{fact:eps_bound}.

 The fact $f(\eps) < \eps$ now follows from \cref{fact:f_bounds} with the observion that $\eps < \min(\frac{1}{2}, \delta)$, which follows from \cref{fact:eps_bound}.

For deriving \cref{fact:eps_cubed_approx}, we first rewrite by adding and subtracting some terms, and applying the triangle inequality:
\begin{align*}
|f(\eps + f(\eps)) - f(\eps) - 2\eps^3| &= |f(\eps + f(\eps)) - (\eps + f(\eps))^2 + \eps^2 - f(\eps)\\
&\quad \; + (f(\eps) + \eps^2)(f(\eps) - \eps^2) + \eps^4 + 2\eps(f(\eps) - \eps^2)|\\
&\leq |f(\eps + f(\eps)) - (\eps + f(\eps))^2| + |\eps^2 - f(\eps)|\\
&\quad \; + (f(\eps) + \eps^2)|f(\eps) - \eps^2| + \eps^4 + 2\eps|f(\eps) - \eps^2|.
\end{align*}
To this we apply \cref{fact:f_approx_square,fact:eps_bound,fact:feps_bound} to obtain the desired bound, where in particular we use that 
\cref{fact:eps_bound} implies $\eps < \frac{1}{100}$:
\begin{align*}
|f(\eps + f(\eps)) - f(\eps) - 2\eps^3|
&\leq |f(\eps + f(\eps)) - (\eps + f(\eps))^2| + |\eps^2 - f(\eps)|\\
&\quad \; + (f(\eps) + \eps^2)|f(\eps) - \eps^2| + \eps^4 + 2\eps|f(\eps) - \eps^2|\\
&\leq \frac{1}{10}(\eps+f(\eps))^3 + \frac{1}{10}\eps^3 + \frac{1}{10}(f(\eps) + \eps^2)\eps^3 + \eps^4 + \frac{1}{5}\eps^4\\
&\leq \frac{8}{10}\eps^3 + \frac{1}{10}\eps^3 + \frac{1}{10}\eps^3\\
&\leq \eps^3,
\end{align*}
so $f(\eps + f(\eps)) - f(\eps) \in [\eps^3, 3\eps^3]$.

\paragraph{$V(\Phi)$ nonempty implies $V(\Psi)$ nonempty.}
Now we can start to prove the validity of the reduction. First suppose that $V(\Phi) \neq \emptyset$, so there is some $P \in V(\Phi)$. It needs to be shown that also $V(\Psi) \neq \emptyset$, to do this we construct a point $Q \in V(\Psi)$. For a variable~$x$ of~$\Phi$, we will use the notation~$x(P)$ for the value of this variable for the point~$P$. A similar notation is used for~$Q$. To define~$Q$, we take $\var{\eps x}(Q) = \eps x(P)$ for all variables~$x$ of~$\Phi$, and we enforce that~$Q$ satisfies all equality constraints of~$\Psi$. This uniquely defines the value of~$Q$ in all other variables of~$\Psi$. In particular, we get that 
\begin{align*}
\var{\eps}(Q) &= \eps\\
\var{\approx 2\eps^3}(Q) &= f(\eps + f(\eps)) - f(\eps)\\
\var{\eta_{x, y}}(Q) &= f(\eps) + 2f(\eps x(P)) + f(\eps y(P)) - f(\eps + \eps y(P)),
\end{align*}
where the last equality holds for all constraints $y = x^2$ in~$\Phi$.

It is left to show that~$Q$ also satisfies all inequalities of~$\Psi$. 
There are three types of these inequalities. 
Firstly, we have inequalities which enforce $|\var{\eps x}(Q)| \leq \var{\eps}(Q)$. 
That these are satisfied for~$Q$ follows from the fact that $|x(P)| \leq 1$ since~$\Phi$ is an 
\etrsmall
formula. 
Secondly, for every inequality $x \geq 0$ in~$\Phi$, we get a corresponding inequality $\var{\eps x} \geq 0$, that this is satisfied follows by combining $\var{\eps x}(Q) = \eps x(P)$ and $x(P) \geq 0$.

Finally, for every constraint $y = x^2$ in~$\Phi$ we get constraints enforcing $|\var{\eta_{x, y}}| \leq 2\var{\approx 2\eps^3}$. To see that these are satisfied, first we shorten the notation a bit by writing $t_1 = \var{t_1}(Q) = \eps$, $t_2 = \var{t_2}(Q) = \eps x(P)$, $t_3 = \var{t_3}(Q) = \eps y(P)$ and $t_4 = \var{t_4}(Q) = \eps + \eps y(P)$. Now $\var{\eta_{x, y}}(Q)$ can be bounded, for this we first use the triangle inequality:
\begin{align*}
|\var{\eta_{x, y}}(Q)| &= |f(t_1) + 2f(t_2) + f(t_3) - f(t_4)|\\
&= |f(t_1)-t_1^2 + 2(f(t_2)-t_2^2) + f(t_3)-t_3^2 - (f(t_4) - t_4^2)\\
&\quad \; + t_1^2 + 2t_2^2 + t_3^2 - t_4^2|\\
&\leq |f(t_1)-t_1^2| + 2|f(t_2)-t_2^2| + |f(t_3)-t_3^2| + |f(t_4) - t_4^2|\\
&\quad \; + |t_1^2 + 2t_2^2 + t_3^2 - t_4^2|\\
\end{align*}
Note that $t_1, t_2, t_3$ and $t_4$ were chosen in such a way to ensure that, given $y(P) = x(P)^2$, we have $t_1^2 + 2t_2^2 + t_3^2 - t_4^2 = 0$. Using this together with \cref{fact:f_approx_square} and the facts that $t_1$, $t_2$ and $t_3$ are bounded in absolute value by $\eps$, and $t_4$ is bounded in absolute value by $2\eps$, we find
\begin{align*}
|\var{\eta_{x, y}}(Q)| &\leq |f(t_1)-t_1^2| + 2|f(t_2)-t_2^2| + |f(t_3)-t_3^2| + |f(t_4) - t_4^2|\\
&\quad \; + |t_1^2 + 2t_2^2 + t_3^2 - t_4^2|\\
&\leq \frac{1}{10}|t_1|^3 + \frac{1}{5}|t_2|^3 + \frac{1}{10}|t_3|^3 + \frac{1}{10}|t_4|^3 + 0\\
&\leq \left(\frac{1}{10} + \frac{1}{5} + \frac{1}{10}\right)\varepsilon^3 
+ \frac{1}{10}(2\varepsilon)^3\\
&\leq \frac{6}{5} \eps^3.
\end{align*}
Finally using \cref{fact:eps_cubed_approx} we derive
\begin{align*}
|\var{\eta_{x, y}}(Q)| &\leq \frac{6}{5} \eps^3\\
&\leq 2(f(\eps + f(\eps)) -f(\eps)).
\end{align*}
This completes the proof that $Q \in V(\Psi)$, so $V(\Psi) \neq \emptyset$.

\paragraph{$V(\Psi)$ nonempty implies $V(\Phi)$ nonempty.}
Next, suppose that there is some $Q \in V(\Psi)$. Now we want to show that $|x(Q)| \leq \delta$ for all variables~$x$ in~$\Psi$, and we want to prove that $V(\Phi) \neq \emptyset$. 
We start by bounding the coordinates. 
Note that the values $\var{\delta_i}(Q)$ can inductively be shown to be smaller than~$\delta$. Here we use that $\var{\delta_i}(Q)$ being smaller than $\delta$ implies that $\var{\delta_i}(Q) \in U$, so the constraint $\var{\delta_{i+1}} = f(\var{\delta_i})$ actually enforces $\var{\delta_{i+1}}(Q) = f(\var{\delta_i}(Q))$.
From this it follows that $\var{\eps}(Q) = \eps$. Consequently, for every variable $x$ in $\Phi$, it can be inferred that $|\var{\eps x}(Q)| \leq |\var{\eps}(Q)| \leq \eps$.
Using this, it can be shown that also all values of the auxiliary variables except for the~$\var{\delta_i}$ are bounded by $100 \eps \leq \delta$. So this shows that~$Q$ is contained in $[-\delta, \delta]^{n}$, where~$n$ is the number of variables of~$\Psi$. This also implies that $Q$ is \domainadherent, since $[-\delta, \delta] \subseteq U$.

Now we need to show that $V(\Phi) \neq \emptyset$. 
We apply \cref{lemma:approx} to show this. We will construct a point $P$ within $V(\Phi_{100\eps})$. This construction implies that $V(\Phi)$ is non-empty, given that $100 \eps \leq 2^{-2^M}$.
We define the point~$P$ by taking $x(P) = \frac{\var{\eps x}(Q)}{\eps}$ for all variables~$x$ of~$\Phi$. It immediately follows that~$P$ satisfies all linear constraints and inequality constraints of~$\Phi$, and it is only left to check that it satisfies the constraints $|y - x^2| \leq 100\eps$ of~$\Phi_{100\eps}$. To do this, first, we observe that, using \cref{fact:eps_cubed_approx},
\begin{align*}
|\var{\eta_{x, y}}(Q)| &\leq 2\var{\approx 2\eps^3}(Q)\\
&= 2(f(\eps + f(\eps)) - f(\eps))\\
&\leq 6\eps^3.
\end{align*}
Now we will try to bound $|y(P) - x(P)^2|$. First we again write $t_1 = \var{t_1}(Q) = \eps$, $t_2 = \var{t_2}(Q) = \eps x(P)$, $t_3 = \var{t_3}(Q) = \eps y(P)$ and $t_4 = \var{t_4}(Q) = \eps + \eps y(P)$. These choices were made such that
\[
t_1^2 + 2t_2^2 + t_3^2 - t_4^2 = 2\eps^2(x(P)^2-y(P)),
\]
so we see
\[
2\eps^2|y(P) - x(P)^2| = \left|t_1^2 + 2t_2^2 + t_3^2 - t_4^2\right|.
\]
Next we apply the triangle inequality to get an expression to which we can apply \cref{fact:f_approx_square} and the bound on $|\var{\eta_{x, y}}(Q)|$:
\begin{align*}
2\eps^2|y(P) - x(P)^2| &= \left|t_1^2 + 2t_2^2 + t_3^2 - t_4^2\right|\\
&= |t_1^2-f(t_1) + 2(t_2^2-f(t_2)) + t_3^2-f(t_3) - (t_4^2-f(t_4))\\
&\quad \; + f(t_1) + 2f(t_2) + f(t_3) - f(t_4)|\\
&\leq |t_1^2-f(t_1)| + 2|t_2^2-f(t_2)| + |t_3^2-f(t_3)| + |t_4^2-f(t_4)|\\
&\quad \; + |f(t_1) + 2f(t_2) + f(t_3) - f(t_4)|\\
&= |t_1^2-f(t_1)| + 2|t_2^2-f(t_2)| + |t_3^2-f(t_3)| + |t_4^2-f(t_4)|\\
&\quad \; + |\var{\eta_{x, y}}(Q)|
\end{align*}
Applying \cref{fact:f_approx_square} and the bound on $|\var{\eta_{x, y}}(Q)|$ yields
\begin{align*}
2\eps^2|y(P) - x(P)^2| &\leq |t_1^2-f(t_1)| + 2|t_2^2-f(t_2)| + |t_3^2-f(t_3)| + |t_4^2-f(t_4)|\\
&\quad \; + |\var{\eta_{x, y}}(Q)|\\
&\leq \frac{1}{10}t_1^3 + \frac{1}{5}t_2^3 + \frac{1}{10}t_3^3 + \frac{1}{10}t_4^3 + 6\eps^3
\end{align*}
Finally we use that $t_1$, $t_2$ and $t_3$ are bounded in absolute value by $\eps$, and that~$t_4$ is bounded by~$2\eps$:
\begin{align*}
2\eps^2|y(P) - x(P)^2| &\leq \frac{1}{10}t_1^3 + \frac{1}{5}t_2^3 + \frac{1}{10}t_3^3 + \frac{1}{10}t_4^3 + 6\eps^3\\
&\leq \frac{1}{10}\eps^3 + \frac{1}{5}\eps^3 + \frac{1}{10}\eps^3 + \frac{8}{10}\eps^3 + 6\eps^3\\
&< 200\eps^3.
\end{align*}
So $|y(P) - x(P)^2| \leq 100\eps$. This proves that $P \in V(\Phi_{100\eps})$, and therefore $V(\Phi) \neq \emptyset$.

This finishes the proof of the validity of the reduction of \etrsquare to \CEEXPL. We conclude that for the given~$f$, the problem \CEEXPL{} is \ER-hard.
\end{proof}

%%%%%%%%%%%%%%%%%%%%%%%%%%%%%%%%%%%%%%%%%%
\subsection{Almost Square Explicit Inequality Constraints}
\label{sec:SquareExplicitInequality}
%%%%%%%%%%%%%%%%%%%%%%%%%%%%%%%%%%%%%%%%%%
In this section, we will prove a number of hardness results about the explicit version of \CCI. Before we can describe these results, we first need the following definition:

\begin{definition}[\CCIEXPL]
Let $\domain \subseteq \R$ and let $f, g \colon \domain \to \R$ be two functions. Now we define the \CCIEXPL problem to be the \CCI problem corresponding to the functions $f^*, g^* \colon \domain \times \R \to \R$ defined by $f^*(x, y) = y - f(x)$ and $g^*(x, y) = g(x) - y$.
\end{definition}

Note that, just as in \cref{def:ceexpl}, the constraints \largerzero{f^*} and \largerzero{g^*} in this definition can be used to enforce $y \geq f(x)$ and $y \leq g(x)$ if we already know that $y \in U$.
We will prove that \CCIEXPL is $\ER$-hard in a large number of cases. In particular, we prove the following:
\begin{restatable}{corollary}{CCIEXPLdifferentiableCOR} \label{cor:fgETRINEQEXPL_differentiable}
Let $\domain \subseteq \R$ be a neighborhood of $0$, and let $f, g \colon \domain \to \R$ be functions which are three times differentiable such that $f(0) = g(0) = 0$ and $f'(0), f''(0), g'(0), g''(0) \in \Q$ with $f''(0), g''(0) > 0$. 
Let $T$ be bounded and nicely computable. In this setting, \CCIEXPL{} is \ER-hard, even when considering only instances where $\delta = T(n)$, with $n$ being the number of variables.
\end{restatable}

This corollary will be an important step towards proving \cref{thm:Equality}. 
Before we can prove this corollary, we first prove another result which is similar to \cref{lemma:delta-f-ETR}.

\begin{lemma} \label{lemma:CCIEXPL_fg_approx_sqr}
Let $\domain \subseteq \R$ be a neighborhood of 0, and let $f, g \colon \domain \to \R$ be functions such that $|f(x) - x^2| \leq \frac{1}{10}|x|^3$ and $|g(x) - x^2| \leq \frac{1}{10}|x|^3$ for all $x \in \domain$. 
Let $T$ be strictly bounded from above by $\frac{1}{4}$ and nicely computable. Furthermore, assume that the interval $[-T(n), T(n)]$ is contained in $\domain$ for all $n$. In this setting, \CCIEXPL{} is \ER-hard, even when only considering instances where $\delta = T(n)$, with $n$ being the number of variables.
\end{lemma}
\begin{proof}
The idea is to use almost the same construction as in \cref{lemma:delta-f-ETR}, so we recommend the reader to first read the proof of this lemma. The first main difference is that some extra care needs to be taken when making the constraints for the~$\var{\delta_i}$ variables. 
Also, the squaring constraints need to be handled in a slightly different way. In order to do this, we replace the variables~$\var{\eta_{x, y}}$ by two new variables~$\var{\eta_{x, y}^{\mathrm{low}}}$ and~$\var{\eta_{x, y}^{\mathrm{high}}}$, which impose a lower bound, respectively upper bound, on the value of $t_1^2 + 2t_2^2 + t_3^2 - t_4^2$. Here $t_1 = \eps$, $t_2 = \eps x$, $t_3 = \eps y$ and $t_4 = \eps + \eps y$, as before.

\paragraph{Reduction.}
Let~$\Phi$ be an \etrsquare formula. We will construct a \CCIEXPL{} formula~$\Psi$ such that $V(\Phi) \neq \emptyset$ if and only if $V(\Psi) \neq \emptyset$. Again we will take $\delta = T(n) < \frac{1}{4}$ at the end of the construction, where $n$ is the final number of variables. Let~$M$ be the constant obtained by applying \cref{lemma:approx} to~$\Phi$, and let~$L$ be a constant such that $2^{-2^L} \leq \frac{1}{100} \cdot 2^{-2^M}$ and $L \geq 3$, just like in the proof of \cref{lemma:delta-f-ETR}.

We again introduce~$\var{\delta_i}$ for $0 \leq i \leq L$, where the variable~$\var{\delta_0}$ should satisfy the constraint $\var{\delta_0} = \delta$. For each $1 \leq i \leq L$ we now add constraints enforcing 
\[
\frac{1}{2}f(\var{\delta_{i-1}}) \leq \var{\delta_i} \leq g(\var{\delta_{i-1}}).
\]
Denote $\var{\eps} = \var{\delta_L}$. The constraints $\var{\delta_i} \leq g(\var{\delta_{i-1}})$ are there to enforce that $\var{\eps} \leq 2^{-2^L}$, and the constraints $\frac{1}{2}f(\var{\delta_{i-1}}) \leq \var{\delta_i}$ are there to enforce that $\var{\eps} > 0$.

We continue by defining variables $\var{\leq g(\eps)}$ and $\var{\lesssim 2\eps^3}$ using constraints
\begin{align*}
\var{\leq g(\eps)} &\leq g(\var{\eps}),\\
\var{\leq g(\eps)} &\geq 0,\\
\var{\lesssim 2\eps^3} &\leq g(\var{\eps} + \var{\leq g(\eps)}) - f(\var{\eps}),\\
\var{\lesssim 2\eps^3} &\geq 0.
\end{align*}

This new variable $\var{\lesssim 2\eps^3}$ is a replacement for the variable $\var{\approx 2\eps^3}$ which occurred in the proof of \cref{lemma:delta-f-ETR}. Later, we will show that $\var{\lesssim 2\eps^3}$ is upper bounded by $3\eps^3$.

Next for each variable~$x$ of~$\Phi$, we add a variable~$\var{\eps x}$ to~$\Psi$, with constraints enforcing $-\var{\eps} \leq \var{\eps x} \leq \var{\eps}$. Constraints of type $x+y = z$, type $x \geq 0$ or type $x = 1$ are handled by replacing them by constraints $\var{\eps x} + \var{\eps y} = \var{\eps z}$, $\var{\eps x} \geq 0$ and $\var{\eps x} = \var{\eps}$, respectively.

For each constraint $y = x^2$, we introduce variables $\var{t_1}$, $\var{t_2}$, $\var{t_3}$ and $\var{t_4}$ with constraints
\begin{align*}
\var{t_1} &= \var{\eps},\\
\var{t_2} &= \var{\eps x},\\
\var{t_3} &= \var{\eps y},\\
\var{t_4} &= \var{\eps} + \var{\eps y}.
\end{align*}
Next we introduce two new variables:~$\var{\eta_{x, y}^{\mathrm{low}}}$ and~$\var{\eta_{x, y}^{\mathrm{high}}}$, together with constraints enforcing
\begin{align*}
\var{\eta_{x, y}^{\mathrm{low}}} &\leq g(\var{t_1}) + 2 g(\var{t_2}) + g(\var{t_3}) - f(\var{t_4}),\\
\var{\eta_{x, y}^{\mathrm{high}}} &\geq f(\var{t_1}) + 2 f(\var{t_2}) + f(\var{t_3}) - g(\var{t_4}),\\
\var{\eta_{x, y}^{\mathrm{low}}} &\geq -2\var{\lesssim 2\eps^3},\\
\var{\eta_{x, y}^{\mathrm{high}}} &\leq 2\var{\lesssim 2\eps^3}.
\end{align*}
Note that the variables $\var{\eta_{x, y}^{\mathrm{low}}}$ and $\var{\eta_{x, y}^{\mathrm{high}}}$ are not completely necessary, and that it is also possible to use direct constraints
\begin{align*}
 -2\var{\lesssim 2\eps^3} &\leq g(\var{t_1}) + 2 g(\var{t_2}) + g(\var{t_3}) - f(\var{t_4}),\\
2\var{\lesssim 2\eps^3} &\geq f(\var{t_1}) + 2 f(\var{t_2}) + f(\var{t_3}) - g(\var{t_4})
\end{align*}
instead. The two variables $\var{\eta_{x, y}^{\mathrm{low}}}$ and $\var{\eta_{x, y}^{\mathrm{high}}}$ are included since these slightly simplify the notation when proving correctness of this construction later on. This completes the construction of~$\Psi$, which can be performed in linear time. 
We finish by choosing $\delta = T(n)$ with $n$ being the number of variables in $\Psi$.
This completes the \CCIEXPL{} instance.

\paragraph{Calculations.}
Let $\eps$ be any real number such that there exist reals $\delta_i$ for $0 \leq i \leq L$ satisfying $\delta_0 = \delta$, $\delta_L = \eps$ and $\frac{1}{2} f(\delta_{i-1}) \leq \delta_i \leq g(\delta_{i-1})$ for all $1 \leq i \leq L$. Now we have the following facts (these facts hold in particular if $\eps = \var{\eps}(Q)$ for some $Q \in V(\Psi)$):
\begin{align}
|f(x) - x^2| &\leq \frac{1}{10}|x|^3 & \text{for } x \in [-\delta, \delta], \label[ineq]{fact:fg_f_approx_square}\\
|g(x) - x^2| &\leq \frac{1}{10}|x|^3 & \text{for } x \in [-\delta, \delta], \label[ineq]{fact:fg_g_approx_square}\\
\frac{1}{2} f(x) &\leq g(x) & \text{for } x \in [-\delta, \delta], \label[ineq]{fact:fg_halff}\\
\eps &\leq \frac{1}{100} \min\left(2^{-2^M}, \delta\right),& \label[ineq]{fact:fg_eps_bound}\\
f(\eps) &< \eps,  \quad g(\eps) < \eps, & \label[ineq]{fact:fg_fgeps_bound}\\
\var{\lesssim 2\eps^3}(Q) &\leq 3\var{\eps}(Q)^3 & \text{for } Q \in V(\Psi), \label[ineq]{fact:fg_2eps_cubed_var}\\
g(\eps + g(\eps)) - f(\eps) & \geq \eps^3. & \label[ineq]{fact:fg_leq_eps_cubed}
\end{align}
\Cref{fact:fg_f_approx_square,fact:fg_g_approx_square} are assumptions from the statement of the lemma. \Cref{fact:fg_halff} follows from this, together with the fact that $|x|$ is bounded by~$\frac{1}{4}$:
\[
\frac{1}{2} f(x) \leq \frac{1}{2}x^2 + \frac{1}{20}|x|^3 \leq x^2 - \frac{1}{10}|x|^3 \leq g(x).
\]
\Cref{fact:fg_eps_bound} can be derived in the same way as \cref{fact:eps_bound} from \cref{lemma:delta-f-ETR}, and now \cref{fact:fg_fgeps_bound} follows from this with \cref{fact:fg_f_approx_square,fact:fg_g_approx_square}.

In order to derive \cref{fact:fg_2eps_cubed_var}, we use the definition of the variable $\var{\lesssim 2\eps^3}$ and apply \cref{fact:fg_f_approx_square,fact:fg_g_approx_square} to this (here we take $\eps = \var{\eps}(Q)$ to simplify the notation a bit):
\begin{align*}
\var{\lesssim 2\eps^3}(Q) &\leq g(\eps + \var{\leq g(\eps)}(Q)) - f(\eps)\\
&\leq (\eps + \var{\leq g(\eps)}(Q))^2 + \frac{1}{10}(\eps + \var{\leq g(\eps)}(Q))^3 - \eps^2 + \frac{1}{10}\eps^3.
\end{align*}
Combining this with the constraint $\var{\leq g(\eps)} \leq g(\var{\eps})$ and \cref{fact:fg_g_approx_square,fact:fg_fgeps_bound}, we get
\begin{align*}
\var{\lesssim 2\eps^3}(Q) &\leq (\eps + \var{\leq g(\eps)}(Q))^2 + \frac{1}{10}(\eps + \var{\leq g(\eps)}(Q))^3 - \eps^2 + \frac{1}{10}\eps^3\\
&\leq (\eps + g(\eps))^2 + \frac{1}{10}(\eps + g(\eps))^3 - \eps^2 + \frac{1}{10}\eps^3\\
&\leq \left(\eps + \eps^2 + \frac{1}{10}\eps^3\right)^2 + \frac{1}{10}\left(\eps + \eps\right)^3 - \eps^2 + \frac{1}{10}\eps^3\\
&= \frac{29}{10}\eps^3 + \frac{6}{5}\eps^4 + \frac{1}{5}\eps^5 + \frac{1}{100}\eps^6.
\end{align*}
Combining this with $\eps \leq \frac{1}{100}$ (which follows from \cref{fact:fg_eps_bound}) yields that $\var{\lesssim 2\eps^3}(Q) \leq 3\eps^3$, as we wanted.

Finally \cref{fact:fg_leq_eps_cubed} follows from \cref{fact:fg_f_approx_square,fact:fg_g_approx_square,fact:fg_fgeps_bound} in the following manner:
\begin{align*}
g(\eps + g(\eps)) - f(\eps) &
\geq (\eps + g(\eps))^2 - \frac{1}{10}(\eps + g(\eps))^3 - \eps^2 - \frac{1}{10}\eps^3\\
&\geq \left(\eps + \eps^2 - \frac{1}{10} \eps^3\right)^2 - \frac{1}{10}(\eps + \eps)^3 - \eps^2 - \frac{1}{10}\eps^3\\
&= \frac{11}{10}\eps^3 + \frac{4}{5}\eps^4 - \frac{1}{5}\eps^5 + \frac{1}{100}\eps^6\\
&\geq \eps^3.
\end{align*}

\paragraph{$V(\Phi)$ nonempty implies $V(\Psi)$ nonempty.}
Suppose that $V(\Phi) \neq \emptyset$, and therefore, there exists some $P \in V(\Phi)$.
Our goal is to demonstrate that $V(\Psi) \neq \emptyset$. To achieve this, we construct a point~$Q$ that lies within $V(\Psi)$.
We start by taking $\var{\delta_0}(Q) = \delta$ and $\var{\delta_i}(Q) = g(\var{\delta_{i-1}})$ for all $1 \leq i \leq L$. 
By \cref{fact:fg_halff}, this definition satisfies all constraints on the $\var{\delta_i}$. Denote $\eps = \var{\eps}(Q) = \var{\delta_L}(Q)$. 
Next we take $\var{\leq g(\eps)}(Q) = g(\eps)$ and $\var{\lesssim 2\eps^3}(Q) = g(\eps + g(\eps)) - f(\eps)$, so by \cref{fact:fg_leq_eps_cubed} we know that $\var{\lesssim 2\eps^3}(Q) \geq \eps^3$.

For all variables~$x$ of~$\Phi$, we take $\var{\eps x}(Q) = \eps x(P)$. Since $V(P) \subseteq [-1, 1]^n$, it follows that all inequalities of the form $-\var{\eps}(Q) \leq \var{\eps x}(Q) \leq \var{\eps}(Q)$ are satisfied in this way. Also for every constraint from $\Phi$ of one of the forms $x + y = z$, $x \geq 0$ or $x = 1$, the corresponding constraint in~$\Psi$ is clearly satisfied.

Next we consider a squaring constraint $y = x^2$ from $\Phi$, for each such constraint we take
\begin{align*}
\var{t_1}(Q) &= \eps,\\
\var{t_2}(Q) &= \var{\eps x}(Q), \\
\var{t_3}(Q) &= \var{\eps y}(Q),\\
\var{t_4}(Q) &= \eps + \var{\eps y}(Q),\\
\var{\eta_{x, y}^{\mathrm{low}}}(Q) &= g(\var{t_1}(Q)) + 2g(\var{t_2}(Q)) + g(\var{t_3}(Q)) - f(\var{t_4}(Q)),\\
\var{\eta_{x, y}^{\mathrm{high}}}(Q) &= f(\var{t_1}(Q)) + 2f(\var{t_2}(Q)) + f(\var{t_3}(Q)) - g(\var{t_4}(Q)).
\end{align*}
Using these definitions, the only constraints for which we still need to check whether $Q$ satisfies them, are the constraints of the form $\var{\eta_{x, y}^{\mathrm{low}}} \geq -2\var{\lesssim 2\eps^3}$ and $\var{\eta_{x, y}^{\mathrm{high}}} \leq 2\var{\lesssim 2\eps^3}$.

We start by checking the first of these constraints. Denote $t_1 = \var{t_1}(Q)$, $t_2 = \var{t_2}(Q)$, $t_3 = \var{t_3}(Q)$ and $t_4 = \var{t_4}(Q)$. Now we can apply \cref{fact:fg_f_approx_square,fact:fg_g_approx_square} to the definition of $\var{\eta_{x, y}^{\mathrm{low}}}$:
\begin{align*}
\var{\eta_{x, y}^{\mathrm{low}}}(Q) &= g(t_1) + 2g(t_2) + g(t_3) - f(t_4)\\
&\geq t_1^2 - \frac{1}{10}|t_1|^3 + 2t_2^2 - \frac{1}{5}|t_2|^3 + t_3^2 - \frac{1}{10}|t_3|^3 - t_4^2 - \frac{1}{10}|t_4|^3\\
&= t_1^2 + 2t_2^2 + t_3^2 - t_4^2 - \left(\frac{1}{10}|t_1|^3 + \frac{1}{5}|t_2|^3 + \frac{1}{10}|t_3|^3 + \frac{1}{10}|t_4|^3\right).
\end{align*}
Since $t_1, \ldots t_4$ were chosen such that $t_1^2 + 2t_2^2 + t_3^2 - t_4^2 = \eps^2x(P)^2 - \eps^2y(P)$, and by the fact that $y(P) = x(P)^2$, it follows that $t_1^2 + 2t_2^2 + t_3^2 - t_4^2 = 0$. Furthermore, $t_1$, $t_2$ and $t_3$ are all bounded by $\eps$ in absolute value, while~$|t_4|$ is bounded by~$2\eps$. This yields
\begin{align*}
\var{\eta_{x, y}^{\mathrm{low}}}(Q) &\geq t_1^2 + 2t_2^2 + t_3^2 - t_4^2 - \left(\frac{1}{10}|t_1|^3 + \frac{1}{5}|t_2|^3 + \frac{1}{10}|t_3|^3 + \frac{1}{10}|t_4|^3\right)\\
&\geq 0 - \left(\frac{1}{10}\eps^3 + \frac{1}{5}\eps^3 + \frac{1}{10}\eps^3 + \frac{8}{10}\eps^3\right)\\
&\geq -2\eps^3.
\end{align*}
Combining this with $\var{\lesssim 2\eps^3}(Q) \geq \eps^3$, we find $\var{\eta_{x, y}^{\mathrm{low}}}(Q) \geq -2\var{\lesssim 2\eps^3}(Q)$.

Next we consider the constraint $\var{\eta_{x, y}^{\mathrm{high}}} \leq 2\var{\lesssim 2\eps^3}$. We apply \cref{fact:fg_f_approx_square,fact:fg_g_approx_square} to the definition of~$\var{\eta_{x, y}^{\mathrm{high}}}$:
\begin{align*}
\var{\eta_{x, y}^{\mathrm{high}}}(Q) &= f(t_1) + 2f(t_2) + f(t_3) - g(t_4)\\
&\leq t_1^2 + \frac{1}{10}|t_1|^3 + 2t_2^2 + \frac{1}{5}|t_2|^3 + t_3^2 + \frac{1}{10}|t_3|^3 - t_4^2 + \frac{1}{10}|t_4|^3\\
&= t_1^2 + 2t_2^2 + t_3^2 - t_4^2 + \frac{1}{10}|t_1|^3 + \frac{1}{5}|t_2|^3 + \frac{1}{10}|t_3|^3 + \frac{1}{10}|t_4|^3.
\end{align*}
Here we can apply that $t_1^2 + 2t_2^2 + t_3^2 - t_4^2 = 0$, and that $|t_1|$, $|t_2|$ and $|t_3|$ are bounded by $\eps$ and $|t_4|$ is bounded by $2\eps$ to get
\begin{align*}
\var{\eta_{x, y}^{\mathrm{high}}}(Q) &\leq t_1^2 + 2t_2^2 + t_3^2 - t_4^2 + \frac{1}{10}|t_1|^3 + \frac{1}{5}|t_2|^3 + \frac{1}{10}|t_3|^3 + \frac{1}{10}|t_4|^3\\
&\leq 0 + \frac{1}{10}\eps^3 + \frac{1}{5}\eps^3 + \frac{1}{10}\eps^3 + \frac{8}{10}\eps^3\\
&\leq 2\eps^3.
\end{align*}
So we get that $\var{\eta_{x, y}^{\mathrm{high}}}(Q) \leq 2\var{\lesssim 2\eps^3}(Q)$.

We conclude that $Q$ satisfies all constraints from $\Psi$, and therefore $Q \in V(\Psi)$. This proves that $V(\Psi) \neq \emptyset$.

\paragraph{$V(\Psi)$ nonempty implies $V(\Phi)$ nonempty.}
Next, let $Q \in V(\Psi)$. Just as in the proof of \cref{lemma:delta-f-ETR}, we want to show that $|x(Q)| \leq \delta$ for all variables~$x$ of~$\Psi$, and we want to prove that~$V(\Phi) \neq \emptyset$. Bounding the coordinates and showing that $Q$ is \domainadherent goes in exactly the same way as in \cref{lemma:delta-f-ETR}. 

In order to demonstrate that $V(\Phi) \neq \emptyset$, we once more apply \cref{lemma:approx}. We construct a point $P$ in $V(\Phi_{100\eps})$, where we again denote $\eps = \var{\eps}(Q)$.
This would imply $V(\Phi) \neq \emptyset$ since $100 \eps \leq 2^{-2^M}$. We take $x(P) = \frac{\var{\eps x}(Q)}{\eps}$ for all variables~$x$ of~$\Phi$. 
Now~$P$ satisfies all linear constraints and inequality constraints of~$\Phi$, and it only remains to be checked that it satisfies the constraints $|y - x^2| \leq 100\eps$ of~$\Phi_{100\eps}$.

We start by proving that $x(P)^2 - y(P) \leq 100\eps$. Denote $t_1 = \var{t_1}(Q) = \eps$, $t_2 = \var{t_2}(Q) = \eps x(P)$, $t_3 = \var{t_3}(Q) = \eps y(P)$ and $t_4 = \var{t_4}(Q) = \eps + \eps y(P)$. We have that
\[
t_1^2 + 2t_2^2 + t_3^2 - t_4^2 = 2\eps^2(x(P)^2-y(P)).
\]
Note that from \cref{fact:fg_f_approx_square,fact:fg_g_approx_square} it also follows that $x^2 \leq f(x) + \frac{1}{10}|x|^3$ and $x^2 \geq g(x) - \frac{1}{10}|x|^3$ for all $x \in [-\delta, \delta]$. Using this, we find
\begin{align*}
2\eps^2(x(P)^2-y(P)) &= t_1^2 + 2t_2^2 + t_3^2 - t_4^2\\
&\leq f(t_1) + \frac{1}{10}|t_1|^3 + 2f(t_2) + \frac{1}{5}|t_2|^3\\
&\quad + f(t_3) + \frac{1}{10}|t_3|^3 - g(t_4) + \frac{1}{10}|t_4|^3\\
&= f(t_1) + 2f(t_2) + f(t_3) - g(t_4)\\
&\quad + \frac{1}{10}|t_1|^3 + \frac{1}{5}|t_2|^3 + \frac{1}{10}|t_3|^3 + \frac{1}{10}|t_4|^3.
\end{align*}
To bound this, we use the variable $\var{\eta_{x, y}^{\mathrm{high}}}$, and the observation that $t_1$, $t_2$ and $t_3$ are bounded in absolute value by $\eps$, and $|t_4|$ is bounded by $2\eps$:
\begin{align*}
2\eps^2(x(P)^2-y(P)) &\leq f(t_1) + 2f(t_2) + f(t_3) - g(t_4)\\
&\quad + \frac{1}{10}|t_1|^3 + \frac{1}{5}|t_2|^3 + \frac{1}{10}|t_3|^3 + \frac{1}{10}|t_4|^3\\
&\leq \var{\eta_{x, y}^{\mathrm{high}}}(Q) + \frac{1}{10}\eps^3 + \frac{1}{5}\eps^3 + \frac{1}{10}\eps^3 + \frac{8}{10}\eps^3\\
&\leq 2\var{\lesssim 2\eps^3}(Q) + 2\eps^3.
\end{align*}
Here we can apply \cref{fact:fg_2eps_cubed_var} to find
\[
2\eps^2(x(P)^2-y(P)) \leq 8\eps^3 < 200\eps^3.
\]
This implies $x(P)^2 - y(P) \leq 100\eps$, as we wanted.

The proof that $x(P)^2 - y(P) \geq -100\eps$ works in a similar manner. Leaving out some intermediate steps, it looks as follows:
\begin{align*}
2\eps^2(x(P)^2-y(P)) &= t_1^2 + 2t_2^2 + t_3^2 - t_4^2\\
&\geq g(t_1) + 2g(t_2) + g(t_3) - f(t_4)\\
&\quad - \left(\frac{1}{10}|t_1|^3 + \frac{1}{5}|t_2|^3 + \frac{1}{10}|t_3|^3 + \frac{1}{10}|t_4|^3\right)\\
&\geq \var{\eta_{x, y}^{\mathrm{low}}}(Q) - 2\eps^3\\
&\geq -8\eps^3 > -200\eps^3,
\end{align*}
and therefore $x(P)^2 - y(P) \geq -100\eps$. This implies that $P \in V(\Phi_{100\eps})$, and therefore $V(\Phi) \neq \emptyset$.

This completes the proof of the validity of the reduction of \etrsquare to \CCIEXPL{}. So for~$f$ and~$g$ satisfying the conditions from the lemma, the problem \CCIEXPL{} is \ER-hard.
\end{proof}

Now that hardness of this restricted version of \CCIEXPL{} is proven, this result can be generalized in small steps until finally \cref{thm:Inequality} is proven.

Before we do this, we first note that in any \CCIEXPL{} formula, constraints of the form $x = q \cdot y$, where $x, y$ are variables and $q$ is a rational constant, can be enforced using a constant number of addition constraints and new variables. To illustrate this, we will discuss the case where $q \in [0,1]$ here. Other cases can be handled in a similar manner. Assume that $q = a/b$ for a positive integer $b$ and an integer $0 \leq a \leq b$. Now we can introduce variables $\var{\frac{i}{b}y}$ for all~$0 \leq i \leq b$, which satisfy constraints
\begin{align*}
\var{\frac{0}{b}y} &= \var{\frac{0}{b}y} + \var{\frac{0}{b}y},&\\
\var{\frac{i+1}{b}y} &= \var{\frac{i}{b}y} + \var{\frac{1}{b}y} & \text{for } 0 \leq i < b,\\
\var{\frac{b}{b}y} &= \var{\frac{0}{b}y} + y,&\\
\var{\frac{a}{b}y} &= \var{\frac{0}{b}y} + x.&
\end{align*}
This exactly enforces that $x = \frac{a}{b} \cdot y$.

The first step in working from \cref{lemma:CCIEXPL_fg_approx_sqr} to \cref{thm:Inequality} is to get rid of the constraint that $T(n)$ has to be bounded by $\frac{1}{4}$.

\begin{lemma}\label{lemma:CCIEXPL_T_no_bound}
Let $\domain \subseteq \R$ be a neighborhood of 0, and let $f, g \colon \domain \to \R$ be functions such that $|f(x) - x^2| \leq \frac{1}{10}|x|^3$ and $|g(x) - x^2| \leq \frac{1}{10}|x|^3$ for all $x \in \domain$. Let $T$ be bounded and \nicelyComputable. In this setting, \CCIEXPL{} is \ER-hard, even when considering only instances where $\delta = T(n)$, with~$n$ being the number of variables.
\end{lemma}

\begin{proof}
Let $c_1$ be some rational constant such that $0 < c_1 < \frac{1}{4}$ and $[-c_1, c_1] \subseteq U$. We will give a self-reduction from instances with $\delta = T^*(n)$ to instances with $\delta = T(n)$, where $T^*$ is some \nicelyComputable{} function bounded by $c_1$ which is yet to be determined.

Let $c_2$ be a rational constant that strictly bounds $T$ from above, and denote $c = \frac{c_2}{c_1}$. 
Without loss of generality we may assume that $c_2 \geq \frac{1}{4}$, so in particular $c \geq 1$. Let $(\delta, \Phi)$ be some instance of \CCIEXPL{} with $n$ variables, and $\delta = T^*(n)$. We will build an \equisatisfiable{} instance $(\delta', \Psi)$ with $m$ variables and fix $T^*$, such that $\delta' = c\delta = T(m)$. For this, we add to $\Phi$ extra variables $\var{\delta}$ and $\var{\delta'}$, together with constraints and auxiliary variables enforcing $\var{\delta'} = \delta'$ and $\var{\delta'} = c\var{\delta}$. Next we replace every constraint of the form $x = \delta$ by a constraint of the form $x = \var{\delta}$. This gives us the formula $\Psi$. Note that the solutions to $\Phi$ directly correspond to solutions of $\Psi$. The fact that $\delta' \geq \delta$ and the promises on $\Psi$ imply that all solutions of $\Psi$ are contained in $[-\delta', \delta']$. Furthermore, the fact that $\Phi$ is \domainadherent implies that also $\Psi$ is \domainadherent.

Note that the number of variables in $\Psi$ is exactly $n$ plus some constant $d$ which only depends on the function $T$. Therefore we can take $T^*(n) = \frac{1}{c}T(n+d)$ and $\delta' = T(n+d)$, this ensures that in the preceding construction we have $\delta' = c \cdot \delta$, and therefore $\Phi$ and $\Psi$ are indeed \equisatisfiable. Furthermore, since $T$ is bounded from above by $c_2$, it follows that $T^*$ is bounded from above by $c_1$. In particular, it follows that $[-T^*(n), T^*(n)] \subseteq U$ for all values of $n$.
\end{proof}

Now the next step when working towards \cref{thm:Inequality}, is to slightly relax the constraints on $f$ and $g$, by allowing the difference with squaring to be any $O(x^3)$ function, instead of just functions bounded by $\frac{1}{10}|x|^3$ in absolute value.

\begin{lemma}
Let $\domain \subseteq \R$ be a neighborhood of 0, and let $f, g \colon \domain \to \R$ be functions such that $f(x) = x^2 + O(x^3)$ and $g(x) = x^2 + O(x^3)$ as $x \to 0$. 
Let $T$ be bounded and \nicelyComputable. In this setting, \CCIEXPL{} is \ER-hard, even when only considering instances where $\delta = T(n)$, with $n$ being the number of variables.
\end{lemma}
\begin{proof}
Let $c$ be a constant such that $|f(x) - x^2| \leq c|x|^3$ and $|g(x) - x^2| \leq c|x|^3$ for all $x \in \domain^{*}$ where $\domain^{*} \subseteq \domain$ is a neighborhood of $0$. Now let $N$ be a positive integer larger than $10c$. This implies that for all $x \in \domain^{*}$
\begin{align*}
|N^2f(x/N) - x^2| &\leq \frac{1}{10}|x|^3 \text{ and}\\
|N^2g(x/N) - x^2| &\leq \frac{1}{10}|x|^3.
\end{align*}
If we define $f^*$ and $g^*$ on the domain $\domain^*$ by $f^*(x) = N^2f(x/N)$ and $g^*(x) = N^2g(x/N)$, then using \cref{lemma:CCIEXPL_T_no_bound}, we get that the problem \CCIEXPL{} is \ER-hard for $f^*$ and $g^*$. For the rest of the proof of this lemma, we will denote this specific \CCIEXPL{} version 
using $f^*$ and $g^*$ by \CCIEXPLstar.

We give a reduction from \CCIEXPLstar to the \CCIEXPL version with $f$ and $g$. Let $(\delta, \Phi)$ be a \CCIEXPLstar instance. Now for every variable $x$ in this instance, we add extra variables $\var{x/N}$, $\var{x/N^2}$ and we add constraints enforcing
\begin{align*}
\var{\frac{x}{N}} &= \frac{\var{x}}{N},\\
\var{\frac{x}{N^2}} &= \frac{\var{x}}{N^2}.
\end{align*}
Next we replace every constraint of the form $y \geq f^*(x)$ (i.e.~\largerzero{y-f^*(x)}) by a constraint $\var{y/N^2} \geq f(\var{x/N})$ (i.e.~\largerzero{\var{y/N^2} - f(\var{x/N})}). Similarly, we replace every constraint of the form $y \leq g^*(x)$ (i.e.~\largerzero{g^*(x)-y}) by a constraint $\var{y/N^2} \leq g(\var{x/N})$ (i.e.~\largerzero{g(\var{x/N}) - \var{y/N^2}}).

This results in a \CCIEXPL{} formula $\Psi$. Note that any solution of $\Phi$ is \domainadherent and contained in $[-\delta, \delta]^n$, and that the domain $\domain^*$ of $f^*$ and $g^*$ is a subset of $U$. From this, it follows that any solution of $\Phi$ corresponds to a \domainadherent solution of $\Psi$ where all variables are in $[-\delta, \delta]$. For the other direction, note that because of the way in which the \largerzeronoarg{} constraint was defined, any solution of $\Psi$ also corresponds to a solution of $\Phi$. So this proves that $\Phi$ and $\Psi$ are \equisatisfiable and that all solutions of $\Psi$ satisfy the necessary conditions.

Note that the number of variables $m$ in the new \CCIEXPL{} instance depends in a linear manner on the number of variables $n$ in $\Phi$. If we also want to ensure that $\delta = T(m)$, then we can define $T^*(n)$ to be a bounded \nicelyComputable{} function such that $T(m) = T^*(n)$ for every possible value of $n$. Then we can decide to only consider \CCIEXPLstar{} instances with $\delta = T^*(n)$ in the described construction.
\end{proof}

In the next lemma, we allow even more possible $f$ and $g$.

\begin{lemma} \label{lemma:almost_poly_ineq}
Let $\domain \subseteq \R$ be a neighborhood of $0$, and let $f, g \colon \domain \to \R$ be functions such that $f(x) = ax + bx^2 + O(x^3)$ and $g(x) = cx + dx^2 + O(x^3)$ as $x \to 0$, where $a, b, c, d \in \Q$ and $b, d > 0$. 
Let $T$ be bounded and \nicelyComputable{}. In this setting, \CCIEXPL{} is \ER-hard, even when only considering only instances where $\delta = T(n)$, with $n$ being the number of variables.
\end{lemma}
\begin{proof}
We define $f^*$ and $g^*$ as $f^*(x) = (f(x)-ax)/b$ and $g^*(x) = (g(x)-cx)/d$. From the constraints on~$f$ and~$g$ it follows that $f^*(x) = x^2 + O(x^3)$ and $g^*(x) = x^2 + O(x^3)$. Therefore we can apply the previous lemma to these functions to find that the \CCIEXPL problem with~$f^*$ and~$g^*$ is \ER-hard. We will denote this problem by \CCIEXPLstar, and give a reduction from \CCIEXPLstar to the \CCIEXPL problem with~$f$ and~$g$ as defined in the lemma statement.

Let $(\delta, \Phi)$ be any instance of \CCIEXPLstar. We denote
\[
\delta' = (1+|a|+|b|+|c|+|d|)\delta.
\]
Now we build an instance $(\delta', \Psi)$ of \CCIEXPL{} in the following manner:
We start by adding a variable $\var{\delta}$ which is meant as a replacement for the $\delta$ in conditions of the form $x = \delta$ in $\Phi$. To introduce this variable, we introduce an auxiliary variable $\var{\delta'}$ and enforce the following constraints:
\begin{align*}
\var{\delta'} &= \delta',\\
\var{\delta'} &= (1+|a|+|b|+|c|+|d|)\var{\delta}.
\end{align*}
We also add every variable of $\Phi$ and all constraints of the form $x + y = z$ or $x \geq 0$ from $\Phi$ to $\Psi$, and for every constraint $x = \delta$ in $\Phi$ we add a constraint $x = \var{\delta}$ to $\Psi$.

For every constraint $y \geq f^*(x)$ of $\Phi$ we introduce a new variable $\var{ax+by}$ to $\Psi$ which we force to equal $ax + by$ using some linear constraints. Furthermore we add a constraint $\var{ax+by} \geq f(x)$. For constraints of the form $y \leq g^*(x)$ we do something similar.

In this way, $(\delta', \Psi)$ is a valid \CCIEXPL{} instance since all values of the variables in any solution can be seen to be bounded by $\delta'$ using the triangle inequality. Furthermore, the new instance $\Psi$ differs from $\Phi$ only by new auxiliary variables and otherwise has exactly the same constraints on the original variables. Thus $V(\Psi)$ is non-empty if and only if $V(\Phi)$ is non-empty, and $\Psi$ is \domainadherent since $\Phi$ is.

Finally, we will show that we might impose $\delta' = T(n)$ on the instances. Note that the number of variables created by the reduction is linear in the number of old variables and the number of constraints of the form $y \geq f^*(x)$ or $y \leq g^*(x)$ in $\Phi$. If $\Phi$ has $n$ variables, then there can be at most $O(n^2)$ different constraints of one of these two forms, so we can find some integer constant $k$ such that $\Psi$ has at most $k\cdot n^2$ variables. Now we can adjust the previous reduction to add sufficiently many extra variables (not occurring in any constraint) to make sure there are exactly $k \cdot n^2$ variables. We can also define $T^*$ as $T^*(n) = \frac{T\left(kn^2\right)}{1 + |a| + |b| + |c| + |d|}$, now restricting the inputs of \CCIEXPLstar{} to cases with $\delta = T^*(n)$ gives the desired result.
\end{proof}

The next step is to notice that any function which is three times differentiable with a nonzero second derivative satisfies the constraints from the previous lemma. 
This leads to the following result:

\CCIEXPLdifferentiableCOR*
\begin{proof}
Using Taylor's theorem, we find that
\begin{align*}
f(x) &= f'(0)x + \frac{f''(0)}{2}x^2 + O(x^3) \text{ and}\\
g(x) &= g'(0)x + \frac{g''(0)}{2}x^2 + O(x^3).
\end{align*}
To this, we can apply the previous lemma to find that \CCIEXPL{} is \ER-hard, even when we only consider instances with $\delta = T(n)$.
\end{proof}

%%%%%%%%%%%%%%%%%%%%%%%%%%%%%%%%%%%%%%%%%%
\subsection{Implicit Constraints}
\label{sec:Implicit}
%%%%%%%%%%%%%%%%%%%%%%%%%%%%%%%%%%%%%%%%%%
Using, \Cref{cor:fgETRINEQEXPL_differentiable} we will show \Cref{thm:Inequality} and \Cref{thm:Equality} in this order.
\Cref{lemma:first_impl} is almost equivalent to \cref{thm:Inequality}. 
The only difference is that the conditions $f_y(0, 0) > 0$ and $g_y(0, 0) < 0$ are added.
As a preparation, we need the Implicit function theorem. We state the exact version that we use here for the convenience of the reader.
\begin{theorem}[Implicit Function Theorem]
    Let $\domain \subseteq \R^2$ be a neighborhood of $(0,0)$.
    Let $f \colon \domain \to \R$ be a $C^3$-function with $f_y(0,0) \neq 0$ defining the set $S = \set{(x,y) \in \domain}{f(x,y) = 0}$.
    Now there is a neighborhood $\domain' \subseteq \R$ of $0$ with $(\domain')^2 \subseteq \domain$ and some $C^3$-function
    $f_{\expl} : \domain' \to \R$ such that
    $\set{(x,y) \in (\domain')^2}{y = f_{\expl}(x)} = S \cap (\domain')^2$. Furthermore, we have  $f_{\expl}'(x) = -\frac{f_x(x,f_{\expl}(x))}{f_y(x,f_{\expl}(x))}$ for all $x \in \domain'$.
\end{theorem}

\begin{lemma} \label{lemma:first_impl}
Let $\domain \subseteq \R^2$ be a neighborhood of the origin.
Let $f, g \colon \domain \to \R$ be two functions, with $f$ \wellbehaved and \convex, and $g$ \wellbehaved and \concave. Furthermore assume that their partial derivatives satisfy $f_y(0, 0) > 0$ and $g_y(0, 0) < 0$.
Let $T$ be bounded and \nicelyComputable. Then the problem \CCI{} is \ER-hard, even when considering only instances where $\delta = T(n)$, with $n$ being the number of variables.
\end{lemma}
\begin{proof}
Using the implicit function theorem, we find that in a neighborhood $(\domain')^2 \subseteq \domain$ of $(0, 0)$, the curve $f(x, y) = 0$ can also be given in an explicit form $y = f_{\expl}(x)$, where $f_{\expl}$ is some $C^3$-function $\domain' \to \R$. So for $(x, y) \in (\domain')^2$ we have $f(x, y) = 0$ if and only if $y = f_{\expl}(x)$. Since $f_y(0, 0) > 0$, it also follows that $f(x, y) \geq 0$ if and only if $y \geq f_{\expl}(x)$. Furthermore, the implicit function theorem also states that the derivative of $f_{\expl}$ is given by
\[
f_{\expl}'(x) = -\frac{f_x(x, f_{\expl}(x))}{f_y(x, f_{\expl}(x))}.
\]
From this, it can be computed that the second derivative in 0 is
\[
f_{\expl}''(0) = -\left(\frac{f_y^2f_{xx} - 2f_xf_yf_{xy} + f_x^2f_{yy}}{f_y^3}\right)(0, 0).
\]
Note that the fact that $f$ is \convex exactly implies that the numerator of this expression is a positive number. Using the assumptions from the lemma statement, we conclude that $f_\expl''(0)$ is a positive rational number.

In an analogous way, we can write the condition $g(x, y) \geq 0$ in the form $y \leq g_{\expl}(x)$ in some neighborhood of $(0, 0)$, where $g_{\expl}$ is a $C^3$-function with rational first and second derivative in $0$, and with positive second derivative. Without loss of generality we assume that~$f_{\expl}$ and $g_{\expl}$ have the same domain $U'$.

We would now like to conclude that the problem \CCI{} is equivalent to the problem \CCIEXPL{} for $f_\expl$ and $g_\expl$, but we need to be slightly careful because of the exact way in which the constraints involving $f$ and $g$ are defined. Recall that with the constraint $f(x, y) \geq 0$ in an \CCI{} instance we actually mean $\largerzero{f}$, which was defined to be satisfied whenever $(x, y)$ falls outside of $\domain$. Similarly the constraint $y \geq f_\expl(x)$ in an \CCIEXPL{} instance is satisfied whenever $x \not\in \domain'$.

We will give a reduction from \CCIEXPL{} to \CCI{}. Let $\delta_0$ be some constant such that $[-\delta_0, \delta_0] \subseteq \domain'$. 
Let $0 < c \leq \frac{1}{2}$ be a rational constant such that $cT(n) \leq \delta_0$ for all $n$. This exists, since $T$ is bounded. Now let $(\Phi, \delta')$ be a \CCIEXPL{} instance with $\delta' = T^*(n)$. Here $T^*$ will be determined later in such a way that $T^*(n) = cT(m)$, where $m$ is the number of variables in the \CCI instance $(\Psi, \delta)$ which we will construct now.
We start constructing $\Psi$ by adding variables $\var{\delta'}$ and $\var{\delta}$, with linear constraints enforcing $\var{\delta} = \delta$ and $\var{\delta'} = c\var{\delta}$. Next we add all linear constraints from $\Phi$ to $\Psi$, except that we replace constraints of the form $x = \delta'$ by $x = \var{\delta'}$. For every constraint $y \geq f_\expl(x)$ in $\Phi$, we add constraints to $\Psi$ enforcing $f(x,y) \geq 0$, $y \geq -\delta'$ and $y \leq \delta'$. Similarly, we replace $y \leq g_\expl(x)$ by constraints enforcing $g(x,y) \geq 0$, $y \geq -\delta'$ and $y \leq \delta'$.

Since $\delta' \leq \delta_0$, it follows that the constraints $\largerzero{y-f_\expl(x)}$ and $\largerzero{f}$ exactly coincide when restricted to $[-\delta', \delta']^2$. Since all solutions to $\Phi$ are promised to have all coordinates in $[-\delta', \delta']$, it follows that every solution to $\Phi$ gives rise to a solution of $\Psi$. From the definition of the \largerzeronoarg constraint it also follows that $\largerzero{f}$ is satisfied whenever both $\largerzero{y-f_\expl(x)}$ and $y \in U$. From this it follows that also every solution to $\Psi$ corresponds to a solution of $\Phi$. Since all solutions of $\Psi$ correspond to solutions of $\Phi$, it also follows that all solutions of $\Psi$ are contained in $[-\delta, \delta]^m$ and are \domainadherent.

To make the correct choice of $T^*$, we proceed as in the end of the proof of \cref{lemma:almost_poly_ineq}. That is, we add some extra variables to $\Psi$ to ensure that the total number of variables is always exactly $m = kn^2$ for some constant $k$, and choose $T^*(n) = cT(kn^2)$ for all $n$.
\end{proof}

From here it is a small step to prove the main result:

\renewcommand{\theorembstring}{(Restated)}
\CCITHM*
\begin{proof}
Without loss of generality, we may assume that $f_y(0, 0) \neq 0$ and $g_y(0, 0) \neq 0$. In any other case, we can just interchange the variables in one of the functions.

In the case where $f_y(0, 0) > 0$ and $g_y(0, 0) < 0$, we can apply the previous lemma and we are done. For the case $f_y(0, 0) < 0$ and $g_y(0, 0) < 0$, we can provide a reduction from \CCI with functions $f^*(x, y) = f(-x, -y)$ and $g^*(x, y) = g(x, y)$. For the case $f_y(0, 0) > 0$ and $g_y(0, 0) > 0$ we can make a reduction from \CCI with functions $f^*(x, y) = f(x, y)$ and $g^*(x, y) = g(-x, -y)$. and for the case $f_y(0, 0) < 0$ and $g_y(0, 0) > 0$, we can provide a reduction from \CCI with functions $f^*(x, y) = f(-x, -y)$ and $g^*(x, y) = g(-x, -y)$. Note that flipping the signs of the inputs of $f$ or $g$ does not influence any second partial derivative, while it does negate the first partial derivatives; therefore the mentioned starting points for the reductions can all be seen to satisfy the conditions from \cref{lemma:first_impl}.

As an example we discuss the case $f_y(0, 0) > 0$ and $g_y(0, 0) > 0$. We want to give reduction from the problem \CCI with functions $f^*(x, y) = f(x, y)$ and $g^*(x, y) = g(-x, -y)$; we denote this \CCI variation by \CCIstar. So suppose that $(\delta, \Phi)$ is a \CCIstar instance. Now we will construct a \CCI instance $(\delta, \Psi)$ (with the~$f$ and~$g$ from the theorem statement). We add every variable of $\Phi$ to $\Psi$, and for every such variable $x$ we also add an extra variable $\var{-x}$, together with a constraint enforcing $x + \var{-x} = 0$. 
Furthermore, we copy every constraint from $\Phi$ to $\Psi$, except for constraints of the form $g^*(x, y) \geq 0$. These constraints are replaced instead by~$g(\var{-x}, \var{-y}) \geq 0$. 
This finishes the construction.

If we also want to enforce that $\delta = T(m)$ with $m$ the number of variables in the new \CCI{} instance, then we can define $T^*$ as $T^*(n) = T(2n)$ and only consider \CCIstar{} instances with~$\delta = T^*(n)$. Note that the reduction always doubles the number of variables, and therefore this implies that also $\delta = T(m)$ with $m$ the number of variables in $\Psi$.
\end{proof}

As a final result in this section, we prove \cref{thm:Equality} as well. To do this, we start from \cref{cor:fgETRINEQEXPL_differentiable} and convert this to a result about \CEEXPL{}.

\begin{lemma}
Let $\domain \subseteq \R$ be a neighborhood of $0$, and let $f\colon \domain \to \R$ be a function which is three times differentiable such that $f(0) = 0$ and $f'(0), f''(0) \in \Q$ with $f''(0) \neq 0$. 
Let $T$ be bounded and \nicelyComputable. In this setting, \CEEXPL{} is \ER-hard, even when considering only instances where $\delta = T(n)$, with $n$ being the number of variables.
\end{lemma}
\begin{proof}
We apply \cref{cor:fgETRINEQEXPL_differentiable} to the case where $g = f$ to find that in this case \CCIEXPL is \ER-hard. We can reduce this problem to \CEEXPL{}. Let $(\delta, \Phi)$ be a \CCIEXPL instance. Now we construct an \CEEXPL{} formula $\Psi$. We copy all constraints of the form $x+y = z$, $x \geq 0$ and $x = \delta$ from $\Phi$. For every constraint $y \geq f(x)$ we introduce two new variables $\var{f(x)}$ and $\var{y-f(x)}$, which we restrict by constraints
\begin{align*}
\var{f(x)} &= f(x),\\
y &= \var{y-f(x)} + \var{f(x)}, \text{ and}\\
\var{y-f(x)} &\geq 0.
\end{align*}
In a similar manner we replace every constraint $y \leq f(x)$ by introducing new variables $\var{f(x)}$ and $\var{f(x)-y}$ and imposing the constraints
\begin{align*}
\var{f(x)} &= f(x),\\
\var{f(x)} &= \var{f(x)-y} + y, \text{ and}\\
\var{f(x)-y} &\geq 0.
\end{align*}
This completes the construction. It can easily be checked that every solution of $\Phi$ corresponds to a solution of $\Psi$, and vice versa.

In order to enforce that $\delta  = T(m)$ with $m$ the number of variables $\Psi$, we use a technique similar to that used in \cref{lemma:almost_poly_ineq}. Again we add extra variables to ensure that the number of variables in $\Psi$ is exactly $kn^2$ for some constant $k$, where $n$ is the number of variables in~$\Phi$. Then we take $T^*(n) = T(kn^2)$ and only consider instances $(\delta, \Phi)$ of \CCIEXPL{} which satisfy $\delta = T^*(n)$.
\end{proof}

\renewcommand{\theoremastring}{(Restated)}
\fETRTHM*
\begin{proof}
This proof is very similar to that of \cref{lemma:first_impl}. Without loss of generality, we may assume that $f_y(0, 0) \neq 0$, otherwise we can swap the variables. Using the implicit function theorem, we can write the condition $f(x, y) = 0$ in some neighborhood $(\domain')^2 \subseteq \domain$ of $(0, 0)$ as $y = f_\expl(x)$, where $f_\expl$ is some $C^3$-function $\domain' \to \R$. Using the fact that the curvature of $f$ is nonzero, the implicit function theorem also tells us that $f_\expl''(0) \neq 0$.

Now we give a reduction from \CEEXPL to \CE. Choose $\delta_0$ such that $[-\delta_0, \delta_0] \subseteq U'$, and let $c \leq \frac{1}{2}$ be a positive rational number such that $cT(n) \leq \delta_0$ for all $n$. Let $(\Phi, \delta')$ be a \CEEXPL instance with $\delta' = T^*(n)$ to be determined later. We will build an instance $(\Psi, \delta)$, where $\delta = T(m)$ with $m$ the final number of variables in $\Psi$. First we add variables $\var{\delta}, \var{\delta'}$ to $\Psi$ with constraints enforcing $\var{\delta'} = c\var{\delta}$. We replace all constraints $x = \delta'$ in $\Phi$ by $x = \var{\delta'}$, and we transfer all other linear equalities and inequalities from $\Phi$ to $\Psi$ as well. Finally, we replace every constraint in $\Phi$ of the form $y = f_\expl(x)$ by constraints enforcing $y = f(x)$, $y \geq -\delta'$ and $y \leq \delta'$.

For the same reasons as in the proof of \cref{lemma:first_impl}, the solutions of $\Phi$ are in correspondence with those of $\Psi$. Furthermore, the solutions of $\Psi$ will again all be \domainadherent and contained in $[-\delta, \delta]^m$ with $m$ the number of variables in $\Psi$. Finally, we may add some extra variables to $\Psi$ to ensure it has exactly $kn^2$ variables for some constant $k$, and then choose~$T^*(n) = cT(kn^2)$ for all $n$.
\end{proof}

%%%%%%%%%%%%%%%%%%%%%%%%%%%%%%%%%%%%%%%%%%%%
%%%%%%%%%%%%%%%%%%%%%%%%%%%%%%%%%%%%%%%%%%%%
\printbibliography

@inproceedings{schaefer1978complexity,
  title={The complexity of satisfiability problems},
  author={Schaefer, Thomas J.},
  booktitle = {Proceedings of the Tenth Annual ACM Symposium on Theory of Computing (STOC 1978)},
  pages={216--226},
  year={1978},
  publisher = {Association for Computing Machinery},
  address = {New York, NY, USA},
  url = {https://doi.org/10.1145/800133.804350},
  doi = {10.1145/800133.804350},
}

@article{bulatov2006dichotomy,
  title={A dichotomy theorem for constraint satisfaction problems on a 3-element set},
  author={Bulatov, Andrei A.},
  journal={J. ACM},
  volume={53},
  number={1},
  pages={66--120},
  year={2006},
  publisher={ACM New York, NY, USA},
  url = {https://doi.org/10.1145/1120582.1120584},
  doi = {10.1145/1120582.1120584},
}

@inproceedings{bulatov2017dichotomy,
  title={A dichotomy theorem for nonuniform {CSP}s},
  author={Bulatov, Andrei A},
  booktitle={2017 IEEE 58th Annual Symposium on Foundations of Computer Science (FOCS)},
  pages={319--330},
  year={2017},
  organization={IEEE},
  doi={10.1109/FOCS.2017.37}
}

@article{zhuk2020proof,
  title={A {P}roof of the {CSP} {D}ichotomy {C}onjecture},
  author={Zhuk, Dmitriy},
  journal={J. ACM},
  volume={67},
  number={5},
  pages={1--78},
  year={2020},
  url = {https://doi.org/10.1145/3402029},
  doi = {10.1145/3402029},
}

@article{marx2005parameterized,
  title={Parameterized complexity of constraint satisfaction problems},
  author={Marx, D{\'a}niel},
  journal={Computational Complexity},
  note  = {Also in CCC 2004},
  volume={14},
  number={2},
  pages={153--183},
  year={2005},
  publisher={Springer},
  doi={10.1007/s00037-005-0195-9}
}

@article{dyer2010approximation,
  title={An approximation trichotomy for {B}oolean \#{CSP}},
  author={Dyer, Martin and Goldberg, Leslie Ann and Jerrum, Mark},
  journal={Journal of Computer and System Sciences},
  volume={76},
  number={3-4},
  pages={267--277},
  year={2010},
  publisher={Elsevier},
  doi = {https://doi.org/10.1016/j.jcss.2009.08.003},
  url = {https://www.sciencedirect.com/science/article/pii/S0022000009000762},
}

@article{viola2020combined,
  title={The combined basic {LP} and affine {IP} relaxation for promise {VCSP}s on infinite domains},
  author = {Viola, Caterina and \v{Z}ivn\'{y}, Stanislav},
  year = {2021},
  issue_date = {July 2021},
  publisher = {Association for Computing Machinery},
  address = {New York, NY, USA},
  volume = {17},
  number = {3},
  issn = {1549-6325},
  url = {https://doi.org/10.1145/3458041},
  doi = {10.1145/3458041},
  journal = {ACM Trans. Algorithms},
  month = jul,
  articleno = {21},
  numpages = {23}
}

@inproceedings{bodirsky2008non,
  title={Non-dichotomies in constraint satisfaction complexity},
  author={Bodirsky, Manuel and Grohe, Martin},
  booktitle={Automata, Languages and Programming (ICALP 2008)},
  pages={184--196},
  year={2008},
  organization={Springer},
  doi={10.1007/978-3-540-70583-3_16}
}

@article{bodirsky2010complexity,
  title={The complexity of temporal constraint satisfaction problems},
  author={Bodirsky, Manuel and K{\'a}ra, Jan},
  journal={J. ACM},
  volume={57},
  number={2},
  pages={1--41},
  year={2010},
  publisher = {Association for Computing Machinery},
  address = {New York, NY, USA},
  url = {https://doi.org/10.1145/1667053.1667058},
  doi = {10.1145/1667053.1667058},
}

@inbook{zhuk2021complete,
author = {Dmitriy Zhuk and Barnaby Martin and Michał Wrona},
title = {The complete classification for quantified equality constraints},
booktitle = {Proceedings of the 2023 Annual ACM-SIAM Symposium on Discrete Algorithms (SODA)},
chapter = {},
pages = {2746-2760},
doi = {10.1137/1.9781611977554.ch103},
URL = {https://epubs.siam.org/doi/abs/10.1137/1.9781611977554.ch103},
eprint = {https://epubs.siam.org/doi/pdf/10.1137/1.9781611977554.ch103},
}

@InCollection{bodirsky2017constraint,
  title={Constraint satisfaction problems over numeric domains},
  author={Bodirsky, Manuel and Mamino, Marcello},
  booktitle =	{The Constraint Satisfaction Problem: Complexity and Approximability},
  pages =	{79--111},
  series =	{Dagstuhl Follow-Ups},
  ISBN =	{978-3-95977-003-3},
  ISSN =	{1868-8977},
  year =	{2017},
  volume =	{7},
  editor =	{Krokhin, Andrei and Zivny, Stanislav},
  publisher =	{Schloss Dagstuhl -- Leibniz-Zentrum f{\"u}r Informatik},
  address =	{Dagstuhl, Germany},
  URL =		{https://drops-dev.dagstuhl.de/entities/document/10.4230/DFU.Vol7.15301.79},
  URN =		{urn:nbn:de:0030-drops-69580},
  doi =		{10.4230/DFU.Vol7.15301.79}
}

@article{bodirsky2012essential,
  title={Essential convexity and complexity of semi-algebraic constraints},
  author={Bodirsky, Manuel and Jonsson, Peter and Von Oertzen, Timo},
  URL = {https://lmcs.episciences.org/1218},
  DOI = {10.2168/LMCS-8(4:5)2012},
  JOURNAL = {{Logical Methods in Computer Science}},
  VOLUME = {{Volume 8, Issue 4}},
  YEAR = {2012},
  MONTH = Oct,
}

@article{jonsson2016constraint,
  title={Constraint satisfaction and semilinear expansions of addition over the rationals and the reals},
  author={Jonsson, Peter and Thapper, Johan},
  journal={Journal of Computer and System Sciences},
  volume={82},
  number={5},
  pages={912--928},
  year={2016},
  publisher={Elsevier},
  doi = {https://doi.org/10.1016/j.jcss.2016.03.002},
  url = {https://www.sciencedirect.com/science/article/pii/S0022000016000271},
}

@article{TensorRank,
  author    = {Marcus Schaefer and
               Daniel \v{S}tefankovi\v{c}},
  title     = {The Complexity of Tensor Rank},
  journal   = {Theory of Computing Systems},
  volume    = {62},
  number    = {5},
  pages     = {1161--1174},
  year      = {2018},
  url       = {https://doi.org/10.1007/s00224-017-9800-y},
  doi       = {10.1007/s00224-017-9800-y},
  bibsource = {dblp computer science bibliography, https://dblp.org}
}

@article{gensane2005improved,
  title={Improved dense packings of congruent squares in a square},
  author={Gensane, Thierry and Ryckelynck, Philippe},
  journal={DCG},
  volume={34},
  number={1},
  pages={97--109},
  year={2005},
  publisher={Springer},
  doi={10.1007/s00454-004-1129-z}
}

@article{schaefer2021complexity,
  title={Complexity of Geometric k-Planarity for Fixed k},
  author={Schaefer, Marcus},
  journal = {Journal of Graph Algorithms and Applications},
  volume={25},
  number={1},
  pages={29--41},
  year={2021},
  doi = {10.7155/jgaa.00548}
}

@article{NestedPolytopesER,
  author    = {Michael Gene Dobbins and
               Andreas Holmsen and
               Tillmann Miltzow},
  title     = {A Universality Theorem for Nested Polytopes},
  journal   = {arXiv:1908.02213},
  year      = {2019},
  archivePrefix = {arXiv},
  doi = {10.48550/arXiv.1908.02213}
}

@inproceedings{Schaefer2010,
title={Complexity of some geometric and topological problems},
  author={Schaefer, Marcus},
  booktitle={GD 2009},
  pages={334--344},
  year={2009},
  series={Lecture Notes in Computer Science (LNCS)},
  organization={Springer},
  doi={10.1007/978-3-642-11805-0_32}
}

@incollection{schaefer2013realizability,
  title={Realizability of graphs and linkages},
  author={Schaefer, Marcus},
  booktitle={Thirty Essays on Geometric Graph Theory},
  editor={J\'{a}nos Pach},
  chapter={23},
  pages={461--482},
  year={2013},
  doi={10.1007/978-1-4614-0110-0_24}
}

@article{richter1995realization,
  title={Realization spaces of 4-polytopes are universal},
  author={Richter-Gebert, J{\"u}rgen and Ziegler, G{\"u}nter M.},
  journal={Bulletin of the American Mathematical Society},
  volume={32},
  number={4},
  pages={403--412},
  year={1995},
  doi={10.1090/S0273-0979-1995-00604-X}
}

@inproceedings{mnev1988universality,
  title={The universality theorems on the classification problem of configuration varieties and convex polytopes varieties},
  author={Mn{\"e}v, Nicolai},
  booktitle={Topology and geometry -- Rohlin seminar},
  pages={527--543},
  year={1988},
  doi={10.1007/BFb0082792}
}

@inproceedings{etrPacking,
  author    = {Mikkel Abrahamsen and
               Tillmann Miltzow and
               Nadja Seiferth},
  title     = {Framework for {$\exists\mathbb{R}$}-Completeness of Two-Dimensional Packing Problems},
  booktitle = {2020 IEEE 61st Annual Symposium on Foundations of Computer Science (FOCS)},
  pages     = {1014--1021},
  note      = {arXiv:2004.07558},
  publisher = {{IEEE}},
  year      = {2020},
  url       = {https://doi.org/10.1109/FOCS46700.2020.00098},
  doi       = {10.1109/FOCS46700.2020.00098},
  bibsource = {dblp computer science bibliography, https://dblp.org}
}

@article{mcdiarmid2013integer,
  title={Integer realizations of disk and segment graphs},
  author={McDiarmid, Colin and M{\"u}ller, Tobias},
  journal={Journal of Combinatorial Theory, Series B},
  volume={103},
  number={1},
  pages={114--143},
  year={2013},
  publisher={Elsevier},
  doi = {https://doi.org/10.1016/j.jctb.2012.09.004},
  url = {https://www.sciencedirect.com/science/article/pii/S0095895612000718},
}

@article{AnnaPreparation,
  author    = {Anna Lubiw and 
  				Tillmann Miltzow and
               Debajyoti Mondal},
  title     = {The Complexity of Drawing a Graph in a Polygonal Region},
  journal = {Journal of Graph Algorithms and Applications},
   year = {2022},
   volume = {26},
   number = {4},
   pages = {421--446},
   doi = {10.7155/jgaa.00602}
}

@InProceedings{AreasKleist,
author = {Michael G. Dobbins and Linda Kleist and Tillmann Miltzow and Pawe\l{} Rz{\k{a}}{\.{z}}ewski
	},
editor="Brandst{\"a}dt, Andreas
and K{\"o}hler, Ekkehard
and Meer, Klaus",
title={$\forall \exists \mathbb{R}$-completeness and area-universality},
booktitle="Graph-Theoretic Concepts in Computer Science (WG 2018)",
year="2018",
publisher="Springer International Publishing",
address="Cham",
pages="164--175",
doi={10.1007/978-3-030-00256-5_14}
}

@article{shitov2016universality,
  title={A universality theorem for nonnegative matrix factorizations},
  author={Shitov, Yaroslav},
    year={2016},
  journal={arXiv:1606.09068},
  doi = {10.48550/arXiv.1606.09068}
}

@article{Shitov16a,
  author    = {Yaroslav Shitov},
  title     = {The Complexity of Positive Semidefinite Matrix Factorization},
  journal   = {{SIAM} Journal on Optimization},
  volume    = {27},
  number    = {3},
  pages     = {1898--1909},
  year      = {2017},
  doi = {10.1137/16M1080616},
 }

@book{boyd2004convex,
  title={Convex Optimization},
  author={Boyd, Stephen and Boyd, Stephen P. and Vandenberghe, Lieven},
  year={2004},
  publisher={Cambridge University Press},
  doi={10.1017/CBO9780511804441}
}

@article{tarasov2008semidefinite,
  title={Semidefinite programming and arithmetic circuit evaluation},
  author={Tarasov, Sergey P. and Vyalyi, Mikhail N.},
  journal={Discrete Applied Mathematics},
  volume={156},
  number={11},
  pages={2070--2078},
  year={2008},
  note = {In Memory of Leonid Khachiyan (1952 - 2005 )},
  issn = {0166-218X},
  doi = {https://doi.org/10.1016/j.dam.2007.04.023},
  url = {https://www.sciencedirect.com/science/article/pii/S0166218X07001370},
}

@article{allender2009complexity,
  title={On the complexity of numerical analysis},
  author={Allender, Eric and B{\"u}rgisser, Peter and Kjeldgaard-Pedersen, Johan and Miltersen, Peter Bro},
  journal={SIAM Journal on Computing},
  volume={38},
  number={5},
  pages={1987--2006},
  year={2009},
  publisher={SIAM},
  doi = {10.1137/070697926},
}

@article{Argyrios2022_AppriximatinETR,
    author = {Deligkas, Argyrios
        and Fearnley, John
        and Melissourgos, Themistoklis
        and Spirakis, Paul G.},
    title = {{Approximating the Existential Theory of the Reals}},
    journal = {Journal of Computer and System Sciences},
    year = {2022},
    volume = {125},
    pages = {106--128},
    doi = {10.1016/j.jcss.2021.11.002}
}

@inproceedings{o2017sos,
  title={{SOS} is not obviously automatizable, even approximately},
  author={O'Donnell, Ryan},
  booktitle =	{8th Innovations in Theoretical Computer Science Conference (ITCS 2017)},
  pages =	{59:1--59:10},
  series =	{Leibniz International Proceedings in Informatics (LIPIcs)},
  ISBN =	{978-3-95977-029-3},
  ISSN =	{1868-8969},
  year =	{2017},
  volume =	{67},
  editor =	{Papadimitriou, Christos H.},
  publisher =	{Schloss Dagstuhl -- Leibniz-Zentrum f{\"u}r Informatik},
  address =	{Dagstuhl, Germany},
  URL =		{https://drops-dev.dagstuhl.de/entities/document/10.4230/LIPIcs.ITCS.2017.59},
  URN =		{urn:nbn:de:0030-drops-81980},
  doi =		{10.4230/LIPIcs.ITCS.2017.59}
}

@article{kratochvil1994intersection,
  title={Intersection graphs of segments},
  author={Kratochv{\'{i}}l, Jan and Matou{\v{s}}ek, Ji{\v{r}}{\'{i}}},
  journal={Journal of Combinatorial Theory, Series B},
  volume={62},
  number={2},
  pages={289--315},
  year={1994},
  publisher={Elsevier},
  doi = {https://doi.org/10.1006/jctb.1994.1071},
  url = {https://www.sciencedirect.com/science/article/pii/S0095895684710719},
}

@article{matousek2014intersection,
  author    = {Matou{\v{s}}ek, Ji{\v{r}}{\'{i}}},
  title     = {Intersection graphs of segments and $\exists \mathbb{R}$},
  journal   ={arXiv:1406.2636},
    year      = {2014},
doi = {10.48550/arXiv.1406.2636}
}

@INPROCEEDINGS{Z92,
  author={Zhang, Xiao-Dong },
  booktitle={Workshop on Physics and Computation}, 
  title={Complexity Of Neural Network Learning In The Real Number Model}, 
  year={1992},
  volume={},
  number={},
  pages={146-150},
  doi={10.1109/PHYCMP.1992.615511},
}

@inproceedings{abrahamsen2021training,
  author    = {Mikkel Abrahamsen and
               Linda Kleist and
               Tillmann Miltzow},
  title     = {Training Neural Networks is {ER}-complete},
  booktitle = {Advances in Neural Information Processing Systems (NeurIPS 2021)},
  note      = {Also in arXiv:2102.09798},
  pages     = {18293--18306},
  year      = {2021},
  url       = {https://proceedings.neurips.cc/paper/2021/hash/9813b270ed0288e7c0388f0fd4ec68f5-Abstract.html},
  bibsource = {dblp computer science bibliography, https://dblp.org}
}

@article{SmoothingGap,
  author    = {Jeff Erickson and
               Ivor {van der Hoog} and
               Tillmann Miltzow},
  title     = {Smoothing the gap between {NP} and {ER}},
  journal = {SIAM Journal on Computing},
  volume = {0},
  number = {0},
  pages = {FOCS20-102-FOCS20-138},
  year = {2022},
  doi = {10.1137/20M1385287},
  note      = {Also in FOCS 2020},
}

@article{erickson2019optimal,
  title={Optimal Curve Straightening is 
  $\exists\mathbb{R}$-Complete},
  author={Erickson, Jeff},
  journal={arXiv:1908.09400},
  year={2019},
  doi={10.48550/arXiv.1908.09400}
}

@article{cardinal2017recognition,
  title={Recognition and complexity of point visibility graphs},
  author={Cardinal, Jean and Hoffmann, Udo},
  journal={DCG},
  volume={57},
  number={1},
  pages={164--178},
  year={2017},
  publisher={Springer},
  doi={10.1007/s00454-016-9831-1}
}

@article{cardinal2017intersection,
  title={Intersection graphs of rays and grounded segments},
  author={Cardinal, Jean and Felsner, Stefan and Miltzow, Tillmann and Tompkins, Casey and Vogtenhuber, Birgit},
  journal = {Journal of Graph Algorithms and Applications},
  year = {2018},
  volume = {22},
  number = {2},
  pages = {273--295},
  doi = {10.7155/jgaa.00470},
  note = {Also in WG 2017},
}

@inproceedings{kang2011sphere,
  title={Sphere and dot product representations of graphs},
  author={Kang, Ross and M{\"u}ller, Tobias},
  booktitle={Proceedings of the Twenty-Seventh Annual Symposium on Computational Geometry (SoCG 2011)},
  pages={308--314},
  year={2011},
  organization={ACM},
  publisher = {Association for Computing Machinery},
  address = {New York, NY, USA},
  url = {https://doi.org/10.1145/1998196.1998249},
  doi = {10.1145/1998196.1998249},
}

@incollection{shor1991stretchability,
  title={Stretchability of pseudolines is {NP}-hard},
  author={Shor, Peter},
  booktitle={Applied Geometry and Discrete Mathematics},
  pages={531--554},
  year={1991},
  doi={10.1090/dimacs/004}
}

@InProceedings{chistikov_et_al:LIPIcs:2016:6238,
  author =	{Dmitry Chistikov and Stefan Kiefer and Ines Marusic and Mahsa Shirmohammadi and James Worrell},
  title =	{{On Restricted Nonnegative Matrix Factorization}},
  booktitle =	{
    International Colloquium on Automata, Languages and Programming (ICALP 2016)},
  series    = {LIPIcs},
  pages =	{103:1--103:14},
  year =	{2016},
  note = {arXiv:1605.06848},
  volume =	{55},
  doi={10.4230/LIPIcs.ICALP.2016.103}
}

@article{FabianMapLabel,
  title={Map Labelling is $\exists\mathbb{R}$-complete},
  author={Fabian Klute and Tillmann Miltzow},
  journal={in preparation},
  year={2022}
}

@article{abrahamsen2019dynamic,
  author    = {Mikkel Abrahamsen and
               Tillmann Miltzow},
  title     = {Dynamic Toolbox for {ETRINV}},
  journal   = {arXiv:1912.08674},
  year      = {2019},
  url       = {http://arxiv.org/abs/1912.08674},
  eprinttype = {arXiv},
  eprint    = {1912.08674},
  bibsource = {dblp computer science bibliography, https://dblp.org},
  doi = {10.48550/arXiv.1912.08674}
}

@inproceedings{abrahamsen2021covering,
  author    = {Mikkel Abrahamsen},
  title     = {Covering Polygons is Even Harder},
  booktitle = {2021 IEEE 62nd Annual Symposium on Foundations of Computer Science (FOCS)},
  pages     = {375--386},
  publisher = {{IEEE}},
  year      = {2022},
  url       = {https://doi.org/10.1109/FOCS52979.2021.00045},
  doi       = {10.1109/FOCS52979.2021.00045},
  bibsource = {dblp computer science bibliography, https://dblp.org}
}

@article{ARTETR,
  author    = {Mikkel Abrahamsen and
               Anna Adamaszek and
               Tillmann Miltzow},
  title     = {The Art Gallery Problem is $\exists \mathbb{R}$-complete},
  journal   = {J. {ACM}},
  volume    = {69},
  number    = {1},
  pages     = {4:1--4:70},
  year      = {2022},
  url       = {https://doi.org/10.1145/3486220},
  doi       = {10.1145/3486220},
  note      = {Also in STOC 2018, 65--73},
  bibsource = {dblp computer science bibliography, https://dblp.org}
}

@book{basu2006algorithms,
  title={Algorithms in real algebraic geometry},
  author={Basu, Saugata and Pollack, Richard and Roy, Marie-Fran\c{c}oise},
   publisher={Springer, Berlin Heidelberg},
    year={2006},
ISBN = {978-3-540-33098-1},
    DOI = {10.1007/3-540-33099-2}
}

@inproceedings{canny1988some,
  title={Some algebraic and geometric computations in {PSPACE}},
  author={Canny, John},
  booktitle = {Proceedings of the Twentieth Annual ACM Symposium on Theory of Computing (STOC 1988)},
  pages={460--467},
  year={1988},
  publisher = {Association for Computing Machinery},
  address = {New York, NY, USA},
  url = {https://doi.org/10.1145/62212.62257},
  doi = {10.1145/62212.62257},
}

@article{berthelsen2019computational,
  title={On the computational complexity of decision problems about multi-player Nash equilibria},
  author={Berthelsen, Marie Louisa T{\o}lb{\o}ll and Hansen, Kristoffer Arnsfelt},
  journal={Theory of Computing Systems},
  volume={66},
  number={3},
  pages={519--545},
  year={2022},
  publisher={Springer},
  doi={10.1007/s00224-022-10080-1}
}

@article{bilo2017existential,
  title={$\exists\mathbb{R}$-complete Decision Problems about (Symmetric) {N}ash Equilibria in (Symmetric) Multi-player Games},
  author={Bil{\`o}, Vittorio and Mavronicolas, Marios},
  year = {2021},
  issue_date = {September 2021},
  publisher = {Association for Computing Machinery},
  address = {New York, NY, USA},
  volume = {9},
  number = {3},
  issn = {2167-8375},
  url = {https://doi.org/10.1145/3456758},
  doi = {10.1145/3456758},
  journal = {ACM Trans. Econ. Comput.},
}

@inproceedings{bilo2016catalog,
  title={A catalog of {EXISTS}-{R}-complete decision problems about {N}ash equilibria in multi-player games},
  author={Bil{\`o}, Vittorio and Mavronicolas, Marios},
  booktitle={33rd Symposium on Theoretical Aspects of Computer Science (STACS 2016)},
  pages =	{17:1--17:13},
  series =	{Leibniz International Proceedings in Informatics (LIPIcs)},
  ISBN =	{978-3-95977-001-9},
  ISSN =	{1868-8969},
  year =	{2016},
  volume =	{47},
  editor =	{Ollinger, Nicolas and Vollmer, Heribert},
  publisher =	{Schloss Dagstuhl -- Leibniz-Zentrum f{\"u}r Informatik},
  address =	{Dagstuhl, Germany},
  URL =		{https://drops-dev.dagstuhl.de/entities/document/10.4230/LIPIcs.STACS.2016.17},
  URN =		{urn:nbn:de:0030-drops-57189},
  doi =		{10.4230/LIPIcs.STACS.2016.17}
}

@article{jeronimo2010minimum,
  title={On the minimum of a positive polynomial over the standard simplex},
  author={Jeronimo, Gabriela and Perrucci, Daniel},
  journal={Journal of Symbolic Computation},
  volume={45},
  number={4},
  pages={434--442},
  year={2010},
  publisher={Elsevier},
  doi = {https://doi.org/10.1016/j.jsc.2010.01.001},
  url = {https://www.sciencedirect.com/science/article/pii/S074771711000009X},
}

@article{Schaefer-ETR,
  author    = {Marcus Schaefer and
               Daniel \v{S}tefankovi\v{c}},
  title     = {Fixed Points, {N}ash Equilibria, and the Existential Theory of the Reals},
  journal   = {Theory of Computing Systems},
  volume    = {60},
  number    = {2},
  pages     = {172--193},
  year      = {2017},
  bibsource = {dblp computer science bibliography, http://dblp.org},
  doi = {10.1007/s00224-015-9662-0}
}

@article{garg2015etr,
  title={{$\exists$R}-Completeness for Decision Versions of Multi-player (Symmetric) {N}ash Equilibria},
  author={Garg, Jugal and Mehta, Ruta and Vazirani, Vijay V. and Yazdanbod, Sadra},
  year = {2018},
issue_date = {February 2018},
publisher = {Association for Computing Machinery},
address = {New York, NY, USA},
volume = {6},
number = {1},
issn = {2167-8375},
url = {https://doi.org/10.1145/3175494},
doi = {10.1145/3175494},
journal = {ACM Trans. Econ. Comput.},
month = jan,
articleno = {1},
numpages = {23}
}

@article{pataki2021exponential,
  title={How do exponential size solutions arise in semidefinite programming?},
  author={Pataki, G{\'a}bor and Touzov, Aleksandr},
  journal = {SIAM Journal on Optimization},
  volume = {34},
  number = {1},
  pages = {977-1005},
  year = {2024},
  doi = {10.1137/21M1434945},
}

@article{marker1992additive,
  title={Additive reducts of real closed fields},
  author={Marker, David and Peterzil, Ya'acov and Pillay, Anand},
  journal={The Journal of Symbolic Logic},
  volume={57},
  number={1},
  pages={109--117},
  year={1992},
  publisher={JSTOR},
  DOI={10.2307/2275179}
}

@article{peterzil1992structure,
 ISSN = {00224812},
 URL = {http://www.jstor.org/stable/2275430},
 author = {Ya'acov Peterzil},
 journal = {The Journal of Symbolic Logic},
 number = {3},
 pages = {779--794},
 publisher = {Association for Symbolic Logic},
 title = {A Structure Theorem for Semibounded Sets in the Reals},
 volume = {57},
 year = {1992},
 DOI={10.2307/2275430}
}

@article{ArtJack,
  author    = {Stade, Jack},
  title     = {The Point-Boundary Art Gallery Problem is $\exists\mathbb{R}$-hard},
  journal   = {CoRR},
  volume    = {abs/2210.12817},
  year      = {2022},
  doi       = {10.48550/arXiv.2210.12817},
  eprinttype = {arXiv},
  eprint    = {2210.12817},
  bibsource = {dblp computer science bibliography, https://dblp.org},
}

@inproceedings{TrainFullNeuralNetworks,
  author       = {Daniel Bertschinger and
                  Christoph Hertrich and
                  Paul Jungeblut and
                  Tillmann Miltzow and
                  Simon Weber},
  title        = {Training Fully Connected Neural Networks is $\exists\mathbb{R}$-Complete},
  booktitle = {Advances in Neural Information Processing Systems (NeurIPS 2023)},
  editor = {A. Oh and T. Neumann and A. Globerson and K. Saenko and M. Hardt and S. Levine},
  pages = {36222--36237},
  publisher = {Curran Associates, Inc.},
  url = {https://proceedings.neurips.cc/paper_files/paper/2023/file/71c31ebf577ffdad5f4a74156daad518-Paper-Conference.pdf},
  volume = {36},
  year = {2023}
}

@article{peterzil1993reducts,
 ISSN = {00224812},
 URL = {http://www.jstor.org/stable/2275107},
 author = {Ya'acov Peterzil},
 journal = {The Journal of Symbolic Logic},
 number = {3},
 pages = {955--966},
 publisher = {Association for Symbolic Logic},
 title = {Reducts of Some Structures Over the Reals},
 volume = {58},
 year = {1993},
 DOI={10.2307/2275107}
}

@inproceedings{Simon,
  author    = {Simon Hengeveld and
               Tillmann Miltzow},
  title     = {A Practical Algorithm with Performance Guarantees for the Art Gallery
               Problem},
  booktitle =	{37th International Symposium on Computational Geometry (SoCG 2021)},
  pages =	{44:1--44:16},
  series =	{Leibniz International Proceedings in Informatics (LIPIcs)},
  ISBN =	{978-3-95977-184-9},
  ISSN =	{1868-8969},
  year =	{2021},
  volume =	{189},
  editor =	{Buchin, Kevin and Colin de Verdi\`{e}re, \'{E}ric},
  publisher =	{Schloss Dagstuhl -- Leibniz-Zentrum f{\"u}r Informatik},
  address =	{Dagstuhl, Germany},
  URL =		{https://drops.dagstuhl.de/entities/document/10.4230/LIPIcs.SoCG.2021.44},
  URN =		{urn:nbn:de:0030-drops-138433},
  doi =		{10.4230/LIPIcs.SoCG.2021.44}
}

@misc{Wikipedia2023Galois,
    author = {Wikipedia},
    title = {Galois theory --- {W}ikipedia{,} The Free Encyclopedia},
    year = {2023},
    url = {https://en.wikipedia.org/wiki/Galois_theory},
    note = {[Online; accessed 7-July-2023]},
}

@article{basu2010bounding,
title = "Bounding the radii of balls meeting every connected component of semi-algebraic sets",
journal = "Journal of Symbolic Computation",
volume = "45",
number = "12",
pages = "1270 - 1279",
year = "2010",
note = "MEGA’2009",
issn = "0747-7171",
doi = "https://doi.org/10.1016/j.jsc.2010.06.009",
url = "http://www.sciencedirect.com/science/article/pii/S0747717110000891",
author = "Saugata Basu and Marie-Françoise Roy"
}
%%%%%%%%%%%%%%%%%%%%%%%%%%%%%%%%%%%%%%%%%%%%
%%%%%%%%%%%%%%%%%%%%%%%%%%%%%%%%%%%%%%%%%%%%

\appendix
%%%%%%%%%%%%%%%%%%%%%%%%%%%%%%%%%%%%%%%%%%%%
\section{Circle-Constraint}
\label{app:Circle}
%%%%%%%%%%%%%%%%%%%%%%%%%%%%%%%%%%%%%%%%%%%%
Here, we discuss the question of expressing multiplication
via linear equations and the circle constraint $x^2+y^2 = 1$.
We note that for real numbers $x$ and $y$ the following equivalence holds:

There exists a real $z$ such that $z^2 + (x+y)^2  = 1$ and $(z+x-y)^2 + (z-x+y)^2 = 1$
if and only if $8xy = 1$ and $|x+y| \leq 1$.
This can be used in turn to express the inversion constraint ($x\cdot y = 1$) after some scaling
and imposing range constraints.
Note that $x\cdot y = 1$ can be used to express squaring as follows
\[\frac{1}{\frac{1}{x} - \frac{1}{x+1}} - x = x^2.\]
And we saw already in the \Cref{subsec:ER} how squaring can be used to express multiplication.

\end{document}